\documentclass[letterpaper,11pt]{article}

\usepackage[utf8]{inputenc}
\usepackage{pstricks,pstricks-add,pst-sigsys} % need XeLaTex to compile

\usepackage{amsthm}
\usepackage{amsmath,amsfonts,amssymb,latexsym,cite}

\usepackage{enumerate}
\usepackage{hyperref} % ref online resource
\usepackage{graphicx} % includes graphics
\graphicspath{{./figures/}} % location of graphs
\usepackage{caption}
\usepackage[margin=0.05in]{subcaption}
\usepackage{color}
\usepackage{enumitem}
\usepackage{epsfig}
\usepackage{graphicx}

\usepackage{latexsym}
\usepackage{amsmath}
\usepackage{amssymb}
\usepackage{soul}
\usepackage[toc,page]{appendix}

\newtheorem{theorem}{Theorem}[section]
\newtheorem{lemma}[theorem]{Lemma}
\newtheorem{claim}[theorem]{Claim}

\newtheorem{corollary}[theorem]{Corollary}
\newtheorem{definition}[theorem]{Definition}

\newcommand{\polylog}{\mbox{polylog}}

\newcommand{\ml}[1]{#1}
\newcommand{\sj}[1]{#1}
\newcommand{\zt}[1]{#1}
\newcommand{\zx}[1]{#1}
\newcommand{\alfun}[1][ ]{\alpha\left( p,\bar{p}_{#1}\right) }
\newcommand{\myfun}[1][ ]{\alpha\left( \myp,\bar{p}_{#1}\right) }
\newcommand{\capopt}{\min_{\bar{p}\in\left[ 0,p\right] }\left[ \alfun\left( 1-H\left( \frac{\bar{p}}{\alfun}\right) \right) \right]}
\newcommand{\myopt}{\min_{\bar{p}\in\left[ 0,\myp\right] }\left[ \myfun\left( 1-H\left( \frac{\bar{p}}{\myfun}\right) \right) \right]}
\newcommand{\mycapa}[1][0]{\myfun[#1]\left( 1-H\left( \frac{\bar{p}_{#1}}{\myfun[#1]}\right) \right) }

\newcommand{\leftcw}{t}
\newcommand{\rightcw}{h}

\newcommand{\pref}[1]{{#1}_{\leftcw}}
\newcommand{\suff}[1]{{#1}_{\rightcw}}

\newcommand{\phat}{\hat{p}}
\newcommand{\phatt}{\pref{\phat}}
\newcommand{\pbar}{\bar{p}}

\newcommand{\chunknum}{\frac{1}{\echunk}}
\newcommand{\chunk}{n\echunk}
\newcommand{\choiceschunk}{\left\lbrace \chunk,2\chunk,\cdots,n-\chunk\right\rbrace }
\newcommand{\setchend}{\mathcal{T}}

\newcommand{\intvquar}{\left( 0,\frac{1}{4}\right) }
\newcommand{\rangep}{\left[ 0, \frac{1}{2} \right) }
\newcommand{\intvone}{\left[ 0,1\right] }

\newcommand{\msg}{\mathcal{U}}
\newcommand{\secr}{\mathcal{S}}
\newcommand{\mlist}{\mathcal{L}}
\newcommand{\exlist}{\mathcal{L}{\left(m\right) }}
\newcommand{\clist}{{\exlist}}
\newcommand{\mlists}{{L}}
\newcommand{\clists}{{L{\left(m\right)}}}
\newcommand{\cwlist}{\left\lbrace \mathbf{x}_{1},\mathbf{x}_{2},\cdots,\mathbf{x}_{\clists}\right\rbrace }
\newcommand{\lecwlist}{\left\lbrace \mathbf{x}_{1},\mathbf{x}_{2},\cdots,\mathbf{x}_{L-1}\right\rbrace }

\newcommand{\echunk}{\theta}
\newcommand{\echval}{\frac{\erate^{2}(1-4\myp)}{4}}
\newcommand{\erate}{\epsilon}
\newcommand{\elist}{\frac{\erate}{2}}
\newcommand{\elfrac}{\frac{(n-t)\erate^2}{16}}
\newcommand{\efrac}[1][t]{\frac{\erate^{2}}{16\alpha^{2}\left( \myp,\bar{p}_{#1}\right)}}

\newcommand{\etwo}{\theta}
\newcommand{\eone}{\epsilon}
\newcommand{\ethree}{\etwo}
\newcommand{\efour}{\frac{\etwo}{2}}
\newcommand{\efive}{\etwo-\etwo^2}

\newcommand{\mcode}{\mathcal{C}}
\newcommand{\megacode}[1][k]{\mathcal{C}_{1}\circ\mathcal{C}_{2} \circ\cdots\circ \mathcal{C}_{#1}}
\newcommand{\megacodw}[1][k]{\mathcal{C}_{1}\left( m,s_1\right) \circ\mathcal{C}_{2}\left( m,s_2\right) \circ\cdots\circ \mathcal{C}_{#1}\left( m,s_{#1}\right) }
\newcommand{\rmegacode}[1][k]{\mathcal{C}_{#1+1}\circ\mathcal{C}_{#1+2}\circ\cdots\circ \mathcal{C}_{\chunknum}}
\newcommand{\rmegacodw}[1][k]{\mathcal{C}_{#1+1}\left( m,s_{#1+1}\right) \circ\mathcal{C}_{#1+2}\left( m,s_{#1+2}\right) \circ\cdots\circ \mathcal{C}_{\chunknum}\left( m,s_{\chunknum}\right) }
\newcommand{\myp}{p^{\prime}}
\newcommand{\ppm}{p} %{p^{\prime}}
\newcommand{\cpm}{\mathbf{x}^{\prime}}
\newcommand{\tstar}{t^{\star}}
\newcommand{\kstar}{k^{\star}}
\newcommand{\lstar}{l^{\star}}

\newcommand{\dl}{\Delta}
\newcommand{\keyspace}{[2^{nS}]}
\newcommand{\msgspace}{[2^{nR}]}
\newcommand{\binaryset}[1]{\left\lbrace 0,1\right\rbrace ^{#1}}
\newcommand{\binaryera}[1]{\left\lbrace 0,1,\Lambda\right\rbrace ^{#1}}
\newcommand{\elistord}{O\left( \frac{1}{\eone}\right) }
\newcommand{\blistord}{O\left( \frac{1}{\erate}\right) }

\theoremstyle{remark}
\newtheorem{remark}[theorem]{Remark}

\begin{document}
\title{A characterization of the capacity of online (causal) binary channels}
\author{
Z. Chen
\thanks{Department of Information Engineering,The Chinese University of Hong Kong, \texttt{cz012@ie.cuhk.edu.hk}}
\and 
S. Jaggi
\thanks{Department of Information Engineering, The Chinese University of Hong Kong, \texttt{jaggi@ie.cuhk.edu.hk}}
\and 
M. Langberg
\thanks{Department of Electrical Engineering, University at Buffalo, The State University of New York, \texttt{mikel@buffalo.edu}}
}

\date{}
\maketitle
\thispagestyle{empty}

\begin{abstract}
\noindent%
In the binary online (or ``causal'') channel coding model, a sender wishes to communicate a
message to a receiver by transmitting a codeword $\mathbf{x} =(x_1,\ldots,x_n) \in \{0,1\}^n$ bit by bit via a channel limited to
at most $pn$ corruptions. The channel is ``online'' in the sense that at the $i$th step of communication the channel decides whether to corrupt the $i$th bit or not based on its view so far, i.e., its decision depends only
on the transmitted bits $(x_1,\ldots,x_i)$. 
This is in contrast to the classical adversarial channel in which the error is chosen by a channel that has a full knowledge on the sent codeword $\mathbf{x}$. 

In this work we study the capacity of binary online channels for two corruption models: the {\em bit-flip} model in which the channel may flip at most $pn$ of the bits of the transmitted codeword, and the {\em erasure} model in which the channel may erase at most $pn$ bits of the transmitted codeword. Specifically, for both error models we give a full characterization of the capacity as a function of $p$. 

The online channel (in both the bit-flip and erasure case) has seen a number of recent studies which present both upper and lower bounds on its capacity. In this work, we present and analyze a coding scheme that improves on the previously suggested lower bounds and matches the previously suggested upper bounds thus implying a tight characterization.
%\sj{mention lookahead?}
\end{abstract}

\newpage
\setcounter{page}{1}

%\documentclass[11pt,letterpaper,titlepage]{article}
%\usepackage[margin=1.5in]{geometry}
%%\usepackage[top=1in, bottom=1in, left=1.5in, right=1.5in]{geometry}
%%\usepackage[top=1.5in, bottom=1.5in, left=1in, right=1in]{geometry}
%\usepackage{setspace}
%\singlespacing

%\newtheorem{theorem}{Theorem}
%\newtheorem{corollary}[theorem]{Corollary}
%\newtheorem{lemma}[theorem]{Lemma}
%\newtheorem{claim}[theorem]{Claim}
%
%\theoremstyle{definition}
%\newtheorem{definition}[theorem]{Definition}
%\newtheorem{example}{Example}
%
%\theoremstyle{remark}
%\newtheorem{remark}[theorem]{Remark}

%\begin{document}
%\title{Causal Bit-flip}
%

\section{Introduction}
Reliable communication over different types of channels has been extensively studied in electrical engineering and computer science. One frequently used communication channel model is the binary erasure channel, in which a bit (a zero or one) is either transmitted intact or erased. Specifically, an erased bit is a visible error, denoted by a special symbol $\Lambda$, which can be identified directly by a receiver. Another frequently studied channel model is the binary bit-flip channel, where bits can be flipped to their complement. 

There are two broad approaches to model (erasure or bit-flip) errors imposed by the binary channel. Shannon's approach is to model the channel as a stochastic process; Hamming's approach is a combinatorial approach to model the channel by an adversarial process that can erase or flip parts of the transmitted codeword arbitrarily, subject only to a limit on the number of corrupted bits. 
%In particular, in Hamming's model the adversarial behavior can depend arbitrarily on the codeword, resulting in arbitrary error patterns. Therefore, reliable communication over adversarial channels requires the corresponding communication schemes to be capable of correcting every error pattern subject to a limit on the number of bits than can be erased or flipped.

It is interesting to further classify the Hamming model for an adversarial binary channel in terms of the adversary's knowledge of the codeword.
Some examples include the standard adversarial channel (also referred to here as the {\em omniscient} adversary), e.g., \cite{gilbert1952comparison,gilbert1957comparison,mceliece1977new}, the {\em causal} (or {\em online}) adversary, e.g., \cite{dey2013codes,langberg2009binary,haviv2011beating,dey2012improved,bassily2014causal}, and the {\em oblivious} adversary, e.g.,\cite{lapidoth1998reliable,Langberg:08,guruswami2010codes}; from the strongest adversarial power to weakest. In one extreme, the omniscient adversarial model (a.k.a. the classical adversarial model) assumes that the channel has full knowledge of the entire codeword, and based on this knowledge, the channel can maliciously choose the error pattern. In the other extreme, the oblivious adversarial model is a model in which the channel is clueless about the codeword and generates errors in a manner that is independent of the codeword being transmitted. The causal adversarial model is an intermediate model between the two extremes, in which the channel decides whether to tamper with a particular bit of the codeword based only on the bits transmitted so far. 
{There are significant differences between the different adversarial models classified above (with respect to their capacity). We elaborate on these differences shortly.}
%In contrast, the look-ahead channel model strengthens and generalizes the causal model in the way that the adversary decides whether to tamper with bit $x_i$ based on his knowledge of $(x_1,x_2,\cdots,x_{i+d})$, where $d$ is the number of look-ahead bits.

%\sj{* Should cite GV codes, LP bound here for the omniscient (Cover and Thomas). 
%	
%	* Should mention status of problems. For example, mention that the full capacity characterization is a major open problem with no significant improvement on either inner bounds for XX years, or outer bounds for YY years. On the other hand, the capacity for oblivious and delayed adversaries were recently characterized. The goal of this paper is to characterize the causal and ``small lookahead’’ capacity, which has been the subject of several recent works.}
%
%
%
%\sj{Move the discussion of delayed adversaries to next para? Or remove entirely from intro? Otherwise obfuscates the fact that only the three primary models discussed in the first presentation.}
%
%
%\zt{Citaions of GV code and LP bound are added. status of problem is menioned in related work. Our goal is mentioned below. I think we can leave the ``lookahead'' to the conclusion part }

In this work we focus on causal adversaries, and study reliable communication over \textit{binary causal adversarial erasure channels} and \textit{binary causal adversarial bit-flip channels}. Specifically, we consider the following communication scenario. A sender (Alice) wishes to transmit a message $m\in\msg$ to a receiver (Bob) over a binary causal adversarial channel by encoding $m$ into a codeword $\mathbf{x}=(x_1,x_2,\cdots,x_n)\in\binaryset{n}$ of length $n$. However, the channel is governed by a causal adversary (Calvin), who can observe $\mathbf{x}$ and manipulate up to a $p$-fraction of the $n$ transmitted bits. More importantly, Calvin decides whether to tamper with the $i$-th bit of the codeword based only on the bits $(x_1,x_2,\cdots,x_i)$ transmitted thus far. Our goal is to find a coding scheme by which Alice can send as many distinct messages as possible while ensuring Bob succeeds in decoding w.h.p. Roughly, if $2^{nR}$ distinct messages can be sent using codewords of length $n$, we say that a code achieves rate $R$. We are interested in the maximum achievable rate $R$, which is the capacity $C_p$ of the channel. (See Section~\ref{sec:model} for precise definitions.)

\subsection{Our Results}
In this work we characterize the capacity of both the binary causal bit-flip channel and the binary causal erasure channel as a function of $p$ (the strength of the adversary). Specifically, we propose and analyze a novel coding scheme which implies a lower bound on the capacity for both channels that matches the known upper bounds~\cite{dey2012improved, bassily2014causal}  (to be described in detail shortly). 
Our main results can be summarized by the following two Theorems.

%\sj{Should be careful here — capacity characterization consist of both achievabilities and converses, but in this work we only present achievabilities — converses are in prior works. So need to present Theorem statements as such (explicitly citing Dey, Jaggi, Langberg, Sarwate, and Bassily-Smith in theorem statements?).}
%
%\sj{Instead of presenting the $p\geq 1/4$ region separately, how about combining into one equation, with the condition given on the RHS }

\begin{theorem}%[Causal bit-flip]
\label{the:main:flip}
Let $\alfun \triangleq 1-4(p-\pbar)$. 
The capacity of the binary causal adversary bit-flip channel is 
\begin{equation}
C^{flip}_{p}=\left \{
\begin{array}{lc}
\capopt, \mbox{ }& p\in[0,1/4] \\
0,\mbox{ } & p \geq 1/4
\end{array}
\right .
\end{equation}
\end{theorem}

\begin{theorem}%[Causal erasure]
\label{the:main:erase}
The capacity of the binary causal adversary erasure channel is 
\begin{equation}
C^{erase}_{p}=\left \{
\begin{array}{lc}
1-2p, \mbox{ }& p\in[0,1/2] \\
0,\mbox{ } & p \geq 1/2
\end{array}
\right .
\end{equation}
\end{theorem}

In fact, as direct by-products of the analysis of our coding scheme, we can show that even if Calvin has ``small'' lookahead, the capacity is essentially unchanged. More precisely, if for any constant $\epsilon > 0$, Calvin decides whether to tamper with the $i$-th bit of the codeword based only on the bits $(x_1,x_2,\cdots,x_{j})$, where $j = \min\{n,i+n\epsilon\}$, then the capacity of the corresponding ``$n\epsilon$-lookahead is at most $f(\epsilon)$ less than the corresponding $C^{flip}$ and $C^{erase}$ we show in Theorems~\ref{the:main:flip} and~\ref{the:main:erase} above (for some continuous $f$). We provide a rough argument in support of this claim in the Remark at the end of Section~\ref{sec:code-analysis}.

%Finally, the techniques that we develop in this work also allow us to characterize the $q$-ary causal ``flip'' and erase adversary problems. We merely point out the ways in which the techniques presented need to change in Section~BLAH.

%
%of binary causal adversarial bit-flip channels. More precisely, we propose and analyze a novel coding scheme that allows Alice to communicate reliably with Bob over the channel at rate 
%\begin{align*}
%R=C_{p}-\erate
%\end{align*}
%for any $\erate>0$, where $C_{p}=\capopt$ and $\alfun=1-4(p-\pbar)$ Our result is optimal based on \cite{dey2013upper}, which provides an upper bound the same as our lower bound for the capacity of the causal bit-flip channel.

\subsection{Previous Work}

We start by briefly summarizing the state-of-the-art for erasure and bit-flip adversarial channels, for both omniscient and oblivious adversaries.
The optimal rate of communication over binary omniscient adversarial channels (for both erasure and bit-flip errors) are long standing open problems in coding theory. The best known lower bounds for the problems derive from the Gilbert-Varshamov codes (the GV bound)~\cite{gilbert1952comparison,gilbert1957comparison}, and the tightest upper bounds (the MRRW bounds) from the work by McEliece {\it et al.}~\cite{mceliece1977new}. 
%Note that the GV bound and the MRRW bound are the best lower bound and upper bound thus far for both the erasure and the bit-flip omniscient adversarial channel.

The literature on Arbitrarily Varying Channels (AVCs, e.g., \cite{lapidoth1998reliable}) implies that the capacity of the binary oblivious adversarial bit-flip channel is $1-H(p)$, and that of oblivious adversarial erasure channels is $1-p$; these match the well-known capacities of the corresponding ``random noise'' channels with bits flipped or erased Bernoulli($p$), but are attainable even for noise patterns that can be chosen (up to an overall constraint of a $p$-fraction corruptions) by an adversary with full knowledge of the codebook, but no knowledge of the actually transmitted codeword.\footnote{In fact, it can even be shown that if Alice is allowed to use {\it stochastic encoding} -- choosing one of multiple possible codewords randomly for each message she wants to transmit -- then even for a {\it maximal probability of error} metric, a vanishingly small probability of error can be attained by capacity achieving codes. That is, there exists a sequence of codes whose rates asymptotically achieve the corresponding capacity, and such that for every message transmitted by Alice and for every corruption pattern imposed by Calvin, can be decoded correctly by Bob for ``most'' codewords corresponding to that message.} An alternate proof of the capacity of the binary oblivious bit-flip channel was presented in~\cite{Langberg:08} by Langberg, and a computationally efficient scheme achieving this rate was presented in~\cite{guruswami2010codes} by Guruswami and Smith.
%--  and random parity-check coding \cite{elias1956coding} is known to achieve this rate.
%The capacity of the binary oblivious adversarial bit-flip channel is $1-H(p)$ \cite{lapidoth1998reliablel} and the work of Guruswami and Smith \cite{guruswami2010codes} demonstrates computationally efficient code constructions to achieve rate-optimal throughput.
% of the oblivious bit-flip channel (and similar techniques would also work for computationally efficient codes for the oblivious erasure channel).

%$1-H(p)$ for all $p<\frac{1}{2}$.

We now turn to the causal setting.
As a causal adversary can never do better than an omniscient adversary and does at least as well as an oblivious one, the upper bounds on capacity for oblivious adversaries specified above act as upper bounds for the causal case as well;
and the lower bounds on capacity for omniscient adversaries act as lower bounds for the causal case.
For the binary causal adversarial bit-flip channel both bounds were improved.
Specifically, the first nontrivial upper bound $\min\left\lbrace 1-H(p),(1-4p)^{+}\right\rbrace $ was given by Langberg {\it et al.}~\cite{langberg2009binary}, and later, the tightest upper bound was given by the continuing work of Dey {\it et al.}~\cite{dey2012improved,dey2013upper} and is equal to $ C^{flip}_p$ of Theorem~\ref{the:main:flip}. The best lower bound was described by Haviv and Langberg~\cite{haviv2011beating} which slightly improves over the GV bound. 
For the binary causal adversarial erasure channel the trivial upper bound of $1-p$ was improved to $1-2p$ (which we demonstrate equals $C^{erase}_p$ in Theorem~\ref{the:main:erase}) by Bassily and Smith~\cite{bassily2014causal} who also present improved lower bounds that  separate the achievable rate for causal adversarial erasures from the rates achievable for omniscient adversarial erasures.

%The capacities and bounds mentioned above are depicted in Figure~\ref{fig:bounds}.

\begin{figure}[htb]
	\centering
	\begin{subfigure}{.5\textwidth}
		\centering
		\includegraphics[width=1\linewidth]{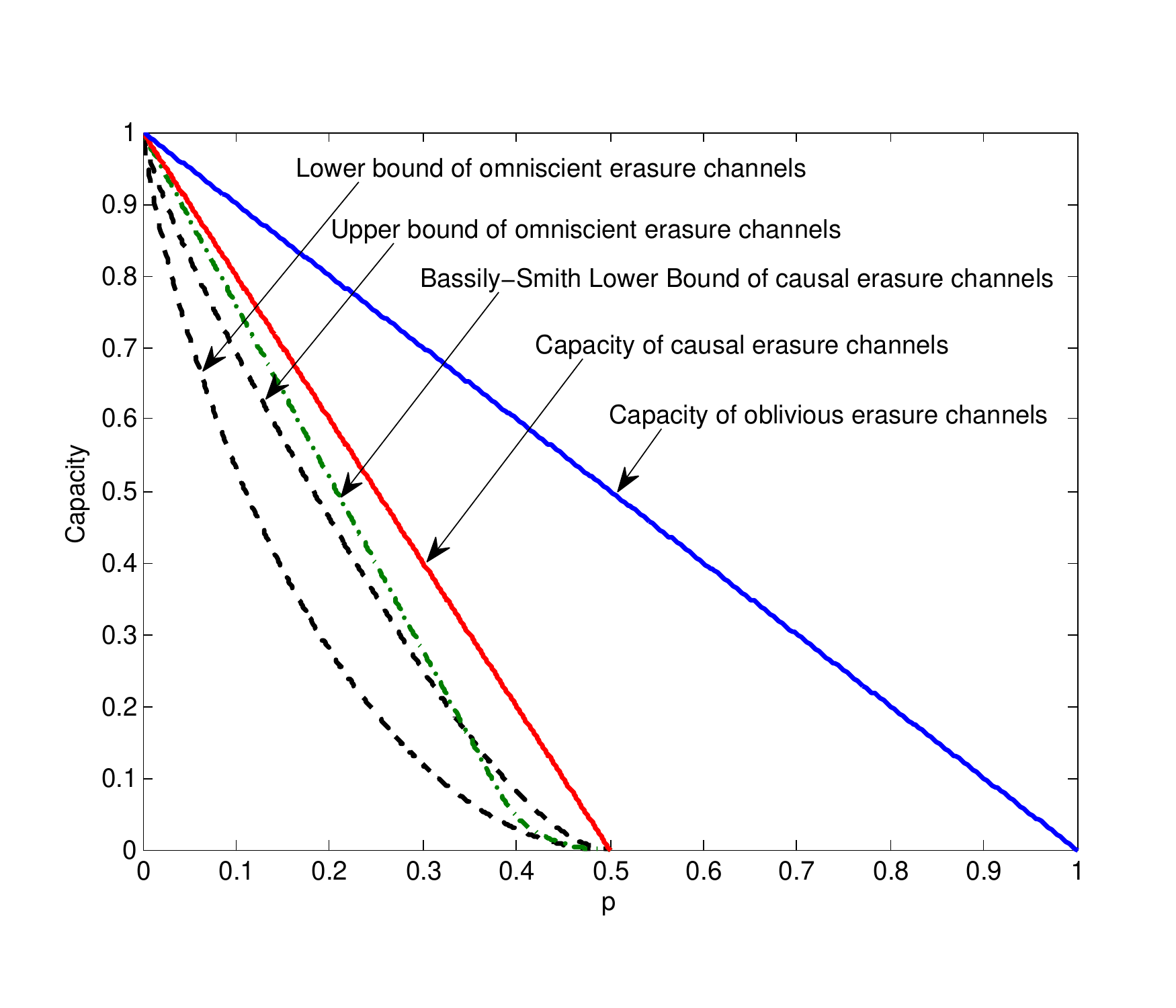}
		\caption{Binary adversarial erasure channels: The bound of $1-p$ (in blue) corresponds to the capacity of binary oblivious erasure channel. The MRRW bound and the GV bound (both in dotted black) are the best known upper and lower bounds for binary omniscient erasure channels. The lower bound for binary causal erasure channels by Bassily and Smith~\cite{bassily2014causal} is plotted in green. The upper bound $1-2p$ (in red) of~\cite{bassily2014causal} matches the lower bound we demonstrate in Theorem~\ref{the:main:erase}.}
		\label{fig:bound-bec}
	\end{subfigure}%
	\begin{subfigure}{.5\textwidth}
		\centering
		\includegraphics[width=1\linewidth]{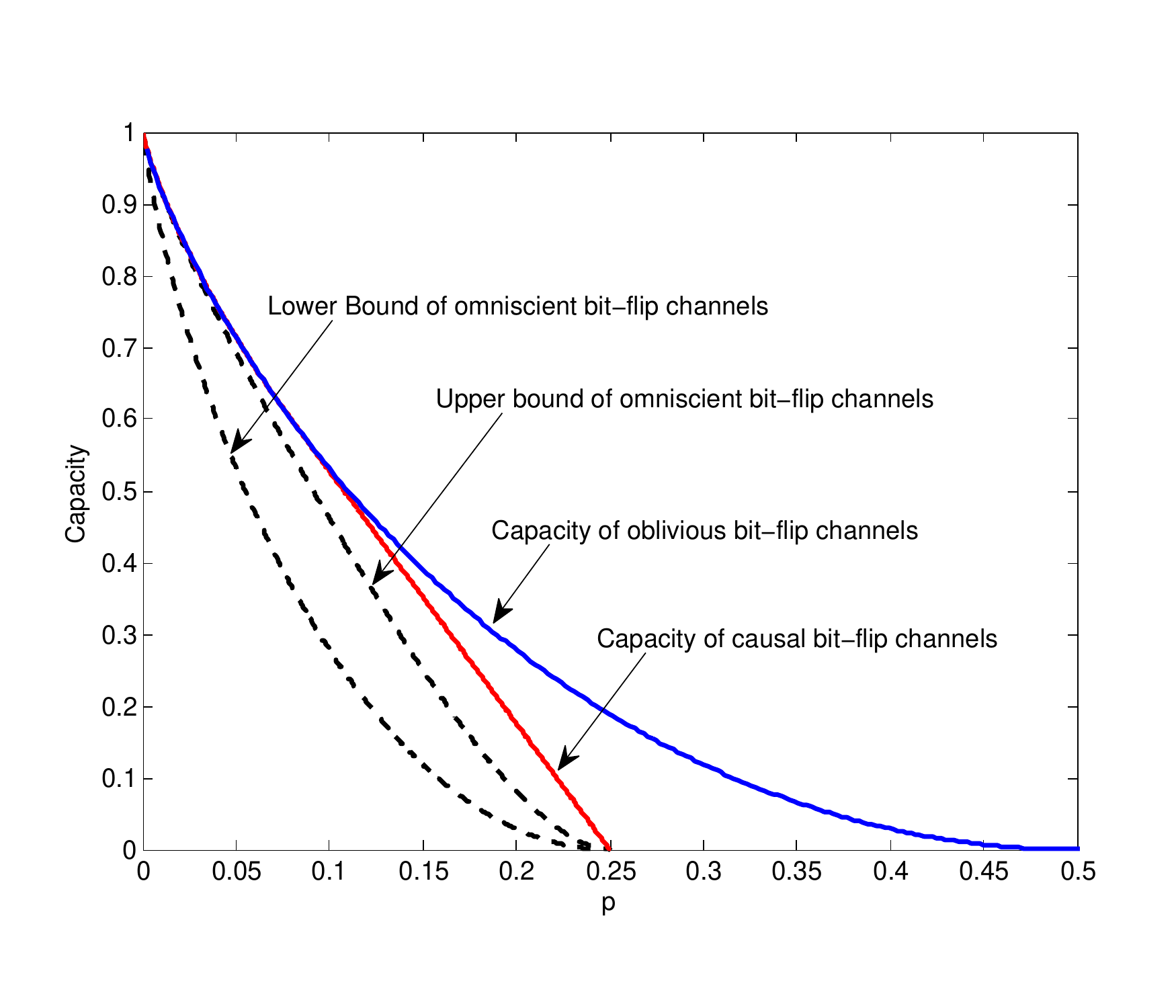}
		\caption{Binary adversarial bit-flip channels: The bound of $1-H(p)$ (in blue) corresponds to the binary oblivious bit-flip channel. The MRRW bound and the GV bound are the best upper and lower bounds (both in dotted black) for the binary omniscient bit-flip channels. For binary causal bit-flip channels, the previous lower bound by Haviv and Langberg~\cite{haviv2011beating} is a slight improvement over the GV bound. Our bound of Theorem~\ref{the:main:flip} matches the previous upper bound by Dey {\it et al.}\cite{dey2012improved,dey2013upper} (in red).}
		\label{fig:bounds-bsc}
	\end{subfigure}
	\caption[Bounds on the capacity of binary adversarial channels]{Bounds on the capacity of binary adversarial channels}
	\label{fig:bounds}
\end{figure}

%The causal model is a variant of the arbitrarily varying model (AVC) \cite{blackwell1960capacities} in which share secret or common randomness may be allowed whereas in our model it is not (we use stochastic encoding). The AVC model where the adversary has access to the entire codeword was studied by Ahlswede and Wolfowitz \cite{ahlswede1970capacity,ahlswede1973channels}. General AVC models have been broadened to subsume channel models with constraints on the adversary for cases where the adversary has no access to the entire codeword \cite{csiszar1988capacity}, or has access to the whole codeword \cite{sarwate2010rateless}. For binary omniscient adversarial channels, the work \cite{langberg2004private} showed that the rate of $1-H(p)$ is achievable using $O(\log n)$ bits of common randomness, and the work \cite{smith2007scrambling} investigated computationally efficient code constructions to achieve the rate. 

Related results include the study of binary delayed adversaries by Dey {\it et al.}~\cite{dey2010coding} who provide a characterization of the capacity in the case of ``delays'' $d$ which are an arbitrarily small (but constant) fraction of the code block length $n$.\footnote{While not presented in that work, the techniques of~\cite{dey2010coding} can be used to show that the same capacity holds even if the delay is $\polylog(n)$ rather than $d={\cal O}(n)$.} The value $d$ here corresponds to an adversarial model in which the decision of whether or not to corrupt the $i$th codeword bit depends only on $(x_1,\ldots,x_{i-d})$ (and the overall constraint on the number of bits that can be corrupted).
It is interesting to note that, in this case as well as the oblivious one, the capacity of the bit-flip and bit-erasure channels matches the corresponding random noise capacities (of $1-H(p)$ and $1-p$). On the other hand, as mentioned, the causal and $n\epsilon$ lookahead settings have strictly lower, but approximately matching, capacities. This seems to imply that the knowledge of the present is critical for Calvin to significantly depress the capacity below the random noise capacity. 

While the above discussion relates to the problem of binary alphabets, the work of  Dey {\it et al.}~\cite{dey2013codes} considered ``large alphabet channels'' (in which the alphabet size is ``significantly larger'' than the block-length $n$) with causal symbol errors.\footnote{The capacity of large alphabet causal symbol erasures is essentially the same as that of omniscient large alphabet symbol erasures, which in turn equals the capacity of random symbol erasures. Such rates can be directly attained by Reed-Solomon codes, and matching converses obtained by Calvin merely randomly erasing $pn$ symbols.} A complete capacity characterization was presented (with corresponding computationally efficient codes attaining capacity), which demonstrated that the capacity of this problem equals $1-2p$, which is the same as the capacity of an omniscient adversary (attained by Reed-Solomon codes, and impossibility of higher rates by the Singleton bound). This demonstrates that the penalty imposed by the causality condition on Calvin diminishes with increasing alphabet size.

Also related to this work is the study of Mazumdar~\cite{mazumdar2014capacity} in which the capacity of memoryless channels where the adversary makes his decisions based only on the value of the currently transmitted bit is addressed. We note that the causal model is also a variant of the AVC model~\cite{blackwell1960capacities,lapidoth1998reliable}, however previous works on AVCs with capacity characterizations do not relate directly to the study at hand on causal adversaries.

%\zt{cited the lapidoth narayan AVC survey}

\subsection{Proof Technique}
%\sj{is \cite{guruswami2009explicit} the right citation? shouldn’t it be AVCs?}
%
%\zt{see \url{http://www.cse.psu.edu/~asmith/pubs/2009/wae-codes.pdf}}
%
%\sj{again, haven’t read \cite{algoet1988sandwich} — just to make sure that the person citing it is mike/knows it’s appropriate. maybe cite AVC literature instead?}
%
In this work we presents a coding scheme for both binary causal adversarial bit-flip and erasure channels that achieves rates that match the known upper bounds (respectively~\cite{dey2012improved,dey2013upper} and~\cite{bassily2014causal}) from previous studies, and as such is optimal. Our coding scheme is an existential one, and is based on {\em stochastic} encoders. Specifically, in our codes the encoder Alice uses internal randomness (not known to the receiver Bob or the adversary Calvin) in the choice of the transmitted codeword. The internal randomness is designed to allow a high probability of successful communication no matter which message Alice is sending to Bob.
%\ml{\st{We note that the use of randomness in the encoding process could be removed at the price of an average success criteria that guarantees successful communication w.h.p. over the message transmitted by Alice. Nevertheless, the use of randomness in the encoding process is an integral part of our analysis.}}
%{\sj{removed the sentence ``We note that the use of randomness in the encoding process could be removed'' }}

In both the erasure and bit-flip channels, the codes we use are random codes. Namely, we pick our codes uniformly at random from all possible codes of a specific rate, and prove using probabilistic tools, {\it e.g.}, \cite{algoet1988sandwich}, that w.h.p. over the code distribution a code chosen at random allows reliable communication over our channels.
The decoder used in both the erasure and bit-flip case involves two major phases: a list decoding phase in which the decoder obtains a short list of \sj{messages} that includes the one transmitted; and a unique decoding phase in which the list is reduced to a single message. Roughly, Bob in his decoding process divides the received word into two parts -- all bits received up to a given time $\tstar$, and all bits received afterwards. The list decoding is done using the first part of the received word, and the process of unique decoding from the list is done using the second part. Both the list decoding and the unique decoding phase are done by considering the family of codewords of a certain given distance from the (first or second part of the) received word.

In the erasure case, given the parameter $p$ (that specifies the fraction of bits that can be erased by the adversary) and the received word, the decoder Bob can pin-point the value of $\tstar$ that will allow successful decoding. Specifically, for any adversarial behavior, we show the existence of a value $\tstar$ that satisfies two special properties: the {\em list decoding condition} that allows Bob to obtain a small list of \sj{messages} from the first part of the received word; and the {\em energy bounding condition} which guarantees that the fraction of bits erased by the adversary in the second part of the received word cannot suffice to confuse Bob between any two \sj{messages in} the list he holds. The ability to list decode is obtained using standard probabilistic arguments that take into account the block length $\tstar$ and the fraction of erasures $p_{\tstar}$ in the first part of the received word.
The ability to uniquely decode from the obtained list involves a more delicate analysis which uses the stochastic nature of our encoding and the causality constraint of Calvin. In particular,  we use the fact that the secret bits used in the encoding of the first part of the codeword (up to position $\tstar$) are independent of those used for the second part. This independence is useful in separating the two decoding phases 
in the sense that the casual adversary at time $\tstar$ is acting with no knowledge whatsoever on the secret bits used by Alice after time $\tstar$. This lack of knowledge sets the stage for the unique decoding phase.

We accommodate different potential values of $\tstar$ by designing a stochastic encoding process in which different parts of the codewords rely on independent secret bits of Alice. Namely, we divide the coding process into {\em chunks}. Each chunk is a random stochastic code of length $\chunk$ for a small parameter $\echunk$ that uses independent randomness from Alice. The final code of Alice is a concatenation of all its chunks. Setting $\echunk$ small enough allows enough flexibility to manage any possible value $\tstar$ chosen by Bob's decoder.

The encoding and decoding process for the bit-flip channel follow the same line of analysis as specified above for the erasure case, but with one major and significant difference. Bob does not know which bits in the transmitted codeword were flipped, and thus by studying the received word, Bob is not able to identify a location $\tstar$ that satisfies the list decoding and energy bounding conditions. To overcome this difficulty, we design an iterative 
decoding process in which Bob starts with a small value of $t$ and performs an attempt to decode.
As before the decoding process first list decodes using the first part of the received word and then uniquely decodes. The list decoding is done according to a certain ``guessed'' value $\phatt$ for the fraction of bit flips in the first part of the received word. Here, $\phatt$ is a carefully designed function of $t$  (also referred to as a ``trajectory'') that is fixed and known to all parties involved in the communication. We show that if Bob guessed correctly, i.e., if the fraction $\hat{p}_t$ is equal to the actual fraction of bits $p_t$ flipped in the first part of the received word, then the decoding will successfully return the correct message. Otherwise, we show that the unique decoding phase will fail in the sense that Bob will not receive any message from the decoding process. Identifying a failure in the decoding process, Bob increases $t$ and repeats the decoding attempt. The crux of our analysis lies in our proof that eventually, no matter what the behavior of Calvin is, there will be a value of $t$, denoted $\tstar$, for which $\phat_{\tstar}$ is approximately $p_{\tstar}$ and the decoding succeeds.
Establishing the existence of the trajectory $\phatt$, and proving that at some point it must be close to $p_t$ is a central part of our proof.
%We note that the obtained capacity $C_p^{flip}$ has an intuitive interpretation which is tied to the ``babble and push'' adversarial strategy that was presented in \cite{dey2013upper} in the proof of the matching upper bound. We elaborate on this intuition in the conclusion to our work.

\subsection{Structure}

In the remaining part of this extended abstract we focus on our proof for the bit-flip case, i.e., Theorem~\ref{the:main:flip}, as it includes the ideas needed for the simpler proof of Theorem~\ref{the:main:erase}. The proof of Theorem~\ref{the:main:erase} appears in a self contained manner in Appendix~\ref{app:erasure}.
In Section~\ref{sec:model} we formally present the channel model, the encoder, and the decoding process. In addition, we present a careful description of the adversarial behavior. Section~\ref{sec:code-analysis} then presents an overview of our code analysis and the proof of Theorem~\ref{the:main:flip} (our main result for the bit-flip channel). Due to space limitations, all the technical claims for the bit-flip channel and their proof appear in the Appendix~\ref{app:flip}.

\section{Model}
\label{sec:model}

\paragraph{Channel Model:}
For any positive integer $i$, let $\left[ i\right] $ denote the set $\left\lbrace 1,2,\cdots,i\right\rbrace $. For a transmission duration of $n$ bits, a binary causal adversarial bit-flip channel can be characterized by a parameter $p\in\left[ 0,1\right] $ and a triple $\left( \mathcal{X}^n,\textsf{Adv},\mathcal{Y}^n\right) $. Here, $p$ is the fraction of bits that Calvin can flip in a codeword, $\mathcal{X}=\left\lbrace 0,1\right\rbrace $ and $\mathcal{Y}=\left\lbrace 0,1\right\rbrace $ are the input and output alphabet of the channel, and $\textsf{Adv}=\left\lbrace \text{Adv}^i|i\in\left[ n\right] \right\rbrace $ is a sequence of mappings that represents the adversarial behavior in each time step.
%such that by the end of transmitting $n$ bits, the number of bit-flips in the codeword is at most $np$. 
More precisely, each map $\text{Adv}^i:\mathcal{X}^i\times\mathcal{Y}^{i-1}\to\mathcal{Y}$ is a function that, at the time of transmitting the $i$-th bit, maps the sequence of channel inputs up to time $i$, $\left( x_1,x_2,\cdots,x_i\right) \in\mathcal{X}^i$, together with the sequence of all previous channel outputs up to time $i-1$, $\left( y_1,y_2,\cdots,y_{i-1}\right) \in\mathcal{Y}^{i-1}$, to an output symbol $y_i\in\mathcal{Y}$.
The functions  $\text{Adv}^i$ must satisfy the adversarial power constraint, namely that at no point in time does the total number of bit-flips exceed $pn$.

\paragraph{Random code distribution:}
We now define a distribution over codes. In our proof, we use this distribution to claim the existence of a fixed code that allows reliable communication between Alice and Bob over the channel model.
In our code construction $R$ denotes the code \textit{rate}, $S$ the {\em private secret rate} of the encoder (to be defined explicitly shortly), and $\echunk$ a ``quantization'' parameter  (specified below).

%\sj{how about making this $S_1,…,S_{1/\theta}$, to make it consistent with the encoder below?}

Let $\msg=\msgspace$ denote Alice's message set and $\secr=\keyspace$ be the set of private random secrets available only to Alice. The encoder randomness $\secr$ is neither shared with the receiver nor the adversary. Let $\Phi$ be the uniform distribution over stochastic codes $\msg \times \secr \to \mathcal{X}^{\chunk}$. Let $\mathcal{C}_1,\mathcal{C}_2,\cdots,\mathcal{C}_{\chunknum}$ be stochastic codes, which are i.i.d. according to the probability distribution $\Phi$. Specifically, $\forall i\in\left[ \chunknum\right] $, the corresponding stochastic code is a map $\mathcal{C}_i:\msg \times \secr \to \mathcal{X}^{\chunk}$ chosen from the distribution $\Phi$.  

\paragraph{Encoder:}
Given a message $m\in\msg$ and $\chunknum$ secrets, $s_1,s_2,\cdots,s_{\chunknum}$ each in $\secr$, a codeword of length $n$ with respect to the message $m$ and the $\chunknum$ secrets is defined to be the concatenation of $\chunknum$ {\em chunks} of sub-codewords,
\begin{align}
	\megacodw[\chunknum] \label{code-design}
\end{align}
where $\mathcal{C}_{i}(m,s_{i})$ is the $i$-th \textit{sub-codeword} in the entire codeword, and $\circ$ denotes the concatenation between two chunks of sub-codewords. To distinguish the concatenated code $\mcode$ from the code for a chunk, we will call $\mathcal{C}_1,\mathcal{C}_2,\cdots,\mathcal{C}_{\chunknum}$ \textit{sub-codes} hereafter. Our analysis then focuses on two mega sub-codes, defined as follows.

\begin{definition}
	\label{def:left-mega-code}
	Let a code $\mcode$ of block-length $n$ consist of $\chunknum$ sub-codes, i.e., $\mcode=\megacode[\chunknum]$. Let $\setchend=\choiceschunk$ and $t\in\setchend$. A \textit{left} mega sub-code of $\mcode$ with respect to $t$ is the concatenation of the first $\frac{t}{\chunk}$ sub-codes of $\mcode$.
\end{definition}

\begin{definition}
	\label{def:right-mega-code}
	Let a code $\mcode$ of block-length $n$ consist of $\chunknum$ sub-codes, i.e., $\mcode=\megacode[\chunknum]$. Let $\setchend=\choiceschunk$ and $t\in\setchend$. A \textit{right} mega sub-code of $\mcode$ with respect to $t$ is the concatenation of the last $\chunknum-\frac{t}{\chunk}$ sub-codes of $\mcode$.
\end{definition}

In our analysis, it is convenient to describe the encoding scheme of Alice in a causal manner. Namely, we will assume that the secret value $s_i$ corresponding to the encoding of the $i$-th chunk is chosen by Alice immediately before the $i$-th chunk is to be transmitted and no sooner.

As mentioned above, we show that with positive probability, the code $\mathcal{C}$ chosen at random based on the distribution above has certain properties that allow reliable communication over our channel model. 

\paragraph{Decoding process:}
The decoding process of Bob is done in an iterative manner. Specifically, upon receiving the entire codeword with bit-flips, Bob identifies the smallest value of $t \geq (1-4p)n$ corresponding to the (end) location of a chunk, and attempts to correctly decode the transmitted message $m$ based on the left and right mega sub-codeword with respect to position $t$. 
%(Formally, Bob identifies the integer value of $k$ such that $k>\frac{1-4p}{\echunk}$ and sets $t=k\chunk$.)
The  decoding process is terminated if a message is decoded by Bob, otherwise the value of $t$ is increased by $\theta n$ (the chunk size) and Bob attempts to decode again. This process continues until $t$ reaches the end of the codeword. If no decodings succeeds until then, a decoder error is declared. 
%In the former case, as we will show, the decoding is successful (with high probability over the random bits of Alice); in the latter case we will declare a decoding error. 

Each attempt of decoding can be divided into two phases. First, at each position $t$, Bob chooses an estimate $\phatt$ (to be specified shortly) for the fraction of bit-flips used by Calvin in the left mega sub-codeword with respect to position $t$. In our proof to come, we show that $\phatt$ satisfies two important conditions, the \textit{list-decoding condition} and the \textit{energy bounding condition} (see Claim~\ref{claim:1}). The list-decoding condition allows  Bob to decode the left mega sub-codeword $\megacodw[\frac{t}{\chunk}]$ through a list decoder with list size $\mlists$. As we will show, the list size $\mlists$ \ml{consists of at most $\blistord$ messages}. 
Here, and in what follows, $\varepsilon>0$ is a constant design parameter that can be considered to be arbitrarily small.
So at this phase Bob obtains a \sj{list $\mlist$ of $\mlists$ {messages}}.
%, one of which is the correct transmitted codeword. 
If it is the case that $\phatt$ equals the true fraction of bits $\pref{p}$ flipped by the adversary up to position $t$, then it holds that the transmitted \ml{message} is in $\mlist$. 

Next, for the second phase,  the energy bounding condition states that, if $\phatt$ equals $\pref{p}$, there are no more than $\left(\frac{1}{4}-\frac{\erate^2}{16}\right)\left(n-t\right)$ bits flipped in the right mega sub-codeword with respect to position $t$. Therefore, as we will show, Bob can use a natural \textit{consistency} decoder (defined below) to determine whether to stop or continue the decoding process. More precisely, the decoding process continues if the consistency decoder fails to return a message and stops if a message $\hat{m}$ is decoded from the \ml{messages} in $\mlist$. The decoder also stops when $t$ has reached size $n-\chunk$.

\begin{definition}
	Let $\varepsilon > 0$. Let $\mathbf{x}_{t},\mathbf{x}^{\prime}_{t}\in\mathcal{X}^{n-t}$ be two right mega sub-words with respect to position $t$. The right mega sub-word $\mathbf{x}_{t}$ is \textit{consistent} with the right mega sub-word $\mathbf{x}^{\prime}_{t}$ if and only if the number of the positions that $\mathbf{x}_{t}$ does not agree with $\mathbf{x}^{\prime}_{t}$ is no more than $\left(\frac{1}{4}-\frac{\erate^2}{16}\right)\left(n-t\right)$.
\end{definition}

\begin{definition}
	A consistency decoder applied to a right mega sub-code $\rmegacode[\frac{t}{\chunk}]$ with respect to position $t$ and list $\mlist$ is a decoder that takes the right mega sub-word of a received word $\mathbf{x}^{\prime}$ and returns a unique message $\hat{m}$ in the list $\mlist$, \sj{one of} whose right mega sub-codewords \sj{is} consistent with that of $\mathbf{x}^{\prime}$. If more than one such message exists, then a decoding error is declared.% and no message is returned.
\end{definition}

Formally, the decoder process of Bob can be described as follows. Essentially, we will use the following definition of $\phat_{t}$ (the estimate to Calvin's corruption fraction at time $t$ used by Bob), which is slightly revised later in Definition~\ref{def:p-hat-t} to be more robust to slight slacknesses that appear in the analysis. Let $p\in\intvquar$, then for $t<n(1-4p)$, $\phatt=0$; otherwise $\phatt = p-\left(\frac{1}{4}-p\right)\left(\frac{n}{t}-1\right)$.  The value of $\phatt$ is $0$ for all $t$ up to $n(1-4p)$ and then it grows up to $p$ as $t$ increases to $n$. For the description below, recall that $\varepsilon>0$ is a constant design parameter that can be considered to be arbitrarily small.

\noindent {\bf 1.} Find the smallest integer $k$ with  $t=k\chunk$ such that $t \geq n(1-4p)$. Namely, let $k_0=\left\lceil \frac{1-4p}{\echunk}\right\rceil$. Then $t=t_0=k_0\chunk$.\newline
 \hspace{0.1in} {\bf 2.} List-decode the left mega sub-code $\megacode[\frac{t}{\chunk}]$ to obtain a list $\mlist$ of \ml{messages} of size $\mlists$, with the list-decoding radius $t\phat_{t}$. 
	More precisely, a \ml{message} is in the list $\mlist$ if \ml{there is a codeword corresponding to $m$ for which} its left mega sub-codeword is of distance no more than $t\phatt$ from the corresponding left mega sub-word of the received word.\newline
\noindent {\bf 3.} Verify the right mega sub-codewords \ml{corresponding to messages} in the list $\mlist$ through a consistency decoder. Specifically, consider the Hamming balls with radius $r=\frac{n-t}{2}-\frac{(n-t)\erate^2}{8}$ centered at the right mega sub-codeword of each codeword \ml{corresponding to a message} in the list $\mlist$. If the corresponding right mega sub-word of the received word is outside all the balls, increase $t$ by $\chunk$ and goto Step (2). If the received right mega sub-word lies in exactly one of the balls, decode to the message corresponding to the center of the ball. If the received right mega sub-word lies in more than one ball 
%and all the centers of the balls correspond to the same message, the received word is also decoded to this message. If the received right mega sub-word lies in more than one ball and the centers of the balls corresponds to different messages, 
a decoding error is declared.

%\begin{enumerate}
%	\item Identify the smallest integer $t=k\chunk$ such that $t>n(1-4p)$. Namely, let $k_0=\left\lceil \frac{1-4p}{\echunk}\right\rceil$. Then $t=t_0=k_0\chunk$.
%	\item List-decode the left mega sub-code $\megacode[\frac{t}{\chunk}]$ to obtain a list $\mlist$ of codewords of size $L$, with the list-decoding radius $t\phat_{t_0}$. 
%	More precisely, a codeword is in the list $\mlist$ if its left mega sub-codeword is of distance no more than $t\phatt$ from the corresponding left mega sub-words of the received word.
%	\item Verify the right mega sub-codewords of the codewords in the list $\mlist$ through a consistency decoder. Specifically, consider the hamming balls with radius $r=\frac{n-t}{2}-\frac{(n-t)\erate^2}{8}$ centered at the right mega sub-codeword of each codeword in the list $\mlist$. If the corresponding right mega sub-word of the received word does not lie in any one of the balls, increase $t$ by $\chunk$ and goto Step (2). If the received right mega sub-word lies in one and only one of the balls, decode to the message corresponding to the center of the ball. If the received right mega sub-word lies in more than one ball and all the centers of the balls correspond to the same message, the received word is also decoded to this message. If the received right mega sub-word lies in more than one ball and the centers of the balls corresponds to different messages, a decoding error is declared.
%\end{enumerate}

For every message $m$, Bob decodes correctly if his estimate $\hat{m}$ equals $m$. 
That is, Bob decodes correctly if for some $\tstar$, the only right mega sub-codeword of the codewords \ml{corresponding to messages in the list $\mlist$} consistent with that of the received codeword corresponds to the message $m$. We show that this indeed happens w.h.p. over the random secrets $\secr^{n-t}$ used by Alice for the right mega sub-codeword with respect to position $\tstar$.
If Bob's estimate $\hat{m}$ is not equal to $m$, Bob is said to make a \textit{decoding error}. The \textit{probability of error} for a message $m$ is defined as the probability over Alice's private secrets $s\in\secr$ that Bob decodes incorrectly. The probability of error for the code $\mcode$ is defined as the maximum of the probability of error for message $m$ over all messages $m\in\msg$.

A rate $R$ is said to be \textit{achievable} if for every $\xi>0,\beta>0$ and every sufficiently large $n$ there exists a code of block length $n$ that allows Alice to communicate $2^{n(R-\beta)}$ distinct messages to Bob with probability of error at most $\xi$. The supremum over $n$ of all achievable rates is the {\it capacity} $C_p$ of the channel.

\paragraph{Adversarial behavior:}
The behavior of Calvin is specified by the channel model above. Specifically, the behavior of Calvin can be characterized by a function $\pref{p}$ defined below which specifies how many errors were ejected by Calvin up-to position $t$. We refer to $\pref{p}$ as a trajectory, and note that the exact trajectory used by Calvin is not known to the decoder Bob.

%there may be several potential trajectories that model Calvin's behavior.
%show some possible trajectories of $\pref{p}$ in Figure~\ref{fig:p-t-cases}. More importantly, if Calvin has to use up all his bit-flips $np$, the highest curve (``squander all the bit-flips at the beginning'') and the lowest curve (``reserve all the bit-flips for the end'') in Figure~\ref{fig:p-t-cases} encompasses a region that all possible trajectories of $\pref{p}$ must lie in.

%\begin{figure}[htbp]
%	\centering
%	\includegraphics[scale=0.6]{p-t-cases-with-number-labels.pdf}
%	\caption{Four possible trajectories of $\pref{p}$ when $p=\frac{1}{8}$, $n=40000$ and $\chunk=32$: $1.$ Squander all the bit-flips at the beginning; $2.$ Typical trajectory of Calvin; $3.$ Reserve all the bit-flips for the end; $4.$ Uniformly distribute all bit-flips to each chunk.}
%	\label{fig:p-t-cases}
%\end{figure}

\begin{definition}[Calvin's Trajectory $\pref{p}$]
	\label{def:p-t}
	Let a codeword of length $n$ consist of $\chunknum$ chunks of sub-codewords. Let $\setchend=\choiceschunk$ and $\leftcw\in\setchend$. Then $\pref{p}$ is the actual portion of bit-flips that occurred in the left mega sub-codeword with respect to $t$.
\end{definition}

In our analysis we assume that Calvin has certain capabilities that may be beyond those available to a causal adversary. This is without loss of generality as we are studying lower bounds on the achievable rate in this work. %and thus, our bounds will also hold for the original (more restricted) casual adversary Calvin. 
We assume that the trajectory of $\phatt$ that Bob uses in his decoding process is known to Calvin.
This implies (as we will show) that Calvin knows the position $\tstar$ that Bob eventually stops his decoding process. 
In addition, we assume that the list of \ml{messages} obtained through Bob's list decoding process can be determined explicitly by Calvin. Moreover, we assume that Calvin knows the message $m$ {\it a priori}. 
%We denote the list of codewords excluding the codeword $\mathbf{x}$ corresponding to the transmitted message as $\exlist$.

At every list-decoding position $t$, we stress that the subsequent secrets, namely, $(s_{\frac{t}{\chunk}+1},s_{\frac{t}{\chunk}+2},\cdots,s_{\chunknum})$ for the right mega sub-codeword are unknown to Calvin. Indeed, given the causal nature of Alice's encoding, these secrets have not even been chosen by Alice at this point in time. The fact that the secrets are hidden from Calvin implies that $(s_{\frac{t}{\chunk}+1},s_{\frac{t}{\chunk}+2},\cdots,s_{\chunknum})$ are completely independent of the list (obtained through Bob's list decoding) $\mlist$ determined by Calvin. This fact is crucial to our analysis.

Also, we strengthen Calvin by allowing him to choose which bits to flip after position $\tstar$ non-causally. Namely, we assume that Calvin chooses his bit-flip pattern after looking ahead to all the remaining bits of the transmitted codeword. As we show, no matter how these bit-flips are chosen, the right mega sub-codeword has at most $\left(\frac{1}{4}-\frac{\erate^2}{16}\right)\left(n-\tstar\right)$ bits flipped. The fact that the distribution of $(s_{\frac{\tstar}{\chunk}+1},s_{\frac{\tstar}{\chunk}+2},\cdots,s_{\chunknum})$ is independent from the list $\mlist$ will allow us to show that Bob succeeds in his decoding.

\section{Code Analysis for the bit-flip Channel}
\label{sec:code-analysis}
Due to space limitations, the technical details of our proof appear entirely in the Appendix. In what follows, we give a roadmap for our proof, including the major high-level arguments used in the Appendix. 
Throughout, $\varepsilon>0$ is a constant design parameter that can be considered to be arbitrarily small.
%Before diving into the details of the code analysis, we first give an overview of the techniques used in the analysis. 

\paragraph{Existence of trajectory $\phatt$:} Our analysis of Bob's decoding begins with selecting a {\it decoding reference trajectory} $\phatt$ (Definition~\ref{def:p-hat-t}) as a proxy trajectory for Calvin's trajectory $\pref{p}$. Recall that for each $t$, $\pref{p}$ is the fraction of bit-flips in the left mega sub-codeword with respect to position $t$, and accordingly, $\phatt$ is the fraction of bits that Bob {\em assumes} are flipped up to position $t$. In general, the trajectories $\phatt$ and $\pref{p}$ are not equal.
We show in Claim~\ref{claim:1}, that for $t\geq n(1-4\ppm)$ the selected decoding reference trajectory $\phatt$ satisfies two important conditions, the \textit{list-decoding condition} \eqref{eq:con-1} and the \textit{energy bounding condition} \eqref{eq:con-2} introduced and explained previously. The list decoding condition guarantees a small list size if decoding is done with radius $\phatt$; and the energy bounding condition restricts the remaining bit-flips that the adversary has for the right mega sub-codeword if Bob's estimate $\pref{p}$ to $\phatt$ is approximately correct.

To prove correctness of our decoding procedure, we must introduce a new trajectory $\tilde{p}_{t}$ and an additional parameter $\myp$ before we continue. These additional parameters are closely related to their counterparts in the sense that $\tilde{p}_{t}$ approximately equals $\phatt$ and $\ppm$ approximately equals $\myp$, but in both cases the former is slightly smaller than the latter.
The parameters are introduced to allow robustness in our analysis which absorbs certain slacknesses that are a result of our code construction and analysis technique (e.g., such as the fact that our chunk size $\chunk$ cannot be made too small). We here give our precise definitions, which can be at times better understood intuitively if the reader keeps the above discussion in mind. All our notations are given in Table~\ref{tab:params}.

\paragraph{Existence of position $\tstar$ for which $\phat_{\tstar} \simeq p_{\tstar}$:}
Next in our analysis we chooses the position $t_{0}=\left\lceil\frac{1-4\myp}{\echunk}\right\rceil\chunk \simeq n(1-4p)$ as a {\em benchmarking} position, and separate our analysis into two cases based on whether $p_{t_{0}}$ is greater than $\tilde{p}_{t_0}$ (See Definition~\ref{def:p-tilde-t}) or not.
We use the following classification:
 
\begin{definition}[High Type Trajectory]
	For any trajectory $\pref{p}$ of Calvin, consider the values of $\pref{p}$ and $\tilde{p}_{t}$ at position $t=t_{0}$. If $p_{t_{0}}\geq\tilde{p}_{t_0}$ then Calvin's trajectory $\pref{p}$ is a \textit{high} type trajectory.
\end{definition}

\begin{definition}[Low Type Trajectory]
	For any trajectory $\pref{p}$ of Calvin, consider the values of $\pref{p}$ and $\tilde{p}_{t}$ at position $t=t_{0}$. If $p_{t_{0}}<\tilde{p}_{t_0}$ then Calvin's trajectory $\pref{p}$ is a \textit{low} type trajectory.
\end{definition}

%In particular, the highest trajectory (Label 1) is the extreme of High Type Trajectory and the lowest trajectory (Label 3) is the extreme of Low Type Trajectory.

For any High Type Trajectory of Calvin, we show in Claim~\ref{claim-intersection} that $\pref{p}$ always intersects with $\phatt$ at some point $t$ after $t_{0}$ no matter what bit-flip pattern is chosen by Calvin (i.e., at point $t$, Bob's estimate $\phatt$ is equal to the actual amount of bit flips $\pref{p}$).
Moreover, by Claim~\ref{claim-ingap} and Claim~\ref{claim:p-tilde-t}, this implies a value $\tstar$ (the chunk end which falls immediately after the intersection point $t$ above) for which it is guaranteed that the remaining bit-flip power of Calvin is low in the sense that the fraction of bit-flips that Calvin can introduce in the right mega sub-codeword with respect to $\tstar$ is less than $\frac{1}{4}-\frac{\erate^2}{16}$. 
On the other hand, for any Low Type Trajectory of Calvin, we already know that $\pref{p}$ is approximately $\phatt$ at the point $t_{0}$ (they are both nearly 0). Thus we show in Claim~\ref{claim-zero-starting} that setting $\tstar$ to be equal to $t_0$ we are again guaranteed that the remaining bit-flip power of Calvin is low in the sense that the fraction of bit-flips that Calvin can introduce in the right mega sub-codeword with respect to $\tstar$ is less than $\frac{1}{4}-\frac{\erate^2}{16}$. Formally:

\begin{definition}
	\label{def:t-star}
	Let $\varepsilon>0$. Let $\myp\in\intvquar$ and $\echunk=\echval$. Let $\setchend=\choiceschunk$ and $t\in\setchend$. 
	%Then 
	\begin{enumerate}[label=(\roman*)]
		\item if $p_{t_{0}}<\tilde{p}_{0}$, $\tstar=t_{0}=\left\lceil\frac{1-4\myp}{\echunk}\right\rceil \chunk$.
		\item if $p_{t_{0}}\geq\tilde{p}_{0}$, $\tstar$ is the smallest value in $\setchend$ such that $p_{\tstar-\chunk}>\phat_{\tstar-\chunk}$ and $p_{\tstar}\leq\phat_{\tstar}$.
	\end{enumerate}
\end{definition}

\paragraph{Success of Bob's decoding:}
%As discussed, in the decoding process 
Bob starts decoding at position $t_{0}$ and continues to decode at subsequent chunk ends until a message is returned by the consistency decoder or until Bob reaches the end of the received word. Claim~\ref{thm:b-list-decodable} and Corollary~\ref{coro:b-list-decodable} (via the list decoding condition \eqref{eq:con-1}) guarantee that Bob in his first phase of decoding will always obtain a list of \sj{messages} of list size $L=\blistord$ from the list decoder no matter what position $t$ is currently being considered. 
The analysis in Claim~\ref{thm:b-list-decodable} and Corollary~\ref{coro:b-list-decodable} and in the claims to come is w.h.p. over our random code construction. 
Moreover, for any $t$, the energy bounding condition \eqref{eq:con-2} implies that, in the case of $\pref{p} \simeq \phatt$, the 
unused bit-flips left for Calvin are less than a $\frac{1}{4}-\frac{\erate^2}{16}$ fraction of the remaining part of the codeword.

We start by studying the case in which the current iteration of Bob satisfies $t=\tstar$ (which implies that $\pref{p} \simeq \phatt$).
In Claim~\ref{claim:good-1}, Claim~\ref{claim:good-2}, and Claim~\ref{claim:good-3} we show that if $t=\tstar$ Calvin's remaining bit-flip power is not sufficient to mislead the consistency decoder, and will allow unique decoding from the list \sj{of messages} Bob holds.
Namely, we show that with high probability over the secret random bits of Alice used in the encoding process, our code design guarantees that the only \sj{message} in our list that is consistent with the transmitted codeword is the one transmitted by Alice.

More precisely, consider the consistency checking phase of Bob in the iteration in which $t=\tstar$. In this iteration we know (via the energy bounding condition \eqref{eq:con-2}) that the number of unused bit-flips of Calvin is less than a $\frac{1}{4}-\frac{\erate^2}{16}$ fraction of the remaining part of the codeword. At this point in time, Bob holds a small list of \sj{messages} $\mlist$ that has been (implicitly) determined by Calvin, and via the consistency decoder wishes to find the unique message $m$ in the list that was transmitted. For any transmitted message $m$, as the list is small, we can guarantee that with high probability over our code design \sj{most of the right mega sub-codewords} corresponding to $m$ \sj{are} roughly of distance $\frac{n-t}{2}$ from any \sj{right mega sub-codeword of any other message} in the list $\mlist$, which in turn implies, given the bound on Calvin's remaining bit-flips, that decoding will succeed. However, this analysis is misleading as one must overcome the adversarial choice of $\mlist$ in establishing correct decoding. (We note that a naive use of the union bound does not suffice to overcome all potential lists $\mlist$.)

For successful decoding regardless of Calvin's adversarial behavior, we use the randomness in Alice's stochastic encoding (not known {\it a priori} to Calvin) and the fact that Calvin is causal. Recall that every message $m$ can be encoded into several codewords based on the randomness of Alice. Let $s_{left}$ and $s_{right}$ be the collection of Alice's random bits used up to and after position $\tstar$ respectively. When Calvin (perhaps partially) determines the list $\mlist$ we may assume that he has full knowledge of $s_{left}$.
However by his causal nature he has no knowledge regarding $s_{right}$. 
As the list $\mlist$ is obtained at position $\tstar$ by Bob, we may now 
take advantage of the fact that it  is independent of the randomness $s_{right}$ used by Alice. 
Specifically, instead of considering a single codeword in our analysis that corresponds to $m$ we consider the family of codewords that on one hand all share a specific $s_{left}$ (which corresponds to Calvin's view up to position $\tstar$) but have different $s_{right}$. From Calvin's perspective at position $\tstar$, all codewords in this family are equivalent and completely match his view so far. Using a family of codewords that are independent of $\mlist$ in our analysis, and allowing the decoding to fail on a small fraction of them, enables us to amplify the success rate of our decoding procedure to the extent that it can be used in the needed union bound.
Our full analysis is given in Claim~\ref{claim:good-1}, Claim~\ref{claim:good-2}, and Claim~\ref{claim:good-3}.

%
%
%the and the code is chosen at random, it follows that the codewords of the list have  we show Given the randomness nature of our code design, with high probability, the right mega sub-codewords of the codewords in the list obtained from the list decoding phase is roughly separated by a hamming distance of $\frac{n-t}{2}$. Since Calvin's remaining energy is bounded by $\frac{1}{4}-\frac{\erate^2}{16}$, he is unable to push a codeword toward any codeword in the list by a distance of $\frac{n-t}{4}$ so that he fails to mislead the consistency decoder. More importantly, it is worth noting that all the secrets in the right mega sub-codewords are obscure to Calvin at the point $t$, and therefore, it is impossible for him to construct a condense list of codewords in the way that the codewords in the list are distributed by a distance less than $\frac{n-t}{2}$. Thus, under the energy bounding condition, with high probability, the consistency decoder returns the correct message.

We now address the case $t \ne \tstar$ in Claim~\ref{claim:p-tilde-t}.
In this case, by previous discussions, it holds that we are in a 
High Type Trajectory of Calvin and that $\pref{p} > \phatt > \tilde{p}_{t}$.
When $t \ne \tstar$ we show that the decoding process of Bob will not return any codewords at all (as all \sj{messages} in the list will fail the consistency test). In this case, we continue with the next value of $t$ (the next chunk end). 

We summarize all the properties of our code in Claim~\ref{claim:prob-good-exist}. With those properties established, through Bob's iterative decoder 
%which applies list-decoding and consistency checking 
we show in Claim~\ref{claim:prob-decode} that Bob is able to correctly decode the transmitted message $m$ w.h.p. over the randomness of Alice. Finally, in Theorem~\ref{thm:capacity} we show that the channel capacity $C_{p}$ claimed is indeed achievable. We depict the flow of our claims, corollaries and theorems in Figure~\ref{fig:proofs-organization}.

{{\noindent {\bf Remark:}} The scenario wherein Calvin has $n\epsilon$ lookahead can also be handled via the codes above. Roughly, if we back off in our rate by $\epsilon$ the trajectory $\phatt$ gets shifted to the left by $n\epsilon$. We then ``sacrifice'' $n\epsilon$ bits to Calvin by demanding that a more stringent energy-bounding condition be satisfied, in which the block length of the second part (succeeding $\tstar$) is reduced by $n\epsilon$. With these tweaks, the remainder of the analysis of the $n\epsilon$-lookahead codes is identical to that of the causal codes discussed above.}

%
%
%\zt{I think I understand the concern. Could it just be our way of saying the definition (and implictly the left mega sub-codeword is fixed and the right one is varying...)
%	
%	
%	If not, personally I prefer the second way to fix the problem. For the notation of a left mega sub-codewrod, use $\mathbf{x}^{k}$, $\mathbf{x}_{k}$, or simply $\megacodw$?
%	
%	LIST things needed to modify accordingly:
%	
%	1.definitions of goodness\\
%	2.claims of goodness\\
%	3.appendix summary\\
%	4.figure 2\\
%	5.same things for erasure case\\
%	6....?\\
%	
%	UPDATE a fixed version of the definition could be 
%	
%	\begin{definition}
%		Let $k=\frac{t}{\chunk}$. Let $\mathbf{x}^{k}=\megacodw$ be a left mega sub-codeword and $\mathbf{x}$ be a full codeword corresponding to $\mathbf{x}^{k}$. A right mega sub-code $\rmegacode$ is \textit{good} with respect to a list $\mlist$ of codewords, a {\bf codeword $\mathbf{x}\notin\mlist$} and a sequence of $l=\chunknum-\frac{t}{\chunk}$ secrets $\left( s_{k+1},s_{k+2},\cdots,s_{\chunknum}\right) $ if the right mega sub-codeword $\rmegacodw$ of the codeword $\mathbf{x}$ is of distance more than $\frac{n-t}{2}-\frac{(n-t)\erate^2}{8}$ from the list $\mlist$.
%	\end{definition}
%	}

\vspace{-0.1in}
\section{Conclusion}
\vspace{-0.1in}

%\ml{Add intuition for $C_p^{flip}$ ...}
In this work we obtain the capacity $C^{flip}_{p}$ for the causal bit-flip adversarial channel and $C^{erase}_p$ for the causal erasure adversarial channel. We believe that for arbitrary $q>2$, capacities of the $q$-ary causal adversarial symbol error and erasure channels can be derived. Specifically, the outer bounds for the $q$-ary problems follow from using the $q$-ary Plotkin bound~\cite{blake1976introduction} replacing the binary Plotkin bound used in the outer bound techniques of~\cite{dey2012improved} and~\cite{bassily2014causal}. The achievability strategies in this paper for $q=2$ get modified by replacing binary list-decoding and energy-bounding conditions with $q$-ary versions, leading to codes with rates achieving the outer bounds mentioned above for the $q$-ary problems. Due to lack of space we do not present details in this abstract.
Given the significant difference between the capacities of causal adversaries (characterized in this work) and the $n\epsilon$-delayed adversaries considered in~\cite{dey2010coding} (with capacities equaling the corresponding random coding capacities), one open question that we do not address in this work is the question of characterizing the capacity of an adversary that is ``nearly'' causal, {\it i.e.}, a delay of say some {\it constant} number of bits, or even just $one$ bit. Another promising direction for future research is whether the techniques of~\cite{guruswami2010codes} to construct computationally efficient capacity-achieving codes for oblivious adversaries can be modified to construct corresponding capacity-achieving codes for causal adversaries.
%It is worth noting that $C^{flip}_{p}$ matches the upper bound presented in \cite{dey2013upper} and that our code analysis has an intuitive interpretation in terms of the study in \cite{dey2013upper}.
%Specifically, in \cite{dey2013upper} the ``babble and push'' adversarial strategy is studied where Calvin first carefully injects a small amount of random noise in his first phase of action (referred to as the ``babble'' phase) while preserving adequate bit-flips to ``push'' the transmitted codeword to an alternative codeword in the second phase. In our analysis, the decoding reference trajectory $\phatt$ selected is zero for $t<n(1-4\myp)$, which is imitating Calvin's ``babbling'' phase, and goes up quickly after $t\geq n(1-4\myp)$, which is corresponding to Calvin's ``pushing'' strategy. Therefore, the selected trajectory $\phatt$ is approximately Calvin's best attacking strategy. More importantly, the selected trajectory $\phatt$ also enables Bob to correctly decode no matter whether Calvin is adopting the ``babble and push'' strategy or not. In other words, no matter what adversarial strategy Calvin chooses the upper bound in \cite{dey2013upper} is achievable, and thus, we obtain capacity $C^{flip}_{p}$.

\newpage
\bibliographystyle{unsrt}
\bibliography{./refs}

\newpage
\begin{appendices}

\section{Bit-flip channel}
\label{app:flip}
\begin{figure}[p]
	\centering
	\includegraphics[scale=0.8]{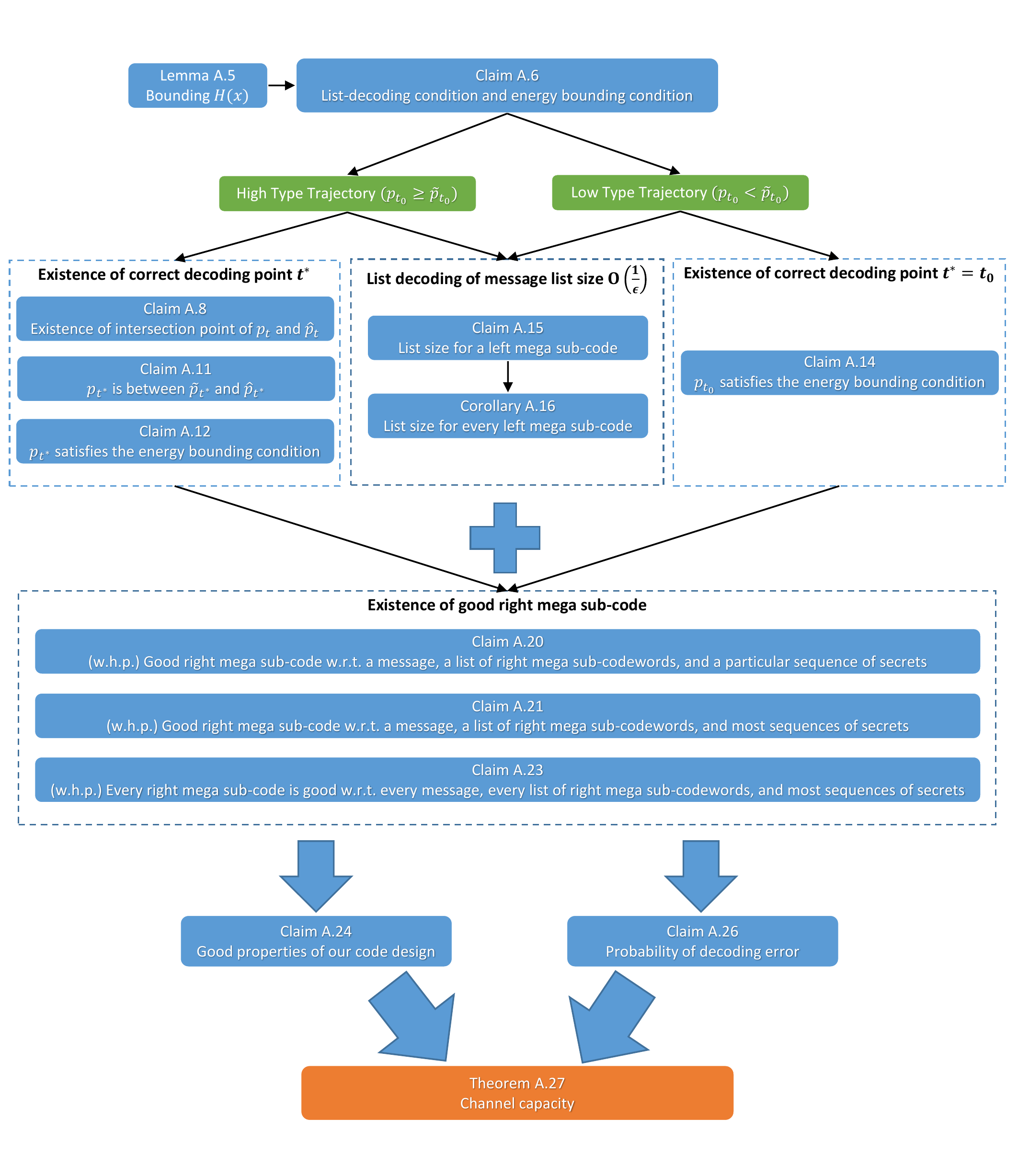}
	\caption{Organization of our claims, corollaries and theorems}
	\label{fig:proofs-organization}
\end{figure}

We start by summarizing several definitions and claims. The detailed presentations of the definitions and claims are followed by the summary.
\begin{enumerate}[label=\textbf{\arabic*}.]
	\item {\bf Preliminary definitions and technical claims}
		\begin{itemize}
			\item Definition~\ref{def:p-bar-t}: Defines $\pbar_t$ for deriving the definition of the decoding reference trajectory $\phatt$.
			\item Definition~\ref{def:p-hat-t}: Defines the decoding reference trajectory $\phatt$, which is a revision of the definition given in Section~\ref{sec:model}.
			\item Lemma~\ref{lemma-1}: A technical lemma which gives a certain upper bound on the binary entropy function.
		\end{itemize}
	\item {\bf The list decoding and energy bounding properties}
		\begin{itemize}
			\item Claim~\ref{claim:1}: This is a central claim which shows that the decoding reference trajectory $\phatt$ satisfies the list-decoding condition and the energy bounding condition.
		\end{itemize}
	\item {\bf Establishing the existence of correct decoding point}
		\begin{itemize}
			\item Claim~\ref{claim-intersection}: Calvin's trajectory $\pref{p}$ intersects with the decoding reference trajectory $\phatt$ no later than the second to last chunk.
			\item Definition~\ref{def:p-tilde-t}: Defines the energy bounding benchmarking trajectory $\tilde{p}_{t}$, which delimit the smallest value of $\pref{p}$ that meets with the energy bounding condition.
			\item Claim~\ref{claim-ingap}: For any High Type Trajectory $\pref{p}$, the value of $\pref{p}$ at the chunk end immediately after the intersection of the decoding reference trajectory $\phatt$ with $\pref{p}$ satisfies the energy bounding condition. 
			\item Claim~\ref{claim:p-tilde-t}: If $\pref{p}$ is larger than $\tilde{p}_{t}$ at point $t$, then $\pref{p}$ satisfies the energy bounding condition.
			\item Claim~\ref{claim-zero-starting}: At point $t_{0}$, if $p_{t_{0}}$ is approximately $\phat_{t_{0}}$ then it satisfies the energy bounding condition.
		\end{itemize}
	\item {\bf List decoding properties}
		\begin{itemize}
			\item Claim~\ref{thm:b-list-decodable}: A left mega sub-code can be list decoded to a list \zt{of messages }of size $\blistord$ with high probability.
			\item Corollary~\ref{coro:b-list-decodable}: Every left mega sub-code can be list decoded to a list \zt{of messages }of size $\blistord$ with high probability.
		\end{itemize}
	\item {\bf Utilizing the energy bounding condition}
		\begin{itemize}
			\item Definition~\ref{def:dist}: Defines the distance between a right mega sub-codeword and a list of \sj{right mega sub-}codewords.
			\item Definition~\ref{def:good-1}: Defines certain {\em goodness} properties of a right mega sub-code with respect to \sj{a message}, %\zx{the left mega sub-codeword of }a given codeword, 
a list \sj{of right mega sub-codewords (of messages excluding the transmitted message}), and a sequence of secrets.
			\item Definition~\ref{def:good-2}: Defines $\sigma$-goodness property of a right mega sub-code with respect to a message, %\zx{the left mega sub-codeword of} a given codeword, 
			a list \sj{of right mega sub-codewords (of messages excluding the transmitted message}), and most sequences of secrets.
			\item Claim~\ref{claim:good-1}: A right mega sub-code is good with respect to a message, % \zx{the left mega sub-codeword of }a given codeword, 
			a list \sj{of right mega sub-codewords (of messages excluding the transmitted message}), and a sequence of secrets. 			\item Claim~\ref{claim:good-2}: A right mega sub-code is $\sigma$-good with respect to \sj{a message} \ml{and a list} 
			%\zx{the left mega sub-codeword of }a given codeword, 
			\sj{of right mega sub-codewords (of messages excluding the transmitted message}.)
			%, and most sequences of secrets with high probability.
		%\item Corollary~\ref{coro:good-2}: A right mega sub-code is $\sigma$-good with respect to \sj{a message} \ml{and a list} 
			%\zx{the left mega sub-codeword of }a given codeword, 
%\sj{of right mega sub-codewords (of messages excluding the transmitted message}).
%a list \sj{of codewords} (excluding the \sj{transmitted} codeword), and most sequences of secrets with high probability.
			\item Claim~\ref{claim:good-3}: A right mega sub-code is $\sigma$-good with respect to 
			\ml{every transmitted message and every list of right mega sub-codewords (of messages excluding the transmitted message).}
			%\zx{the left mega sub-codeword of }every codeword and every list.
		\end{itemize}
	\item {\bf Summary and proof of Theorem~\ref{the:main:flip}}
		\begin{itemize}
			\item Claim~\ref{claim:prob-good-exist}: With high probability our code $\mcode$ possesses the needed properties.
			\item Claim~\ref{claim:prob-decode}: With high probability Bob succeeds in decoding.
			\item Theorem~\ref{thm:capacity}: Rephrasing of Theorem~\ref{the:main:flip} (channel capacity).
		\end{itemize}
\end{enumerate}
	
\subsection{Preliminary definitions and technical claims}

Throughout our analysis, Calvin can flip at most  fraction $\ppm$ of the bits of the transmitted codeword.
Let $\myp=\ppm+\frac{\erate^2}{16}$. The parameter $\myp$ is considered in several of the derivations to follow.
We take $\epsilon>0$ and $\theta>0$ to be small constants.
All parameters and their relations appear in Table~\ref{tab:params}.

\begin{definition}[$\pbar_{t}$]
	\label{def:p-bar-t}
	Let $\myp\in\intvquar$. Let a codeword of length $n$ consist of $\chunknum$ chunks of sub-codewords. Let $\setchend=\choiceschunk$ and $\leftcw\in\setchend$. Then $\pbar_{t}$ is defined as
	\begin{align}
	\label{eq:p-bar-t}
	\pbar_t=\myp-\frac{1}{4}+\frac{t}{4n}
	\end{align} 
\end{definition}

\begin{remark}
	For $t\geq n(1-4\myp)$ we have $\pbar_{t}\geq 0$. % and we plot $\pbar_{t}$ in Figure~\ref{fig:p-bar-t}.
	%\begin{figure}[htbp]
	%	\centering
	%	\includegraphics[scale=1]{p-bar-t.pdf}
	%	\caption{The trajectory $\pbar_{t}$ with $\myp=\frac{1}{8}$ and $n=40000$ for $t\in[20000,40000]$}
	%	\label{fig:p-bar-t}
	%\end{figure}
\end{remark}

\begin{definition}[Decoding Reference Trajectory $\phatt$]
	\label{def:p-hat-t}
	Let $\erate>0$ and $\myp\in\intvquar$. Let $\myfun[t]=1-4(\myp-\pbar_{t})$ for $\pbar_{t}\in\left[0,\myp\right]$. Then for $t\in\left[ n(1-4\myp),n\right] $, $\phatt$ is defined as 
	\begin{align}
	\label{eq:p-hat-t}
	\phatt=\frac{\pbar_{t}}{\myfun[t]}+\efrac
	\end{align} For $t<n(1-4\myp)$, $\phatt=0$.
\end{definition}

\begin{remark}
	\label{remark-p-hat-t}
	%Let $\myfun=1-4(\myp-\pbar)$ for $\pbar\in[0,\myp]$ so we have $\myfun\in[1-4\myp,1]$.	
	From Definition~\ref{def:p-bar-t}, for $\leftcw\in\left[n(1-4\myp),n\right]$ we have $\pbar_{t}\in\left[ 0,\myp\right] $ and
	\begin{align}
		\myfun[t] &= 1-4(\myp-\pbar_{t}) \nonumber\\
		&= 1-4\left(\myp-\myp+\frac{1}{4}-\frac{t}{4n}\right) \nonumber\\
		&= \frac{t}{n}
	\end{align}
	
	By substituting $\pbar_{t}$ and $\myfun[t]$ into \eqref{eq:p-hat-t}, $\phatt$ can be expressed as a function of $t$, i.e., for $t\in\left[ n(1-4\myp),n\right] $ we have
	\begin{align}
		\phatt&=\frac{\pbar_{t}}{\myfun[t]}+\efrac\nonumber\\
		&=\frac{n\pbar_{t}}{t}+\frac{n^2\erate^2}{16t^2}\nonumber\\
		&=\frac{n}{t}\left( \myp-\frac{1}{4}+\frac{t}{4n}\right) +\frac{n^2\erate^2}{16t^2}\nonumber\\
		&=\frac{n^2\erate^2}{16}\cdot\frac{1}{t^2}-\frac{n(1-4\myp)}{4}\cdot\frac{1}{t}+\frac{1}{4}\label{eq:p-hat-t-2}
	\end{align}
	In Figure~\ref{fig:p-hat-t} we plot $\phatt$ as a function of $t$ for a particular choice of $n,\myp$ and $\erate$.
	\begin{figure}[htb]
		\centering
		\includegraphics[scale=0.8]{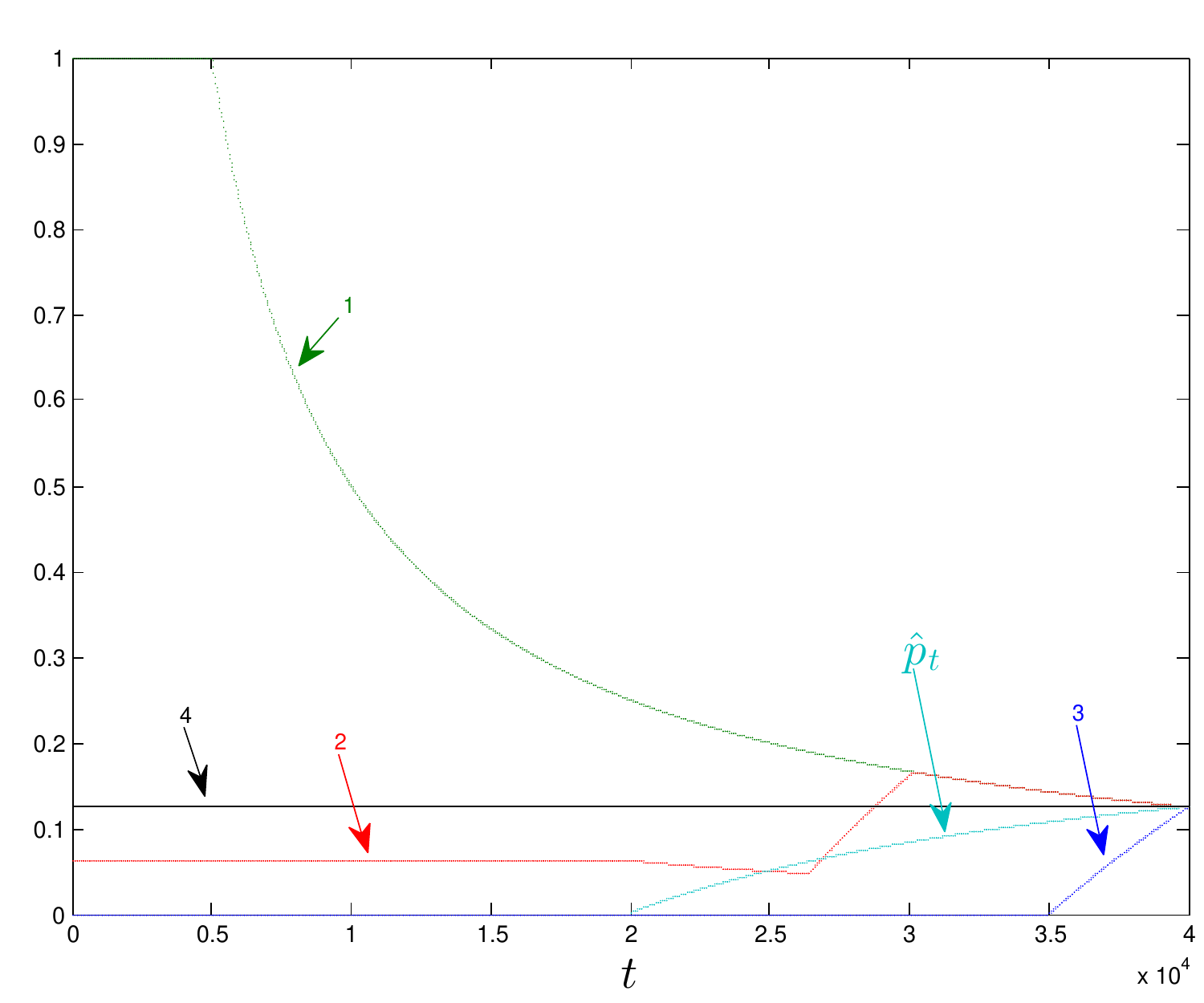}
		\caption[The decoding reference trajectory $\phatt$]{The decoding reference trajectory $\phatt$ as a function of $t$ when $n=40000,\myp=\frac{1}{8}$ and $\erate=0.08$. Four possible trajectories of Calvin are shown in comparison: $1.$ Squander all the bit-flips at the beginning; $2.$ A typical trajectory of Calvin; $3.$ Reserve all the bit-flips for the end; $4.$ Uniformly distribute all bit-flips to each chunk.}
		\label{fig:p-hat-t}
	\end{figure}
\end{remark}

\begin{lemma}
	\label{lemma-1}
	Let $q\in\left[ 0,\frac{1}{2}\right) $ and $\gamma\in\left( 0,\frac{1}{2}\right) $. Then
	\begin{align*}
		H(q+\gamma)<H(q)+2\gamma\log\frac{1}{\gamma}
	\end{align*}
\end{lemma}

\begin{proof}
	To prove the lemma, we first show that
	\begin{align*}
		\log (1-x)+2x\geq 0
	\end{align*}
	for $x\in\left[ 0,\frac{1}{2}\right) $ and
	\begin{align*}
		\log (1-x)+2x<0
	\end{align*} 
	for $x\in\left( \frac{1}{2},1\right] $. 
	
	Let $f(x)=\log (1-x)+2x$ where $x\in\left[ 0,1\right] $. Then $f^{\prime}(x)=2-\frac{1}{(1-x)\ln 2}$. Solving $f^{\prime}(x)=0$, we obtain $x=1-\frac{1}{2\ln 2}<\frac{1}{2}$. Then for $x\in\left( 0,1-\frac{1}{2\ln 2}\right) $, $f^{\prime}(x)>0$ and for $x\in\left( 1-\frac{1}{2\ln 2},1\right) $, $f^{\prime}(x)<0$.
	
	Since $f(0)=f\left( \frac{1}{2}\right) =0$, then for $x\in\left[ 0,\frac{1}{2}\right) $ we have $\log (1-x)+2x\geq0$, and therefore,
	\begin{align}
		\label{eq:le-1}
		\log \frac{1}{1-x}\leq 2x
	\end{align}
	On the other hand, for $x\in\left( \frac{1}{2},1\right] $ we have $\log (1-x)+2x<f\left( \frac{1}{2}\right) =0$, and thus, replacing $(1-x)$ by $x$ we have for $x\in\left[ 0,\frac{1}{2}\right) $
	\begin{align}
		\label{eq:le-11}
		2(1-x)<\log\frac{1}{x}
	\end{align}
	
	Since the binary entropy function $H(q)$ is concave, namely, the second derivative of $H(q)$ is negative for $q\in\left( 0,1\right) $, then
	\begin{align*}
		\frac{H(q+\gamma)-H(q)}{q+\gamma-q}<\frac{H(\gamma)-H(0)}{\gamma-0}
	\end{align*}
	Therefore, we have
	\begin{align}
		H(q+\gamma)-H(q) &<H(\gamma)-H(0)\nonumber\\
		&= \gamma\log \frac{1}{\gamma}+(1-\gamma)\log \frac{1}{1-\gamma}\nonumber\\
		&\leq \gamma\log \frac{1}{\gamma}+(1-\gamma)2\gamma\label{eq:le-2}\\
		&< \gamma\log \frac{1}{\gamma}+\gamma\log \frac{1}{\gamma}\label{eq:le-3}\\
		&= 2\gamma\log \frac{1}{\gamma}\nonumber
	\end{align}
	where \eqref{eq:le-2} follows by \eqref{eq:le-1} and \eqref{eq:le-3} follows by \eqref{eq:le-11}.
\end{proof}

\subsection{The list decoding and energy bounding properties}

\begin{claim}
	\label{claim:1}
	Let $\erate>0$ and $\myp\in\intvquar$. Let $C_{\myp}=\myopt$ and $R=C_{\myp}-\erate$ where $\myfun=1-4(\myp-\pbar)$ and $\pbar\in\left[0,\myp\right]$. Let a codeword of length $n$ consist of $\chunknum$ chunks of sub-codewords where $\echunk=\echval$. Let $\setchend=\choiceschunk$. Then $\forall\leftcw\geq n(1-4\myp)$ and $\leftcw\in\setchend$ there exists $\phatt\in\intvone$ such that
	\begin{align}
		\label{eq:con-1}
		\leftcw\left( 1-H(\phatt)\right)-\frac{n\erate}{2} \geq nR
	\end{align}
	and that 
	\begin{align}
		\label{eq:con-2}
		n\myp-\leftcw\phatt+\elfrac\leq\frac{n-\leftcw}{4}
	\end{align}
\end{claim}

\begin{proof}
	From Remark~\ref{remark-p-hat-t}, for $\leftcw\in\left[n(1-4\myp),n-\chunk\right]$ we have $\myfun[t]=\frac{t}{n}$, and it follows that $n\myfun[t]=t=k\chunk$ for some integer $k$. Substituting $t=n\myfun[t]$ into \eqref{eq:con-1} and dividing both sides by $n$, we obtain
	\begin{align}
		\label{eq:con-1-reduced}
		\myfun[t]\left( 1-H(\phatt)\right) -\elist\geq R
	\end{align}
	
	Next, from Definition~\ref{def:p-hat-t}, substituting \eqref{eq:p-hat-t} into the \textit{left hand side} (LHS) of \eqref{eq:con-1-reduced} we have
	\begin{align}
		\lefteqn{\myfun[t]\left( 1-H\left( \frac{\pbar_{t}}{\myfun[t]}+\efrac\right) \right) -\elist}\nonumber\\
		%& \myfun[t]\left( 1-H\left( \frac{\pbar_{t}}{\myfun[t]}+\efrac\right) \right) -\elist\nonumber\\
		& &>&\myfun[t]\left( 1-H\left( \frac{\pbar_{t}}{\myfun[t]}\right) -2\sqrt{\efrac}\right) -\elist \label{eq:claim-1-left-larger-1}\\
		& &=& \myfun[t]\left( 1-H\left( \frac{\pbar_{t}}{\myfun[t]}\right) -\frac{\erate}{2\myfun[t]}\right) -\elist\nonumber\\
		& &=& \myfun[t]\left( 1-H\left( \frac{\pbar_{t}}{\myfun[t]}\right) \right) -\erate \nonumber
	\end{align}
	where \eqref{eq:claim-1-left-larger-1} follows from a variant of Lemma~\ref{lemma-1}. More precisely, for $\gamma<\frac{1}{16}$ note that $\gamma\log\frac{1}{\gamma}<\sqrt{\gamma}$, and therefore, $H(q+\gamma)<H(q)+2\sqrt{\gamma}$.
	
	Let $\pbar_{0}$ be the optimum $\pbar$ that minimizes the expression $\mycapa[]$, namely, $$C_{\myp}=\mycapa$$ then
	\begin{align}
		\myfun[t]\left( 1-H\left( \frac{\pbar_{t}}{\myfun[t]}\right) \right) -\erate\geq\mycapa-\erate=R
	\end{align}
	and thus,
	\begin{align}
		\myfun[t]\left( 1-H\left( \frac{\pbar_{t}}{\myfun[t]}+\efrac\right) \right) -\elist\geq R
	\end{align}
	
	Thus far we have satisfied condition~\eqref{eq:con-1} in our claim. To see condition~\eqref{eq:con-2}, we substitute \eqref{eq:p-hat-t} into the LHS of \eqref{eq:con-2} and note that $\myfun[t]=\frac{t}{n}$ we have
	\begin{align}
		n\myp-\leftcw\left( \frac{\pbar_{t}}{\myfun[t]}+\efrac\right) +\elfrac & = n\myp-\leftcw\left( \frac{n\pbar_{t}}{\leftcw}+\frac{n^2\erate^2}{16\leftcw^2}\right) +\elfrac\nonumber\\
		& = n\myp-n\pbar_{t}-\frac{n^2\erate^2}{16\leftcw}+\elfrac\nonumber\\
		& < n\myp-n\pbar_{t}\nonumber\\
		& = n\myp-n\left( \myp-\frac{1}{4}+\frac{t}{4n}\right)\label{eq:con-2-less-1}\\
		& = \frac{n-t}{4}\nonumber
	\end{align}
	where \eqref{eq:con-2-less-1} follows by substituting \eqref{eq:p-bar-t} into $\pbar_{t}$.
\end{proof}

\begin{remark}
	Condition~\eqref{eq:con-1} in Claim~\ref{claim:1} corresponds to the list-decodability of our code.
	We will refer to it as the \textit{list-decoding condition}.
	We refer to Condition~\eqref{eq:con-2} as the \textit{energy bounding condition} as it bounds the number of bit-flips left for the adversary to impose on the remaining codeword. In Figure~\ref{fig:p-region}, the uppermost solid curve corresponds to the setting of $\phatt$ for which the list-decoding condition is tight; the lowermost solid curve corresponds to the setting of $\phatt$ for which the energy bounding condition is tight. Any curve between these two solid curves satisfies both the list-decoding condition and energy bounding condition. As is shown in Figure~\ref{fig:p-region}, the decoding reference trajectory $\phatt$ we selected is a bit above the lowest solid curve.  
	\begin{figure}[htbp]
		\centering
		\includegraphics[scale=0.8]{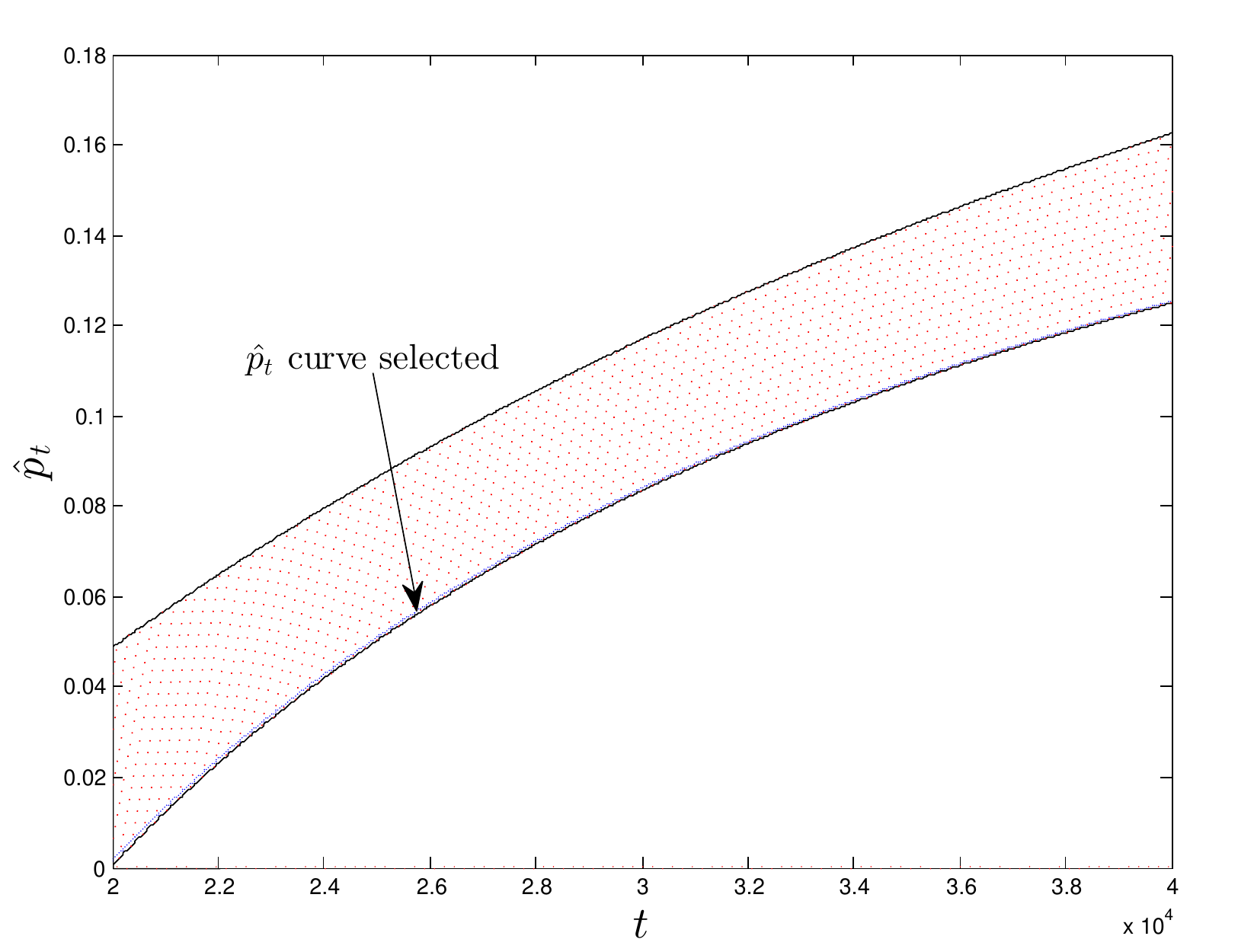}
		\caption{The region encompassed by the list-decoding and the energy bounding condition \zx{for the case where $n=40000$, $\myp=\frac{1}{8}$, $\erate=0.08$ and $t\in[20000,40000]$}.}
		\label{fig:p-region}
	\end{figure}
\end{remark}

\subsection{Establishing the existence of correct decoding point}
First we show that $\phatt$ must eventually be greater than $\pref{p}$.

\begin{claim}
	\label{claim-intersection}
	Let $\erate>0$ and $\myp\in\intvquar$. Let a codeword of length $n$ consist of $\chunknum$ chunks of sub-codeword where $\echunk=\echval$. Then we have
	\begin{align}
		\phat_{(n-\chunk)}(n-\chunk)\geq n\myp
	\end{align}
\end{claim}

\begin{proof}
	From \eqref{eq:p-bar-t} we have $\pbar_t=\myp-\frac{1}{4}+\frac{t}{4n}$. Here we have $t=n-\chunk$ so 
	\begin{align}
		\pbar_{n-\chunk}&=\myp-\frac{1}{4}+\frac{n-\chunk}{4n} \nonumber\\
		&=\myp-\frac{\echunk}{4}
	\end{align}
	From Remark~\ref{remark-p-hat-t}, we have $\myfun[t]=\frac{t}{n}$, and therefore, $\alpha\left( \myp,\pbar_{(n-\chunk)}\right) =\frac{n-\chunk}{n}=1-\echunk$. Then
	\begin{align}
		\phat_{(n-\chunk)} &=\frac{\pbar_{(n-\chunk)}}{\myfun[(n-\chunk)]}+\frac{\erate^2}{16\alpha^{2}\left( \myp,\pbar_{(n-\chunk)}\right) }\nonumber\\
		&=\frac{\pbar_{(n-\chunk)}}{1-\echunk}+\frac{\erate^2}{16\left( 1-\echunk\right) ^2}\nonumber\\
		&=\frac{\myp-\frac{\echunk}{4}}{1-\echunk}+\frac{\erate^2}{16\left( 1-\echunk\right) ^2}
	\end{align}
	Hence, we have
	\begin{align}
		\phat_{(n-\chunk)}(n-\chunk) &=\phat_{(n-\chunk)}(1-\echunk)n\nonumber\\
		&=\left( \myp-\frac{\echunk}{4}+\frac{\erate^2}{16\left( 1-\echunk\right) }\right) n\nonumber\\
		&>n\myp-\frac{n\erate^2}{16}+\frac{n\erate^2}{16-4\erate^2}\label{eq:claim-2-left-less-1}\\
		&>n\myp
	\end{align}
	where \eqref{eq:claim-2-left-less-1} follows by $\echunk=\echval<\frac{\erate^2}{4}$.
\end{proof}

\begin{definition}[Energy Bounding Benchmarking Trajectory $\tilde{p}_{t}$]
	\label{def:p-tilde-t}
	Let $\erate>0$ and $\myp\in\intvquar$. Let a codeword of length $n$ consist of $\chunknum$ chunks of sub-codewords. Let $\setchend=\choiceschunk$ and $\leftcw\in\setchend$. Let $\myfun[t]=1-4(\myp-\pbar_{t})$ for $\pbar_{t}\in\left[0,\myp\right]$. Then for $t\in\left[ n(1-4\myp),n\right] $, $\tilde{p}_{t}$ is defined as 
	\begin{align}
		\label{eq:p-tilde-t}
		\tilde{p}_{t}=\frac{\pbar_t}{\myfun[t]}+\frac{(n-t)\erate^2}{16t}
	\end{align} For $t<n(1-4\myp)$, $\tilde{p}_{t}=0$.
\end{definition}

\begin{remark}
	We plot $\tilde{p}_{t}$ in comparison with $\phatt$ in Figure~\ref{fig:p-tilde-t}. As is shown in Figure~\ref{fig:p-tilde-t}, the difference between $\tilde{p}_{t}$ and $\phatt$ is very small, which is analyzed in Claim~\ref{claim-ingap}.
	\begin{figure}[htbp]
		\centering
		\includegraphics[scale=0.8]{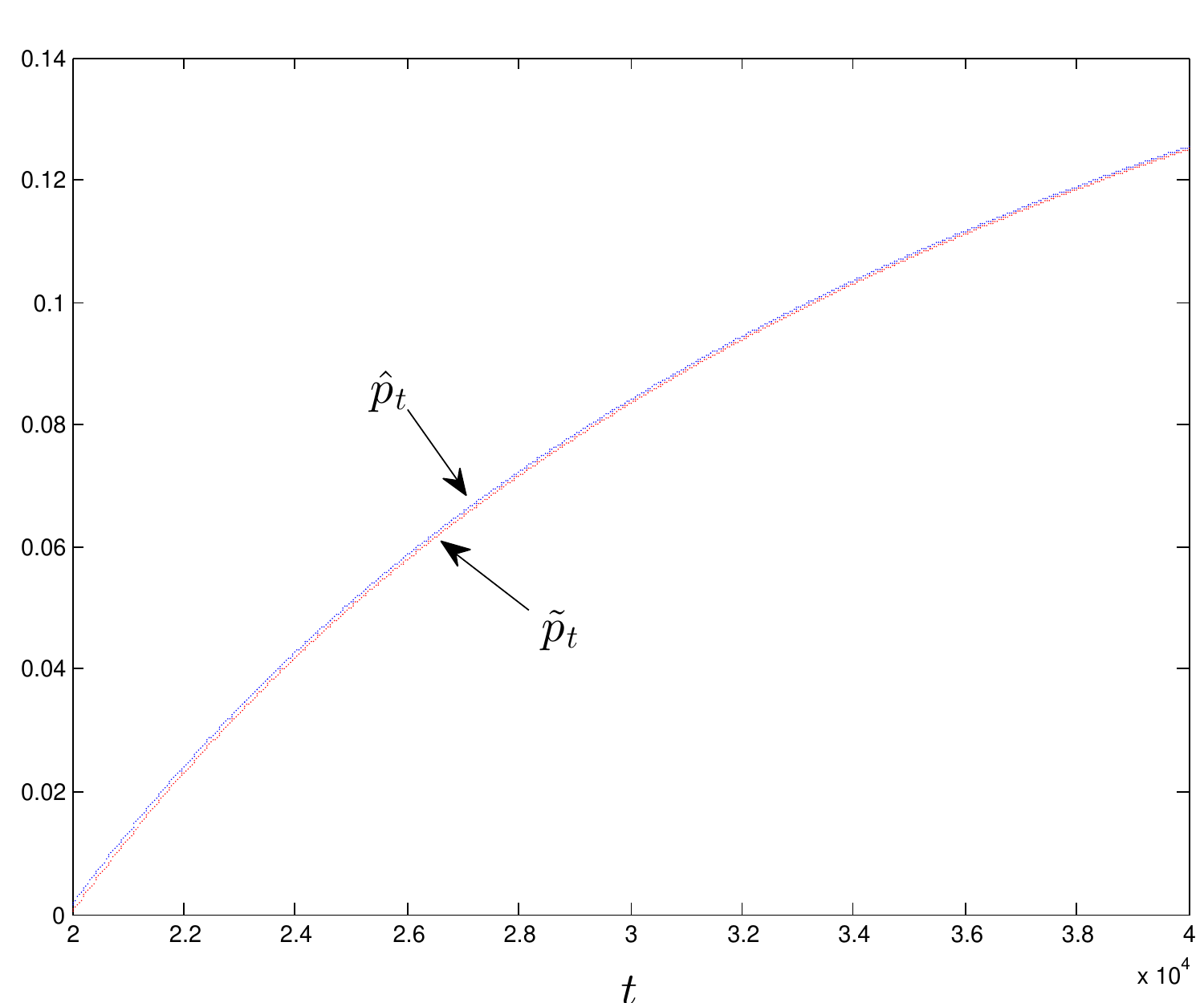}
		\caption{The trajectory of $\tilde{p}_{t}$ in comparison with $\phatt$ for $n=40000$, $\myp=\frac{1}{8}$, $\erate=0.08$ and $t\in[20000,40000]$.}
		\label{fig:p-tilde-t}
	\end{figure}
\end{remark}

\begin{claim}
	\label{claim-ingap}
	Let $\erate>0$ and $\myp\in\intvquar$. Let a codeword of length $n$ consist of $\chunknum$ chunks of sub-codeword where $\echunk=\echval$. Let $\setchend=\choiceschunk$. Let $\tstar\in\setchend$ and $\tstar\in\left[ n(1-4\myp),n-\chunk\right] $. If $p_{\tstar-\chunk}>\phat_{\tstar-\chunk}$ and $p_{\tstar}\leq\phat_{\tstar}$, then we have $p_{\tstar}>\tilde{p}_{\tstar}$.
\end{claim}

\begin{proof}
	Let $\tstar=\kstar\chunk$ where integer $\kstar\in\left[ \frac{1-4\myp}{\echunk},\frac{1-\echunk}{\echunk}\right] $. Since $p_{\tstar-\chunk}>\phat_{\tstar-\chunk}$, we have
	\begin{align}
		p_{\tstar}&\geq\frac{(\tstar-\chunk)p_{\tstar-\chunk}}{\tstar} \nonumber\\
		&>\frac{(\tstar-\chunk)\phat_{\tstar-\chunk}}{\tstar} \nonumber\\
		&=\frac{(\kstar\chunk-\chunk)\phat_{\tstar-\chunk}}{\kstar\chunk} \nonumber\\
		&=\frac{(\kstar-1)\phat_{\tstar-\chunk}}{\kstar}
	\end{align}
	
	From \eqref{eq:p-hat-t-2} we have $\phat_{\tstar}=\frac{n^2\erate^2}{16}\cdot\frac{1}{{\tstar}^2}-\frac{n(1-4\myp)}{4}\cdot\frac{1}{\tstar}+\frac{1}{4}$. Then the gap between $\phat_{\tstar}$ and $\frac{\kstar-1}{\kstar}\phat_{\tstar-\chunk}$ can be determined to be
	\begin{align}
		\lefteqn{\phat_{\tstar}-\frac{(\kstar-1)\phat_{\tstar-\chunk}}{\kstar}}\nonumber\\
		%&\phat_{\tstar}-\frac{(\kstar-1)\phat_{\tstar-\chunk}}{\kstar} \\
		& &= &\frac{n^2\erate^2}{16(\kstar\chunk)^2}-\frac{n(1-4\myp)}{4(\kstar\chunk)}+\frac{1}{4}-\frac{\kstar-1}{\kstar}\left( \frac{n^2\erate^2}{16(\kstar\chunk-\chunk)^2}-\frac{n(1-4\myp)}{4(\kstar\chunk-\chunk)}+\frac{1}{4}\right) \nonumber\\
		& &= &\frac{\erate^2}{16(\kstar\echunk)^2}-\frac{\erate^2}{16\kstar(\kstar-1)\echunk^2}+\frac{1}{4\kstar}\nonumber\\
		& &< &\frac{\erate^2}{16(\kstar\echunk)^2}-\frac{\erate^2}{16\kstar\echunk}+\frac{1}{4\kstar} \label{eq:claim-gap-0}\\
		& &< &\frac{\erate^2}{16(\kstar\echunk)^2}-\frac{\erate^2}{16\kstar\echunk}+\frac{\echunk}{4(1-4\myp)} \label{eq:claim-gap-1}\\
		& &< &\frac{\erate^2}{16(\kstar\echunk)^2}-\frac{\erate^2}{16\kstar\echunk}+\frac{\erate^2}{16} \label{eq:claim-gap-2}
	\end{align}
	where \eqref{eq:claim-gap-0} follows by $\kstar(\kstar-1)\echunk^2<(\kstar\echunk)^2\leq \kstar\echunk$, \eqref{eq:claim-gap-1} follows by $\kstar\geq\frac{1-4\myp}{\echunk}$, and \eqref{eq:claim-gap-2} follows by $\echunk=\echval$.
	
	From \eqref{eq:p-hat-t} and \eqref{eq:p-tilde-t} we have $\phat_{\tstar}=\frac{\pbar_{\tstar}}{\myfun[\tstar]}+\frac{\erate^2}{16\alpha^2\left(\myp,\pbar_{\tstar}\right)}$ and $\tilde{p}_{\tstar}=\frac{\pbar_{\tstar}}{\myfun[\tstar]}+\frac{(n-\tstar)\erate^2}{16\tstar}$. Hence, we have 
	\begin{align}
		\phat_{\tstar}-\tilde{p}_{\tstar}&=\frac{n^2\erate^2}{16(\kstar\chunk)^2}-\frac{(n-\kstar\chunk)\erate^2}{16\kstar\chunk}\nonumber\nonumber\\
		&=\frac{\erate^2}{16(\kstar\echunk)^2}-\frac{\erate^2}{16\kstar\echunk}+\frac{\erate^2}{16}\\
		&>\phat_{\tstar}-\frac{(\kstar-1)\phat_{\tstar-\chunk}}{\kstar}\\
		&>\phat_{\tstar}-p_{\tstar}
	\end{align}
	Thus, we have $p_{\tstar}>\tilde{p}_{\tstar}$.
\end{proof}

\begin{claim}
	\label{claim:p-tilde-t}
	Let $\erate>0$ and $\myp\in\intvquar$. Let a codeword of length $n$ consist of $\chunknum$ chunks of sub-codewords. Let $\setchend=\choiceschunk$ and $\leftcw\in\setchend$. Let $\suff{p}$ be the portion of bit-flips in the right mega sub-codeword with respect to position $\leftcw$. Then if $p_t>\tilde{p}_t$ we have
	\begin{align}
		\suff{p}<\frac{1}{4}-\frac{\erate^2}{16}
	\end{align}
\end{claim}

\begin{proof}
	From Remark~\ref{remark-p-hat-t} and Definition~\ref{def:p-bar-t}, we have $\myfun[t]=\frac{t}{n}$ and $\pbar_{t}=\myp-\frac{1}{4}+\frac{t}{4n}$. Then 
	\begin{align}
	\tilde{p}_t&=\frac{\pbar_t}{\myfun[t]}+\frac{(n-t)\erate^2}{16t} \nonumber\\
	&=\frac{n\myp}{t}-\frac{n-t}{4t}+\frac{(n-t)\erate^2}{16t}\label{eq:p-tilde-t-2}
	\end{align}
	Since $\pref{p}>\tilde{p}_t$ we have
	\begin{align}
	\label{eq:p-tilde-one-forth}
		n\myp-\leftcw \pref{p}+\elfrac<n\myp-\leftcw \tilde{p}_t+\elfrac=\frac{n-\leftcw}{4}
	\end{align}
	and thus, $$\suff{p}=\frac{n\myp-\leftcw\pref{p}}{n-\leftcw}<\frac{1}{4}-\frac{\erate^2}{16}$$.
\end{proof}

\begin{remark}
	It also follows by \eqref{eq:p-tilde-one-forth} that if $\pref{p}<\tilde{p}_t$ then $\suff{p}>\frac{1}{4}-\frac{\erate^2}{16}$.
\end{remark}

\begin{claim}
	\label{claim-zero-starting}
	Let $\erate>0$ and $\myp\in\intvquar$. Let a codeword of length $n$ consist of $\chunknum$ chunks of sub-codeword where $\echunk=\echval$. Let $k_0=\left\lceil\frac{1-4\myp}{\echunk}\right\rceil$ and $t_0=k_0\chunk$. Let $\ppm=\myp-\frac{\erate^2}{16}$. Then $\forall p_{t_0}\in\left[ 0,\tilde{p}_{t_0}\right]  $, we have 
	\begin{align}
		n\ppm-t_0 p_{t_0}+\frac{(n-t_0)\erate^2}{16}\leq\frac{n-t_0}{4}
	\end{align}
\end{claim}

\begin{proof}
	Since $\frac{1-4\myp}{\echunk}\leq\left\lceil\frac{1-4\myp}{\echunk}\right\rceil<\frac{1-4\myp+\echunk}{\echunk}$, we have $(1-4\myp)n\leq t_0<(1-4\myp+\echunk)n$.
	\begin{align}
		n\ppm-t_0 p_{t_0}+\frac{(n-t_0)\erate^2}{16}&\leq n\ppm+\frac{(n-t_0)\erate^2}{16}\nonumber\\
		&=n\left(\myp-\frac{\erate^2}{16}\right)+\frac{(n-t_0)\erate^2}{16}\label{eq:zero-starting-larger-1}\\
		&=n\myp-\frac{t_0\erate^2}{16}\nonumber\\
		&\leq n\myp-\frac{(1-4\myp)n\erate^2}{16}\label{eq:zero-starting-larger-2}\\
		&=n\myp-\frac{n\echunk}{4}\label{eq:zero-starting-larger-3}\\
		&=\frac{n-(1-4\myp+\echunk)n}{4}\nonumber\\
		&<\frac{n-t_0}{4}\label{eq:zero-starting-larger-4}
	\end{align}
	where \eqref{eq:zero-starting-larger-1} follows by substituting $\ppm=\myp-\frac{\erate^2}{16}$, \eqref{eq:zero-starting-larger-2} follows by $t_0\geq(1-4\myp)n$, \eqref{eq:zero-starting-larger-3} follows by $\echunk=\echval$ and \eqref{eq:zero-starting-larger-4} follows by $t_0<(1-4\myp+\echunk)n$.
\end{proof}

\subsection{List decoding properties}
\begin{claim}
	\label{thm:b-list-decodable}
	Let $\dl>0$ and $S\ml{=}\frac{\echunk^3}{8}$. Provided that $\leftcw\left( 1-H(\phatt)\right)-\frac{n\erate}{2} \geq nR$ as in Claim~\ref{claim:1} and $t=k\chunk$, with probability over code design at least $1-2^{-\dl}$, the code $\megacode$ is list-decodable with radius $t\phatt$ and list size 

	\begin{align*}
		L=\frac{1+\frac{\dl}{t}}{1-H\left( \phatt\right) -\frac{nR}{t}-\frac{nS}{t\echunk}}
	\end{align*}
	
\end{claim}
	
\begin{proof}
	The proof follows ideas in \cite[Thm. 10.3]{guruswami2001list}, 
	\ml{and is modified slightly to correspond to stochastic codes. We stress that although the code is stochastic and each message corresponds to several codewords, we analyze the number $L$ of different {\em messages} with codewords that fall into a Hamming ball of limited radius.}
	The number of words of length $t$ in the Hamming ball with radius $r_e=t\phatt$ can be determined as
	\begin{align}
		\sum^{r_e}_{i=0}\binom{t}{i} & <2^{tH\left( \phatt\right) }
	\end{align}
	
	\ml{
	The number of potential codewords in $k=\frac{t}{\chunk}$ chunks is $\left( 2^{\chunk}\right) ^k=2^{k\chunk}=2^t$.
	We study the number of different messages corresponding to codewords that may lie in such a ball.
	Each message $m$ corresponds to $2^{tS/\theta}$ codewords of length $t$. Since the encoding of each message is independent of other messages, the probability that there exists more than $L$ messages with corresponding codewords of length $t$ all of which lie in the Hamming ball of radius $r_e$ centered at a received word is at most}

\begin{align}
		\ml{\binom{2^{nR}}{L+1}\cdot \left(2^{tS/\echunk}\right)^{L+1}}\cdot \left( \frac{2^{tH\left( \phatt\right) }}{2^t}\right) ^{(L+1)} & <2^{\left( nR+\frac{nS}{\echunk}\right) (L+1)}\left( \frac{2^{tH\left( \phatt\right) }}{2^t}\right) ^{(L+1)}\\
		& =2^{\left( nR+\frac{nS}{\echunk}\right) (L+1)}2^{-t\left( 1-H\left( \phatt\right) \right) (L+1)}\nonumber\\
		& =2^{-t\left(  1-H\left( \phatt\right)  -\frac{nR}{t}-\frac{nS}{t\echunk}\right) (L+1)}
	\end{align}

	Thus, the probability that the word of $k=\frac{t}{\chunk}$ chunks received is list-decoded to a list of size greater than $L$ is at most 
	\begin{align}
		2^{t}\cdot2^{-t\left( \left( 1-H\left( \phatt\right) \right) -\frac{nR}{t}-\frac{nS}{t\echunk}\right) (L+1)}\label{eq:thm-1-prob-large-list}
	\end{align}
	
	To quantify \eqref{eq:thm-1-prob-large-list}, we study
	\begin{align}
		t-t\left( \left( 1-H\left( \phatt\right) \right) -\frac{nR}{t}-\frac{nS}{t\echunk}\right) (L+1)<-\dl\label{eq:thm-1-prob-large-list-pow}
	\end{align}
	
	Notice that from \eqref{eq:con-1} we have
	\begin{align}
		1-H\left( \phatt\right) -\frac{nR}{t}-\frac{nS}{t\echunk} &\geq\frac{n\erate}{2t}-\frac{nS}{t\echunk}\nonumber\\
		&=\frac{n\erate}{2t}-\frac{n\echunk^2}{8t}\label{eq:thm-list-gap}\\
		&>0\nonumber
	\end{align}
	Hence, solving \eqref{eq:thm-1-prob-large-list-pow} for $L$ we have
	\begin{align}
		L > \frac{1+\frac{\dl}{t}}{1-H\left( \phatt\right) -\frac{nR}{t}-\frac{nS}{t\echunk}}-1\label{eq:thm-1-list-size}
	\end{align}
	Therefore, if $L$ satisfies \eqref{eq:thm-1-list-size} the code $\megacode$ is $L$-list decodable with probability at least $1-2^{-\dl}$.
\end{proof}

\begin{corollary}
	\label{coro:b-list-decodable}
	Let $\dl=\log\left(\ppm n^3\right) $. Then with probability over code design at least $1-\frac{1}{\ppm n}$, for any $t=k\chunk$ and $k\in\left[ \frac{1-4\myp}{\echunk},\frac{1-\echunk}{\echunk}\right] $ such that $\leftcw\left( 1-H(\phatt)\right)-\frac{n\erate}{2} \geq nR$, the code $\megacode$ is $L$-list decodable with radius $t\phatt$ and list size
	
	\begin{align}
		L=\frac{1+\frac{\log\left(\ppm n^3\right) }{t}}{1-H\left( \phatt\right) -\frac{nR}{t}-\frac{nS}{t\echunk}}=\blistord
	\end{align}
	
\end{corollary}

\begin{proof}
	By Claim~\ref{thm:b-list-decodable}, with probability $1-2^{-\log(\ppm n^3)}$ the code $\megacode$ is $L$-list decodable with list size $L$ being
	\ml{
	\begin{align*}
		\frac{1+\frac{\log\left(\ppm n^3\right) }{t}}{1-H\left( \phatt\right) -\frac{nR}{t}-\frac{nS}{t\echunk}}
	\end{align*}}
	Therefore, the probability that the code has a list size greater than $L$ is at most $2^{-\log(\ppm n^3)}=\frac{1}{\ppm n^3}$.
	
	Since $t<n$ and $t\phatt<n$, the probability that the code $\megacode$ is $L$-list decodable for any $k=\frac{t}{\chunk}$ is at least
	\begin{align*}
		1-n\cdot n\cdot\frac{1}{\ppm n^3}=1-\frac{1}{\ppm n}
	\end{align*}
	
	In addition, since $\leftcw\left( 1-H(\phatt)\right)-\frac{n\erate}{2} \geq nR$, from \eqref{eq:thm-list-gap} we have $1-H(\phatt)-\frac{nR}{t}-\frac{nS}{t\echunk}\geq\frac{n\erate}{2t}-\frac{n\echunk^2}{8t}>\frac{\erate}{4}$. For large enough $n$, we have $\frac{\log(\ppm n^3)}{t}=O\left( \frac{\log n}{n}\right)$. Thus, we obtain
	
	\begin{align*}
	L<\frac{1+O\left( \frac{\log n}{n}\right)}{\frac{\erate}{4}}=\blistord
	\end{align*}
\end{proof}

%\begin{definition}
%	\label{def:mega-code}
%	A mega sub-code $\megacode[l]$ of length $l$ is the concatenation of $l$ sub-codes.
%\end{definition}

\subsection{Utilizing the energy bounding condition}
\zx{Unless otherwise specified, for any integer $t\in\setchend$ where $\setchend=\choiceschunk$, integer $k=\frac{t}{\chunk}$ is the number of chunks of a left mega sub-code (or sub-codeword) with respect to position $t$ and $l=\chunknum-\frac{t}{\chunk}$ is the number of chunks of a right mega sub-code (\zt{or }sub-codeword) with respect to position $t$.}

\begin{definition}
	\label{def:dist}
	A right mega sub-codeword $\rmegacodw$ is of \textit{distance} $d$ from a set of right mega sub-codewords if the Hamming distance between the right mega sub-codeword $\rmegacodw$ and any right mega sub-codeword in the given set is at least $d$.
\end{definition}
\sj{In what follows we will define properties of our code with respect to a list of {\em right mega sub-codewords} $\clist$. This list consists of all the {right mega sub-codewords} corresponding to the $\mlists$ messages in $\mlist$ obtained by Bob in the list decoding phase of his decoding, {\it excluding} the true message $m$ Alice wishes to communicate to Bob, if it is indeed in the list $\mlist$ (it may not be, if $p_t > \phatt$ for the $t$ under consideration). Hence the size $\clists$ of $\clist$ is at most  $2^{nSl} \cdot \mlists$ (if the true message $m \notin \mlist$), and is at most $2^{nSl} \cdot (\mlists-1)$ (if the true message $m \in \mlist$). %We denote the list of codewords here by $\mlist$, which is the same notation used for a list of messages in the previous sections. We believe that this will not cause any confusion, and note that when we need to distinguish between lists of messages and lists of codewords, we will introduce corresponding notation.
}

\begin{definition}
	\label{def:good-1}
	A right mega sub-code $\rmegacode$ is \textit{good} with respect to a list \sj{$\clist$ of right mega sub-codewords, \zx{a message $m$}} and a sequence of $l=\chunknum-\frac{t}{\chunk}$ secrets $\left( s_{k+1},s_{k+2},\cdots,s_{\chunknum}\right) $ if the right mega sub-codeword $\rmegacodw$ is of distance more than $\frac{n-t}{2}-\frac{(n-t)\erate^2}{8}$ from the list \sj{$\clist$}. 
	%Note that the left and right mega sub-codewords above correspond to the same message $m$.
\end{definition}

\begin{definition}
	\label{def:good-2}
	A right mega sub-code $\rmegacode$ is \textit{$\sigma$-good} with respect to a list \sj{$\clist$ of right mega sub-codewords and \zx{a message $m$}} if the right mega sub-code $\rmegacode$ is good with respect to the \zx{message $m$}, the list \sj{$\clist$} and a $(1-\sigma)$ portion of sequences of $l=\chunknum-\frac{t}{\chunk}$ secrets in the set $\secr^{l}$.
\end{definition}

\begin{claim}
	\label{claim:good-1}
	%\st{Let $\clist$ be the list of all codewords corresponding to a list of messages of size $\blistord$. Let $m$ be a message such that $m\notin\mlist$.} 
	Given a sequence of $l=\chunknum-\frac{t}{\chunk}$ secrets $(s_{k+1},s_{k+2},\cdots,s_{\chunknum})\in\secr^{l}$, with probability larger than $1-2^{-\delta(n-t)}$ over code design, a right mega sub-code $\rmegacode$ is good with respect to the \sj{message $m$, the list $\clist$} and the sequence $(s_{k+1},s_{k+2},\cdots,s_{\chunknum})$ of secrets, where $\delta=\frac{\echunk^2}{4}$ and $S\ml{=}\frac{\echunk^3}{8}$.
\end{claim}

\begin{proof}
	%To prove our claim, we need to condition on the realization of the right mega sub-code $\rmegacode$ with respect to the right mega sub-codewords of the codewords in the list $\mlist$. We denote a realization of the right mega sub-codewords as an allocation $A$. Specifically, an allocation is characterized by an assignment of the right mega sub-codewords $\megacwmp$ for every $\mpm\in\mlist$ and $(s_1,s_2,\cdots,s_l)\in\secr^{l}$. In what follows we will consider the set of all allocations denoted by $\mathcal{A}$.
	
	%$\forall A\in\mathcal{A}$, d
	\sj{Let $\cwlist$ be the list of right mega sub-codewords $\clist$.
	Note that $\clists=2^{nSl} \cdot \blistord$.
}	Define the forbidden region with respect to the list \sj{$\clist$} as
	\begin{align*}
		\sj{F_{\clist}}=\bigcup_{i=1}^{L} B\left( \mathbf{x}_i,r \right) 
	\end{align*}
	where $B\left( \mathbf{x}_i,r \right)$ is the Hamming ball with center $\mathbf{x}_i$ and radius $r=\frac{n-t}{2}-\frac{(n-t)\erate^2}{8}$. 
	\zx{We depict the notion of the forbidden region in Figure~\ref{fig:forbid-1}.}
	\begin{figure}[htbp]
		\centering
		\includegraphics[scale=0.7]{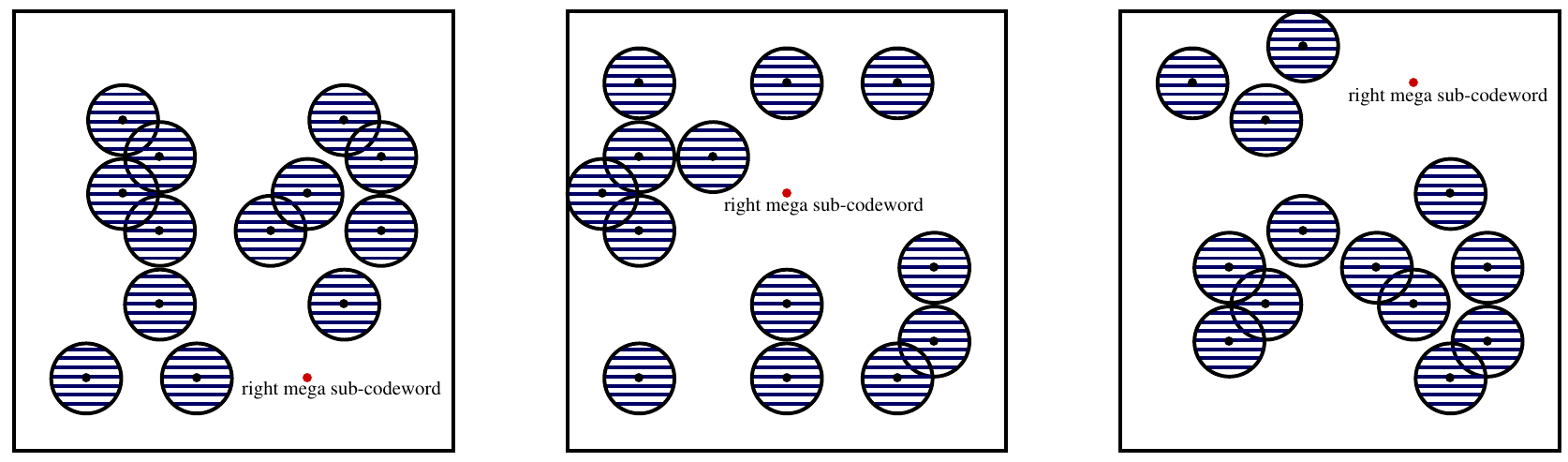}
		\caption[Three realizations of forbidden regions]{Three realizations of forbidden regions: In each realization, shaded disks correspond to the forbidden region and the isolated red point is a right mega sub-codeword outside the forbidden region.}
		\label{fig:forbid-1}
	\end{figure}
	
	Since the size of the list \sj{$\clist$} is \sj{$\clists$}, the number of \zx{words of length $n-t$ in the forbidden region $F_{\sj{\clist}}$} can be determined to be
	\sj{
	\begin{align}
		{\clists}\sum_{i=0}^{r}\binom{n-t}{i}&< {\clists} 2^{(n-t) H\left( \frac{1}{2}-\frac{\erate^2}{8}\right) }\label{eq:size-forbidden-region1}\\
		&= 2^{(n-t)\left( \frac{\log {\clists}}{n-t}+H\left(\frac{1}{2}-\frac{\echunk}{2(1-4\myp)}\right) \right) }\nonumber\\
		&< 2^{(n-t)\left( \frac{\log {\clists}}{n-t}+H\left(\frac{1}{2}-\frac{\echunk}{2}\right) \right) }\nonumber\\
		&< 2^{(n-t)\left( \frac{\log {\clists}}{n-t}+\left(1-\frac{\echunk^2}{2\ln2} \right) \right) }\label{eq:size-forbidden-region2}
	\end{align}
	}
	where \eqref{eq:size-forbidden-region1} follows by the upper bound on the volume of a Hamming ball, %\cite{rudra2007lecture}, 
	\eqref{eq:size-forbidden-region2} follows from the Taylor series of the binary entropy function in a neighborhood of $\frac{1}{2}$, i.e., $H(p)=1-\frac{1}{2\ln 2}\sum\limits_{i=1}^{\infty}\frac{(1-2p)^{2i}}{(2i-1)i}$.
	
	We assume that \ml{$\sj{\clists}=2^{nSl}\cdot \blistord$}. Since $t\in\setchend$, we have $t\leq n-\chunk$. Then by choosing $n$ to be sufficiently large and $S\ml{=}\frac{\echunk^3}{8}$, \ml{we have for some constant $c$ that
	$$
	\frac{\echunk^2}{2\ln2}-\frac{\log \sj{\clists}}{n-t}=\frac{\echunk^2}{2\ln2}-\frac{\log (c/\epsilon)}{n-t}-\frac{S}{\echunk}>\frac{\echunk^2}{4}.
	$$}
	Thus, we can choose $\delta=\frac{\echunk^2}{4}$ and obtain \ml{$\delta<\frac{\echunk^2}{2\ln2}-\frac{\log \sj{\clists}}{n-t}$.}
	%-\frac{S}{\echunk}$. 
	It follows that
	\ml{
	\begin{align}
		\frac{\log \sj{\clists}}{n-t}+\left(1-\frac{\echunk^2}{2\ln2} \right) \leq1-\delta \label{eq:size-forbidden-dl}
	\end{align}
	}
	Substituting \eqref{eq:size-forbidden-dl} into \eqref{eq:size-forbidden-region2}, we have
	\ml{
	\begin{align}
		\sj{\clists}\sum_{i=0}^{r}\binom{n-t}{i}<2^{(n-t)(1-\delta)}\label{eq:size-forbidden-region}
	\end{align}
	}
	
	Let $\rmegacodw$ be a right mega sub-codeword \ml{corresponding to $m$}. If the right mega sub-codeword $\rmegacodw$ is not in the region $F_{\sj{\clist}}$, then by Definition~\ref{def:good-1}, the code $\rmegacode$ is good with respect to \ml{$m$}, the list $\sj{\clist}$ and the secrets $(s_{k+1},s_{k+2},\cdots,s_{\chunknum})$.
	
	Therefore, the probability over $\rmegacode$ that mega sub-codeword $\rmegacodw$ does not lie in the forbidden region $F_{\sj{\clist}}$ is 
	%\begin{align*}
	%	&\mathbb{P}\left[ \megacodw\notin F_{\mlist}\right] \\ =&\sum_{A\in\mathcal{A}}\mathbb{P}\left[ A\right] \mathbb{P}\left[ \megacodw\notin F_{\mlist}|A\right] 
	%\end{align*}
	
	%From \eqref{eq:size-forbidden-region} we have
	%\begin{align*}
	%	\mathbb{P}\left[ \megacodw\notin F_{\mlist}|A\right] &>\frac{2^{n-t}-2^{(n-t)(1-\delta)}}{2^{n-t}}\\
	%	&=1-2^{-(n-t)\delta}
	%\end{align*}
	%and therefore,
	\begin{align}
		\mathbb{P}\left[ \rmegacodw\notin F_{\sj{\clist}}\right] &>\frac{2^{n-t}-2^{(n-t)(1-\delta)}}{2^{n-t}}\nonumber\\
		&=1-2^{-(n-t)\delta}
	\end{align}
\end{proof}

\begin{claim}
	\label{claim:good-2}
	%Let $\mlist$ be a list of all codewords corresponding to a set of messages of size $\blistord$. Let $m$ be a message such that $m\notin\mlist$. 
	With probability larger than $1-2^{-n^2}$ over code design, a right mega sub-code $\rmegacode$ of length $l=\chunknum-\frac{t}{\chunk}$ is $\sigma$-good with respect to \ml{$m$} and the list \ml{$\clist$}, where \sj{$\sigma=2^{-\frac{nS}{4}}$} and $S=\frac{\echunk^3}{8}$.
\end{claim}
\begin{proof}
	%Let $\mathcal{A}$ be all possible allocations of mega sub-codewords corresponding to $\mlist$. For a fixed allocation $A\in\mathcal{A}$, t
	\ml{Let $\secr=[2^{Sn}]$ be the set of integers between 0 and $2^{Sn}-1$. We start by considering a partition of the set of right mega sub-codewords corresponding to $m$ into $\secr^{l-1}$ disjoint subsets.  Specifically, we partition the set of secrets $\secr^l$ into $\secr^{l-1}$ disjoint sets. 
	%\zt{(We partition the set $\secr^l$ into $2^{nS(l-1)}$ disjoint sets because $\secr^{(l-1)}$ is a set symbol?)} 
	Each set is indexed by an element $(s_{k+2},\dots,s_{\frac{1}{\theta}})$ in $\secr^{l-1}$.
The set $\secr_\mathbf{s^*}$ corresponding to $\mathbf{s^*}=(s^*_{k+2},\dots,s^*_{\frac{1}{\theta}})$ equals:
$$
\secr_\mathbf{s^*}=\{\mathbf{s}=(a,s^*_{k+2}+a,\dots,s^*_{\frac{1}{\theta}}+a) \mid a \in [2^{Sn}]\}
$$
where addition is done modulo $2^{Sn}$.
It holds that
$$
\secr^l=\bigcup_{\mathbf{s^*} \in \secr^{l-1}}\secr_\mathbf{s^*}.
$$
Let $\mathbf{s^*} \in \secr^{l-1}$. In our analysis below we use the fact that any two $l$-tuples $\mathbf{s}=(s_{k+1},s_{k+2},\dots,s_{\frac{1}{\theta}})$ and $\mathbf{s'}=(s'_{k+1},s'_{k+2},\dots,s'_{\frac{1}{\theta}})$ in $\secr^l$ that appear in $\secr_\mathbf{s^*}$ have the property that all their coordinates differ. Namely that $s_{k+1}\ne s'_{k+1}, \dots, s_{\frac{1}{\theta}} \ne s'_{\frac{1}{\theta}}$.
	
Now consider the set of $2^{Sn}$ right mega sub-codewords $\rmegacodw$ corresponding to $l$-tuples $\mathbf{s}=(s_{k+1},s_{k+2},\dots,s_{\frac{1}{\theta}})$ from a certain set $\secr_\mathbf{s^*}$ in the partition specified above. 
Each such mega sub-codeword consists of $l$ chunks.
By our construction, the set of $l\cdot 2^{Sn}$ chunks of the right mega sub-codewords corresponding to $\mathbf{s}=(s_{k+1},s_{k+2},\dots,s_{\frac{1}{\theta}}) \in \secr_\mathbf{s^*}$ are independent and uniformly distributed.
This follows directly from our code construction and the property of  $\secr_\mathbf{s^*}$ discussed above.
%from the other chunks (as the corresponding secrets $s_i$ in each coordinate are different for different right mega sub-codewords). 
%right mega sub-codewords of the \ml{codewords corresponding to $m$} is picked independently for different $s$ values, 
Thus, for $\mathbf{s}=(s_{k+1},s_{k+2},\dots,s_{\frac{1}{\theta}})$ and $\mathbf{s'}=(s'_{k+1},s'_{k+2},\dots,s'_{\frac{1}{\theta}})$ in $\secr_\mathbf{s^*}$,} the event that a right mega \ml{sub-code $\rmegacode$}
%sub-codeword $\rmegacodw$ 
is \textit{not} good with respect to \ml{$m$}, the list $\sj{\clist}$ and the secrets $(s_{k+1},s_{k+2},\cdots,s_{\chunknum})$ is independent from the event that a right mega 
\ml{sub-code $\rmegacode$}
%sub-codeword $\mathcal{C}_{k+1}\left( m,s^{\prime}_{k+1}\right) \circ\mathcal{C}_{k+2}\left( m,s^{\prime}_{k+2}\right) \circ\cdots\circ \mathcal{C}_{\chunknum}\left( m,s^{\prime}_{\chunknum}\right) $ 
is \textit{not} good with respect to \ml{$m$}, the list \sj{$\clist$} and the secrets $(s^{\prime}_{k+1},s^{\prime}_{k+2},\cdots,s^{\prime}_{\chunknum})$.
	
	From Claim~\ref{claim:good-1}, a right mega sub-code $\rmegacode$ is \textit{not} good with respect to \ml{$m$}, the list \sj{$\clist$} and a sequence of secrets $(s_{k+1},s_{k+2},\cdots,s_{\chunknum})$ with probability less than $2^{-(n-t)\delta}$. Thus, the probability that a right mega sub-code $\rmegacode$ is \textit{not} good with respect to \ml{$m$}, the list \sj{$\clist$} and a certain $\sigma$ portion of sequences of $l$ secrets in the set \ml{$\secr_\mathbf{s^*}$} is less than 
	\ml{
	\begin{align*}
		\left( 2^{-(n-t)\delta}\right) ^{\sigma 2^{nS}}=2^{-(n-t)\delta\sigma 2^{nS}}.
	\end{align*}
	}
	
	The number of all possible $\sigma$-portion of the set \ml{$\secr_\mathbf{s^*}$} is 
	\ml{
	\begin{align*}
		\binom{2^{nS}}{\sigma 2^{nS}}<2^{2^{nS}H(\sigma)}.
	\end{align*}
	}
	\ml{We say that a right mega sub-code $\rmegacode$ is \textit{$\sigma$-good} with respect to a list $\clist$ of right mega sub-codewords, a message $m$ and a secret set $\secr_\mathbf{s^*}$ if the right mega sub-code $\rmegacode$ is good with respect to the \zx{message $m$}, the list $\clist$ and a $(1-\sigma)$ portion of sequences of secrets in the set $\secr_\mathbf{s^*}$.}
Therefore, the probability over code design that a right mega sub-code $\rmegacode$ is \textit{not} $\sigma$-good with respect to \ml{$m$, list \sj{$\clist$}, and $\secr_\mathbf{s^*}$} is
	\ml{\begin{align}
		\mathbb{P}\left[ \rmegacode\text{ is not }\sigma\text{-good w.r.t.} \ m,\sj{\clist},\secr_\mathbf{s^*}\right] & \leq 2^{-(n-t)\delta\sigma 2^{nS}}2^{2^{nS}H(\sigma)}\nonumber\\
		&\leq 2^{-\chunk \delta\sigma 2^{nS}}2^{2^{nS}H(\sigma)}\label{eq:cl-gd-2-larger-1}\\
		&= 2^{2^{nS}\left( -\chunk \delta\sigma+H(\sigma)\right) }\nonumber\\
		&< 2^{2^{nS}\left( -\chunk \delta\sigma-2\sigma\log\sigma\right) }\label{eq:cl-gd-2-larger-2}\\
		&= 2^{2^{\frac{3nS}{4}}\left( -\echunk^3+2S\right) \frac{n}{4}}\label{eq:cl-gd-2-larger-3}\\
		&= 2^{2^{\frac{3n\echunk^3}{32}}\left( -\frac{3}{4}\right) \frac{n\echunk^3}{4}}\label{eq:cl-gd-2-larger-4}\\
		&< 2^{-n^3}\label{eq:cl-gd-2-larger-5}
	\end{align}}
	where \eqref{eq:cl-gd-2-larger-1} follows by \ml{$n-t\geq\chunk $}, \eqref{eq:cl-gd-2-larger-2} follows by Lemma~\ref{lemma-1}, \eqref{eq:cl-gd-2-larger-3} follows by substituting \ml{$\sigma=2^{-\frac{nS}{4}}$} and $\delta=\frac{\echunk^2}{4}$, \eqref{eq:cl-gd-2-larger-4} follows by substituting $S=\frac{\echunk^3}{8}$, and \eqref{eq:cl-gd-2-larger-5} follows for sufficiently large $n$.
	
	\ml{Now union bounding over all sets $\secr_\mathbf{s^*}$ in the partition of $\secr^l$, we get for sufficiently large $n$ that 
	\begin{align}
		\mathbb{P}\left[ \exists \mathbf{s^*}:\ \rmegacode\text{ is not }\sigma\text{-good w.r.t.} \ m,\sj{\clist},\secr_\mathbf{s^*}\right] & 
		\leq 2^{-n^3}2^{Sn(l-1)}<2^{-n^2}
	\end{align}}

	\ml{Finally, we notice that being \textit{$\sigma$-good} with respect to a list $\clist$ of right mega sub-codewords, a message $m$ and any secret set $\secr_\mathbf{s^*}$ in the partition of $\secr^l$ implies being $\sigma$-good with respect to \ml{$m$} and list \sj{$\clist$}.}
	Hence, the probability over code design that a right mega sub-code $\rmegacode$ is $\sigma$-good with respect to \ml{$m$} and list \sj{$\clist$} is
	
	\begin{align}
		\mathbb{P}\left[ \rmegacode\text{ is }\sigma\text{-good w.r.t.}\sj{ m, {\clist}}\right]>1-2^{-n^2}.
	\end{align}

	%The probability over $\megacode$ that mega sub-code $\megacode$ is $\sigma$-good with respect to message $m$ and list $\mlist$ is given by
	%\begin{align*}
	%& \mathbb{P}\left[ \megacode\text{ is }\sigma\text{-good w.r.t. }m,\mlist\right] \\
	%=& \sum_{A\in\mathcal{A}}\mathbb{P}\left[ A\right] \mathbb{P}\left[ \megacode\text{ is }\sigma\text{-good w.r.t. }m,\mlist|A\right] 
	%\end{align*}
	%So we have 
	%\begin{align*}
	%\mathbb{P}\left[ \megacode\text{ is }\sigma\text{-good w.r.t. }m,\mlist\right]>1-2^{-n^2}
	%\end{align*}
\end{proof}

%\begin{corollary}
%	\label{coro:good-2}
	%Let $\mlist$ be a list of all codewords corresponding to a set of messages of size $\blistord$. Let $m$ be a message, and let $\exlist$ be the set of codewords in $\mlist$ that do not correspond to $m$. 
	%With probability larger than $1-2^{-n^2}$ over code design, a right mega sub-code \ml{$\rmegacode$} of length $l=\chunknum-\frac{t}{\chunk}$ is $\sigma$-good with respect to \ml{$m$} and the list \sj{$\clist$}, where \sj{$\sigma=2^{-\frac{nS}{4}}$} and $S=\frac{\echunk^3}{8}$.
%\end{corollary}

%\begin{proof}
	%\ml{This Corollary repeats Claim~\ref{claim:good-2} and can be removed from the manuscript. However, to keep the numbering of the claims consistent with the previously submitted version, we leave the Corollary in.}
	%Since \sj{no codewords in $\clist$ correspond to $m$}, it follows from Claim~\ref{claim:good-2} that the right mega sub-code $\rmegacode$ is $\sigma$-good with respect to \ml{$m$} and the list $\exlist$ with probability larger than $1-2^{-n^2}$ over code design for $\sigma=2^{-\frac{nS}{4}}$ and $S=\frac{\echunk^3}{8}$.
%\end{proof}

\begin{remark}
	\label{re:coro-good-2}
	\zx{The goodness of a right mega sub-code is what guarantees that the consistency check in the decoding process succeeds. Specifically, if a code is good with respect to a certain list and a certain \ml{message $m$}; and in addition the right mega sub-codeword has few bit-flips;  then if \ml{$m$} is in the list it will be (w.h.p.) the unique element that passes the consistency checking phase of Bob, and if it is not in the list the consistency phase of Bob will not return any message (w.h.p.).}
\end{remark}

%We denote the list of codewords excluding the codeword $\mathbf{x}$ corresponding to the transmitted message as $\exlist$.

\begin{claim}
	\label{claim:good-3}
	%\ml{Let $\mlist$ and $\exlist$ be as defined in Corolloary~\ref{coro:good-2}}.
	%Let $\sigma=2^{-\frac{nSl}{4}}$ and $S=\frac{\echunk^3}{8}$. With probability larger than $1-2^{-n}$ over code design, for every \ml{message $m$, every list $\mlist$ (or $\exlist$)}, and every chunk end $t\in\setchend$, the right mega sub-code is \textit{$\sigma$-good} with respect to \ml{$m$ and  $\mlist$ (or $\exlist$)}.
	Let \ml{$\sigma=2^{-\frac{nS}{4}}$} and $S=\frac{\echunk^3}{8}$. With probability larger than $1-2^{-n}$ over code design, for every \sj{message $m$, every list $\clist$}, and every chunk end $t\in\setchend$, the right mega sub-code is \textit{$\sigma$-good} with respect to \sj{$m$ and  $\clist$}.
\end{claim}

\begin{proof}
	%Here we should resume the Calvin's power constraint to $\ppm$ so that the energy bounding condition is satisfied for any $\ppm$. 
	The number of possible lists that can be obtained at a certain left mega chunk end $t$ \sj{depends on a set of messages of size is $\frac{c}{\epsilon}$ for some constant $c$ and is thus at most of size 
	\begin{align}
		\label{eq:list-upper-bound}
		{{2^{Rn}} \choose {\frac{c}{\epsilon}}}\leq 2^{Rcn/\epsilon}	
	\end{align}}
	
	%From Claim~\ref{claim:good-2} and Corollary~\ref{coro:good-2} we 
	From Claim~\ref{claim:good-2} we know that for \ml{$\sigma=2^{-\frac{nS}{4}}$} and $S=\frac{\echunk^3}{8}$, the probability that a right mega sub-code $\rmegacode$ is $\sigma$-good with respect to \sj{all $m$, any list $\clist$}, and every left mega chunk end $t$ is at least 
	\ml{
	\begin{align}
		1-2^{nR}\cdot 2^{Rcn/\epsilon}	\cdot \chunknum\cdot 2^{-n^2} &> 1-2^{-n^2+3cn/\epsilon}\nonumber\\
		&> 1-2^{-n}
	\end{align} 
	}
	for sufficiently large $n$.
\end{proof}

\subsection{Summary and proof Theorem~\ref{the:main:flip}}

\begin{claim}
	\label{claim:prob-good-exist}
	Let $\myp\in\intvquar$ and $\erate>0$. Let $\echunk=\echval$ and $\phatt=\frac{n^2\erate^2}{16}\cdot\frac{1}{t^2}-\frac{n(1-4\myp)}{4}\cdot\frac{1}{t}+\frac{1}{4}$. Let $\ppm=\myp-\frac{\erate^2}{16}$ be the fraction of a codeword Calvin flips. With probability at least $1-\frac{1}{\ppm n}-2^{-n}$ over code design, there exists a good code $\mcode$ such that the following properties are satisfied
	\begin{itemize}
		\item For any bit-flip pattern of the adversary, there exists a position $\tstar=\kstar\chunk$ such that the left mega sub-code with respect to position $\tstar$, $\megacode[\kstar]$, is list decodable with list size $L=\blistord$ and that the transmitted \sj{message $m$} %(with corresponding codeword $\mathbf{x}$)} 
		is in $\mlist$. %$\mlist'$ be the list of codewords corresponding to messages in $\mlist$, and 
		\sj{Let $\clist$ be the list of right mega sub-codewords corresponding to $\mlist \setminus\{m\}$.}
		\item \zx{For any position $t$ for which $t_0\leq t\leq\tstar$}, the right mega sub-word received with respect to position $t$ has no more than $\frac{1}{4}-\frac{\erate^2}{16}$ of its bits flipped; \zx{and, in addition, the right mega sub-code $\rmegacode$} is $\sigma$-good with respect to \ml{the message $m$ transmitted} and the list \sj{$\clist$} where \sj{$\sigma=2^{-\frac{nS }{4}}$} and $S=\frac{\echunk^3}{8}$.
		%, and $l =\chunknum-\frac{t}{\chunk}$.
	\end{itemize}
\end{claim}

\begin{proof}
	We consider all possible bit-flip patterns of the adversary by analyzing all of Calvin's possible trajectories. As mentioned above, all possible trajectories of Calvin can be classified into two types, the High Type trajectory and the Low Type Trajectory. 
	
	For any Low Type Trajectory, we have $p_{t_0}<\tilde{p}_{t_0}<\phat_{t_0}$. Notice that by our choice of $\phatt$, the list-decoding condition in Claim~\ref{claim:1} is always satisfied. Therefore, with list decoding radius $t_{0}\phat_{t_{0}}$, by Corollary~\ref{coro:b-list-decodable}, the left mega sub-code $\megacode[\frac{t_{0}}{\chunk}]$ is list decodable with list size $\blistord$ with probability $1-\frac{1}{\ppm n}$ over code design. In addition, since $t_{0}p_{t_0}<t_{0}\phat_{t_0}$, we have \ml{$m\in\mlist$}. So far the first property is satisfied for any Low Type Trajectory.
	
	On the other hand, for any $\ppm=\myp-\frac{\erate^2}{16}$, by Claim~\ref{claim-zero-starting} the right mega sub-word with respect to position $t_{0}$ has no more than $\frac{1}{4}-\frac{\erate^2}{16}$ of its bits flipped. Then by Claim~\ref{claim:good-3}, the right mega sub-code $\rmegacode[\frac{t_{0}}{\chunk}]$ is $\sigma$-good with respect to \ml{$m$ and list \sj{$\clist$}} with probability $1-2^{-n}$ over code design. Hence, for any Low Type Trajectory, our code design possesses the two properties stated in the claim. Moreover, in this case we have $\tstar=t_0$ (Shown in Figure~\ref{fig:decoding-case-1}).
	\begin{figure}[htb]
		\centering
		\includegraphics[scale=0.8]{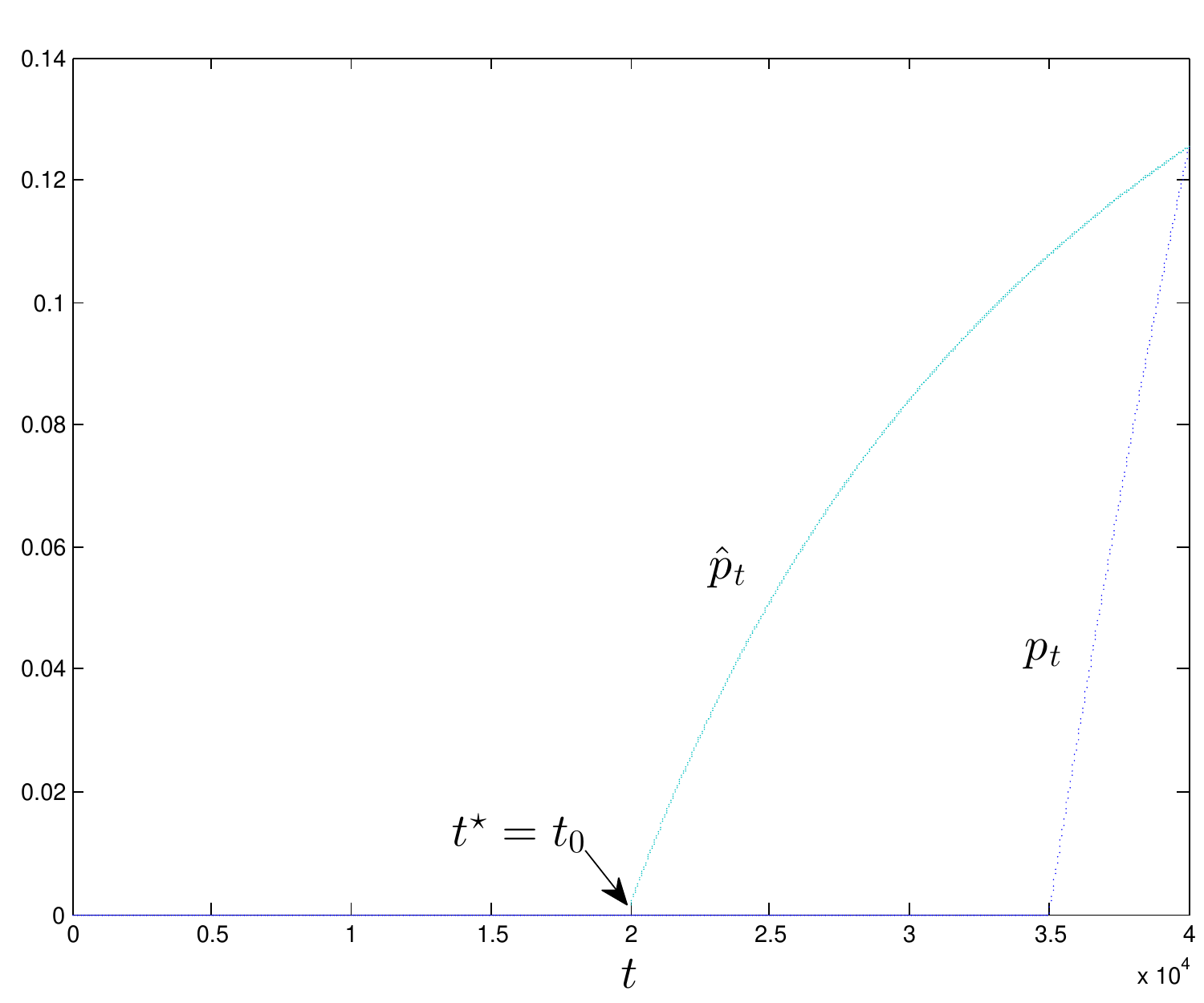}
		\caption{Case $1$: Stop decoding at $\tstar=t_{0}$ \zx{under the setting where $n=40000$, $\myp=\frac{1}{8}$, $\erate=0.08$ and $t\in[20000,40000]$.}}
		\label{fig:decoding-case-1}
	\end{figure}
	
	For any High Type Trajectory, we have $p_{t_0}>\tilde{p}_{t_0}$. By Claim~\ref{claim-intersection}, given any trajectory $\pref{p}$ of High Type, the trajectory $\pref{p}$ always intersects with $\phatt$ no later than the second to last chunk end. Let $\tstar$ be the chunk end immediately after the first intersection point (Shown in Figure~\ref{fig:decoding-case-2}). 
	Then at any position $t\leq \tstar$, by Corollary~\ref{coro:b-list-decodable} with list decoding radius $t\phat_{t}$, the left mega sub-code $\megacode[\frac{t}{\chunk}]$ is list decodable with list size $\blistord$ with probability $1-\frac{1}{\ppm n}$ over code design. Also, for $\tstar$, since $\tstar p_{\tstar}<\tstar\phat_{\tstar}$, the transmitted \ml{message $m$} is in the list $\mlist$. 
	
	From Claim~\ref{claim-ingap} and Claim~\ref{claim:p-tilde-t}, for any trajectory $\pref{p}$  of Calvin of High Type, if $t \leq \tstar$ then 
%	the trajectory $\pref{p}$ intersects with the decoding reference trajectory $\phatt$ for the first time, then at the chunk end $\tstar$ immediately after the intersection point, 
	the right mega sub-word with respect to position $t$ has no more than $\frac{1}{4}-\frac{\erate^2}{16}$ of its bits flipped. Furthermore, the right mega sub-code with respect to position $t$, $\rmegacode[\frac{t}{\chunk}]$, is $\sigma$-good with respect to \ml{$m$ and list \sj{$\clist$}} with probability $1-2^{-n}$ over code design. Thus far, for any High Type Trajectory, both the properties in the claim are also satisfied by our code design.  
	%(Shown in Figure~\ref{fig:decoding-case-2}).

	\begin{figure}[htb]
		\centering
		\includegraphics[scale=0.8]{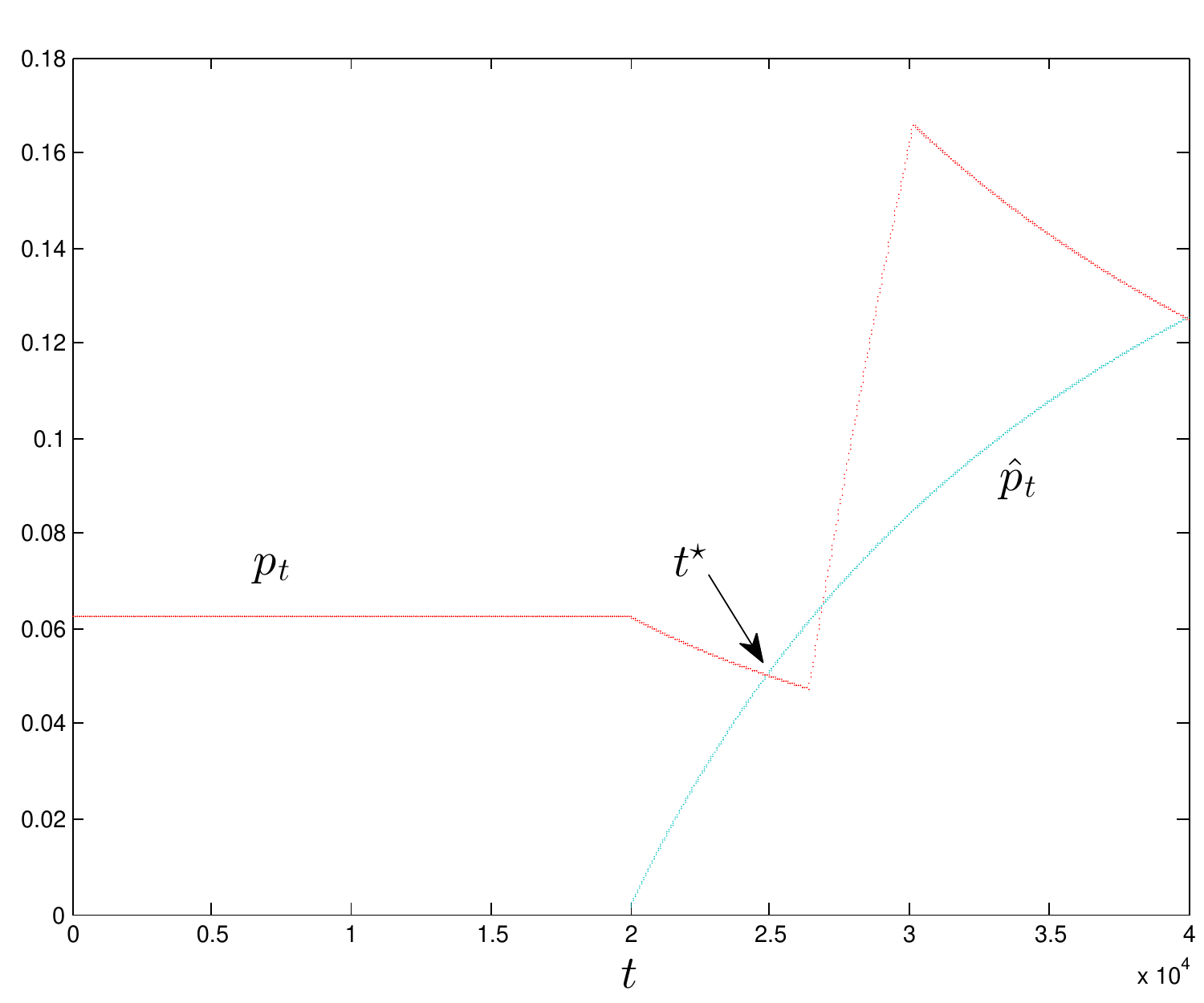}
		\caption{Case $2$: Stop decoding at $\tstar$ \zx{under the setting where $n=40000$, $\myp=\frac{1}{8}$, $\erate=0.08$ and $t\in[20000,40000]$.}}
		\label{fig:decoding-case-2}
	\end{figure}
	
	In conclusion, the probability that the code $\mcode$ possesses the two properties is at least $1-\frac{1}{\ppm n}-2^{-n}$.
\end{proof}

\begin{remark}
\label{rem:decode}
	%More importantly, in the sequential list decoding process, Claim~\ref{claim:good-3} ensures that the sequential decoding process correctly stops at  with probability $1-2^{-n}$ over code design.
	Note that, using the code from Claim~\ref{claim:prob-good-exist}, the position $\tstar$ can found by Bob through an iterative decoding process starting from the position $t_{0}$, and therefore, the coding process of Bob can stop at $\tstar$ correctly. More precisely, Claim~\ref{claim:prob-good-exist} ensures that every time Bob obtains a list of codewords, then no matter if the transmitted \ml{message $m$} is in the list \sj{$\mlist$} or not, the right mega sub-code with respect to position $t \leq \tstar$ is $\sigma$-good with respect to \ml{$m$} and the \sj{right mega sub-codeword} list \sj{$\clist$}. In other words, if $t$ is strictly smaller than $\tstar$ then the consistency decoding of Bob will not return any message, and when $t=\tstar$ the consistency decoding will return the correct message (all with high probability over the randomness of Alice). Thus, Bob can correctly determine whether to continue the decoding process or not.
\end{remark}

\begin{claim}
	\label{claim:prob-decode}
	Let $\myp\in\intvquar$ and $\erate>0$. Let $\echunk=\echval$ and $\phatt=\frac{n^2\erate^2}{16}\cdot\frac{1}{t^2}-\frac{n(1-4\myp)}{4}\cdot\frac{1}{t}+\frac{1}{4}$. Let $\ppm=\myp-\frac{\erate^2}{16}$ be the fraction of a codeword Calvin flips. Let $R=C_{p'}-\erate$ where $C_{p'}=\myopt$ and $\myfun=1-4(p'-\pbar)$. For any message $m\in\msgspace$ and its corresponding encoding $\mathbf{x}\in\binaryset{n}$ using the code established in Claim~\ref{claim:prob-good-exist} and the encoder of Section~\ref{sec:model}, the decoding procedures described in Section~\ref{sec:model} allows Bob to correctly decode the \ml{message $m$} with probability at least $1-n2^{-\frac{n\echunk^3}{32}}$ over the random secrets $s\in\secr$ available to Alice. 
\end{claim}

\begin{proof}
	A decoding error occurs if in Step (3) of the decoding procedure the consistency decoder fails to return a single message or if the decoder returns a message that is not equal to the transmitted message. For all $t$ strictly less than $\tstar$ of Claim~\ref{claim:prob-good-exist}, we have by property (2) of Claim~\ref{claim:prob-good-exist}, Remark~\ref{rem:decode}, and by the definition of Step (3) of our decoding procedure that the consistency check in the decoding process will not return any \ml{message} (with probability $1-\sigma$ over the randomness of the encoding). In addition, for $t=\tstar$, with the same probability, the consistency check of the decoding process will return the correct message.
	In both cases, the success probability is obtained by the probability that the sequence of $l$ secrets used in the right mega sub-codeword is not chosen from the particular $\sigma$ portion of $\secr^{l}$ that may cause a decoding failure. From Claim~\ref{claim:good-3}, we have \ml{$\sigma=2^{-\frac{nS}{4}}$} where \ml{$S=\frac{\theta^3}{8}$}.
	% and $l\geq 1$.
	
	Therefore, the probability of successful decoding is at least
	\ml{
	\begin{align*}
		1-n\sigma=1-n2^{-\frac{nS}{4}}=1-n2^{-\frac{n\echunk^3}{32}}.
	\end{align*}
	}
\end{proof}

\begin{theorem}[Rephrasing of Theorem~\ref{the:main:flip}]
	\label{thm:capacity}
	Let $\myp\in\intvquar$, $\erate>0$ and $\ppm=\myp-\frac{\erate^2}{16}$. The capacity of the binary causal adversary bit-flip channel $C_{p}$ is $\capopt$ where $\alfun=1-4(p-\pbar)$.%$\min_{\bar{p}\in\left[ 0,\ppm\right] }\left[ \alpha(\ppm,\pbar)\left( 1-H\left( \frac{\bar{p}}{\alpha(\ppm,\pbar)}\right) \right) \right]$, where $\alpha(\ppm,\pbar)=1-4(\ppm-\pbar)$ for $\pbar\in[0,\ppm]$.
\end{theorem}

\begin{proof}
Let $\xi>0$ and $\beta>0$.
The converse is proven in \cite{dey2013upper}. Namely, in \cite{dey2013upper} it is shown that for any code $\mcode$ with stochastic encoding of rate $R=C_{p}+\beta$, the average error probability is lower bounded by $\beta^{O\left(\frac{1}{\beta}\right)}$.

The achievability proof (our main result) follows from the claim above. Specifically, for sufficiently large $n$ it holds by Claim~\ref{claim:prob-decode} that the decoding error is bounded by above by $\xi$.
In addition, for sufficiently small $\epsilon$, by the definition of $p'$ and by the continuity of the entropy function, the code rate $R=C_{p'}-\erate$ of Claim~\ref{claim:prob-decode} is at least $C_p - \beta$.
Therefore, for sufficiently large $n$, $2^{nR}=2^{n\left(C_{p}-\beta\right)}$ distinct messages can be reliably transmitted over our channel with error probability at most $\xi$. Hence, the channel capacity of the binary causal adversarial bit-flip channel is $C_{p}$.

%		Let $\echunk=\frac{\beta^2(1-4p)}{4}=\frac{1}{n^{\frac{1}{6}}}$ for any $\ppm$ and let $\mcode$ be the code with rate $R=C_{p}-\beta$ constructed by the encoding process mentioned in Section~\ref{sec:model}. Then by Claim~\ref{claim:prob-good-exist} and Claim~\ref{claim:prob-decode} with probability at least $1-\frac{1}{pn}-2^{-n}$, the probability of decoding error over code design is at most $2^{-\frac{n\echunk^3}{32}}=2^{-\frac{n\beta^6(1-4p)^3}{2048}}\leq\xi$. %=2^{-\frac{\sqrt{n}}{32}}=\xi$.
%		Therefore, $2^{nR}=2^{n\left(C_{p}-\beta\right)}$ distinct messages are allowed to be reliably transmitted by using the code $\mcode$ and the probability of decoding error of the code $\mcode$ is at most $\xi$.
%	\end{enumerate}
%	Hence, the channel capacity of the binary causal adversarial bit-flip channel is $C_{p}$.
\end{proof}

\begin{table}[p] %htbp
\vspace{-1in}
	\renewcommand{\arraystretch}{1.5}
	\centering
	\caption{Table of Parameters (for the bit-flip case)}
	\begin{tabular}{|c|c|c|}
		\hline
		symbol & description & equality/range \\ \hline
		%$\alfun$ & $\alfun=1-4(p-\pbar)$ &  \\ \hline
		$m$ & message & $m\in\msg$ \\ \hline
		$s$ & secret & $s\in\secr$ \\ \hline
		$\mathbf{x}$ & codeword & $\mathbf{x}\in\mathcal{X}^{n}$\\ \hline
		$\msg$ & message set & $\msg=\left[2^{nR}\right]$ \\ \hline
		$\secr$ & secret set & $\secr=\left[2^{nS}\right]$\\ \hline
		$\mathcal{X}$ & input alphabet & $\binaryset{}$ \\ \hline
		$\mathcal{Y}$ & output alphabet & $\binaryset{}$ \\ \hline
		$\mathcal{T}$ & set of chunk ends & $\choiceschunk$ \\ \hline
		$\mcode$ & code & \eqref{code-design} \\ \hline
		$\Phi$ & uniform distribution of stochastic codes & \\ \hline
		$C_{p}$ & capacity & $\capopt$ \\ \hline
		$R$ & message rate & \zx{$R=C_{\myp}-\erate$} \\ \hline 
		$S$ & secret rate & \ml{$S=\frac{\theta^3}{8}$} \\ \hline
		$\erate$ & gap in the rate from $C_{p}$ & \\ \hline
		$\echunk$ & $\chunknum$ is the number of chunks in a codeword & $\echunk=\echval$\\ \hline
		$\myp$ & assumed fraction of a codeword that can be flipped & $\myp=\ppm+\frac{\erate^2}{16}$\\ \hline
		$\ppm$ & actual fraction of a codeword that can be flipped & $\ppm\in\intvquar$\\ \hline
		$\pref{p}$ & Calvin's trajectory & Definition~\ref{def:p-t} \\ \hline
		$\pbar_{t}$ & for deriving $\phatt$ & \eqref{eq:p-bar-t} \\ \hline
		$\phatt$ & decoding reference trajectory & \eqref{eq:p-hat-t} \\ \hline
		$\tilde{p}_{t}$ & energy bounding benchmarking trajectory & \eqref{eq:p-tilde-t} \\ \hline
		$n$ & block length & \\ \hline
		$k$ & number of chunks in the left mega sub-codeword & $k=\frac{t}{\chunk}$\\ \hline
		$l$ & number of chunks in the right mega sub-codeword & $l=\chunknum-\frac{t}{\chunk}$\\ \hline
		$t$ & length of left mega sub-codeword & $t\in\setchend$ \\ \hline
		$\tstar$ & correct decoding point & Definition~\ref{def:t-star} \\ \hline
		$\mlist$ & a list of \sj{messages} & \\ \hline
		$\clist$ & a list of \sj{right mega sub-codewords} & \\ 
			& \sj{excluding codewords corresponding to $m$} & \\ \hline
		$\mlists$ & list size \sj{of $\mlist$} & $\blistord$ \\ \hline
		\sj{$\clists$ }& \sj{list size of $\clist$} & \sj{$2^{nSl} \cdot \blistord$} \\ \hline
		%$\dl$ & appears in the proof of Claim~\ref{thm:b-list-decodable} & $\dl>0$ \\ \hline
		%$\delta$ & appear in Claim~\ref{claim:good-1} & $\delta=\frac{\echunk^2}{4}$ \\ \hline
		$\sigma$ & fraction of bad secret sequences in $\secr^{l}$ & \ml{$\sigma=2^{-\frac{nS}{4}}$} \\ \hline
	\end{tabular}
	\label{tab:params}
\end{table}

\newpage
\section{Code Analysis for the Erasure Channel} 
\label{app:erasure}
In this section we address the bit erasure case.
For completeness, we repeat parts of the model.

\subsection{Model}
\label{sec:model:erase}

\paragraph*{Channel Model:}
For any positive integer $i$, let $\left[ i\right] $ denote the set $\left\lbrace 1,2,\cdots,i\right\rbrace $. For a transmission duration of $n$ bits, a binary causal adversarial erasure channel can be characterized by a parameter $p\in\left[ 0,1\right] $ and a triple $\left( \mathcal{X}^n,\textsf{Adv},\mathcal{Y}^n\right) $. Here, $p$ is the fraction of bits that Calvin can erase in a codeword, $\mathcal{X}=\left\lbrace 0,1\right\rbrace $ and $\mathcal{Y}=\left\lbrace 0,1,\Lambda\right\rbrace $ is the input and output alphabet of the channel, and $\textsf{Adv}=\left\lbrace \text{Adv}^i|i\in\left[ n\right] \right\rbrace $ is a sequence of mappings such that by the end of transmitting $n$ bits, the number of erasures in the codeword is at most $np$. More precisely, each map $\text{Adv}^i:\mathcal{X}^i\times\mathcal{Y}^{i-1}\to\mathcal{Y}$ is a (randomized) function that, at the time of transmitting the $i$-th bit, maps the sequence of channel inputs up to time $i$, $\left( x_1,x_2,\cdots,x_i\right) \in\mathcal{X}^i$, together with the sequence of all previous channel outputs up to time $i-1$, $\left( y_1,y_2,\cdots,y_{i-1}\right) \in\mathcal{Y}^{i-1}$, to an output symbol $y_i\in\left\lbrace x_i,\Lambda\right\rbrace $.
The functions  $\text{Adv}^i$ must satisfy the adversarial power constraint, namely that at no point in time does the total number of bit-flips exceed $pn$.

\paragraph*{Random Code Distribution:} Follows the bit-flip case.

\paragraph*{Encoder:} Follows the bit-flip case.

\paragraph*{Decoding Process:}
The decoding process of Bob is divided into two phases. Firstly, upon receiving the entire codeword with erasures, Bob identifies the smallest integer value of $\kstar$ such that the number of erasures $\lambda$ in the first $\kstar$ chunks of the received word satisfy: $\tstar\geq Rn+\lambda+\ethree \tstar$ and $np-\lambda+\frac{(n-\tstar)\etwo}{2}\leq\frac{n-\tstar}{2}$, where $\tstar=\kstar\chunk$ is the number of bits in the first $\kstar$ chunks. In our proof to come, we show that indeed such a $\kstar$ exists.
Then Bob decodes the left mega sub-codeword $\megacodw[\kstar]$  \zx{for $\lambda$ erasures} by using a list decoder \sj{(of messages)} with list size $L$. As we will show, the list size $L$ is at most $\elistord$. At this phase, the list decoded by Bob includes $L$ \sj{messages}, one of which corresponds to the transmitted message $m$. Let $\mlist$ denote the set of $L$ messages in the decoded list.
%In addition, for any sub-code $\mathcal{C}_i$ we define the restriction of $\mathcal{C}_i$ to $\mathcal{L}$ as $\encods[i]=\allenc$.
%\mikel{In sentence above, and in remaining parts of manuscript, we need to index $\encods$ by the sub code. Maybe use $\mathcal{C}_i^{\mathcal{L}}$ or maybe $\mathcal{C}_i|_{\mathcal{L}}$ }

%We assume that Calvin knows the message $m$, and at position $\tstar$, he can also determine the set $\mlist$. Further, we assume that the adversary does not know the secret $s_t$ used in chunk $\subcodeword{t}{m}{s_t}$ when determining $\mlist$, however, he does know the secret $s_t$ when determining the erasure pattern for the chunk $\subcodeword{t}{m}{s_t}$. Under these assumptions, Calvin is actually more powerful than a causal adversary. Moreover, we even allow Calvin to be able to look ahead within a chunk of sub-codeword\footnote{This means Calvin can have full knowledge of the entire sub-codeword even when only the first few bits of the sub-codeword are transmitted.}, which also gives him more power.

%In the second phase of decoding, Bob finds the smallest integer $\tstar$ larger than $\kstar$ (from the first phase) for which the number of erased entries in the received $t^*$-th chunk is at most $\left( \frac{1}{2}-\efour\right) \chunk$.
In the second phase of decoding, Bob now considers the right mega sub-codeword $\rmegacodw[\kstar]$ and uses the natural \textit{consistency} decoder (defined below) to decode a message $\hat{m} \in \mlist$. As we will show, the right mega sub-codeword with respect to position $\tstar=\kstar\chunk$ has at most a fraction of $\left( \frac{1}{2}-\efour\right)$ of its entries erased.

\begin{definition}
	Let $k\in\left[\chunknum-1\right]$. Let $\mathbf{x}\in\mathcal{X}^{k\chunk}$ be a mega sub-codeword and $\mathbf{x}^\prime\in\binaryera{k\chunk}$ be a mega sub-codeword with erasures, where $\Lambda$ denotes erasure. The mega sub-codeword $\mathbf{x}$ is \textit{consistent} with the erased mega sub-codeword $\mathbf{x}^\prime$ if and only if $\mathbf{x}$ agrees with $\mathbf{x}^\prime$ in every unerased position.
\end{definition}

\begin{definition}
	A consistency decoder applied to a right mega sub-code $\rmegacode[\frac{t}{\chunk}]$ with respect to position $t$ is a decoder that takes the right mega sub-codeword of a received codeword $\mathbf{x}^{\prime}$ and returns a unique message $\hat{m}$ corresponding to \sj{a message} in the list $\mlist$, at least one of whose right mega sub-codewords is consistent with that of $\mathbf{x}^{\prime}$. If more than one such message exists, then a decoding error is declared.
\end{definition}

Formally, the decoder process of Bob can be described as follows:
\begin{itemize}
	\item \zx{Identify an integer $\kstar$ with $\tstar=\kstar\chunk$ such that $\tstar\geq Rn+\lambda+\ethree \tstar$ and $np-\lambda+\frac{(n-\tstar)\etwo}{2}\leq\frac{n-\tstar}{2}$}
	\item List decode the left mega sub-code $\megacode[\kstar]$ with respect to position $\tstar$ to obtain a list $\mlist$ of \sj{messages} of size $L$, \zx{with list decoding radius $\lambda$}.
	\item \zx{Verify the right mega sub-codewords of the \sj{messages} in the list $\mlist$ through a} consistency decoder to recover a message $\hat{m}$ \sj{which is} in $\mlist$ or declare a decoding error.
\end{itemize}

Bob will decode correctly if his estimate $\hat{m}$ is equal to $m$. 
\sj{As $m$ is in $\mlist$}, correct decoding will happen if the only right mega sub-codewords with respect to position $\tstar$ \ml{consistent with} that of the received codeword correspond to the message $m$. In our analysis, we will show that this indeed happens with high probability over the random secrets used by Alice for encoding the right mega sub-codeword.

If Bob's estimate $\hat{m}$ is not equal to $m$, Bob is said to make a \textit{decoding error}. The \textit{probability of error} for a message $m$ is defined as the probability over Alice's private secrets that Bob decodes incorrectly. The probability of error for the code $\mcode$ is defined as the maximum over all messages $m\in\msg$ of the probability of error for message $m$.

A rate $R$ is said to be \textit{achievable} if for every $\xi>0$, $\beta>0$ and every sufficiently large $n$ that there exists a code of block length $n$ that allows Alice to communicate $2^{n(R-\beta)}$ distinct messages with Bob with probability of error at most $\xi$. The supremum over $n$ of all achievable rates is the capacity of the channel, denoted by $C_p$.

\paragraph*{Adversarial Behavior:}
The behavior of Calvin is specified by the channel model above. Nevertheless, in our analysis we assume that Calvin has certain capabilities that may be beyond those available to a causal adversary. This can be done without loss of generality as we are studying lower bounds on the achievable rate in this work, and thus our bounds will also hold for the original (more restricted) causal adversary Calvin. 

To be more specific, in our analysis, we model the behavior of Calvin in two phases corresponding to the phases of Bob. In the first phase, Calvin behaves causally and erases certain bits of the transmitted codeword. At some point in time, no matter which bits Calvin decides to erase, we show (by Claim~\ref{claim:e-1} and Claim~\ref{thm:e-list-decodable}) that there will be an integer $\tstar$ for which two properties hold:
\begin{itemize}
	\item \zx{For any erasure pattern of the adversary,} there exists a position $\tstar$ such that the left mega sub-code with respect to $\tstar$, $\megacodw[\kstar]$ is list-decodable with list size $L=\elistord$,  and that the \sj{transmitted message $m\in\mlist$}.
	\item \zx{At position $\tstar$, the right mega sub-word received with respect to position $\tstar$ has no more than $\frac{1}{2}-\efour$ of its bits erased.}
	%	, and the right mega sub-code $\rmegacode[\kstar]$ is $\sigma$-good with respect to the left mega sub-codeword of the codeword $\mathbf{x}$ and list $\exlist=\mlist \setminus \{\mathbf{x}\}$ where $\sigma=2^{-\frac{nS\lstar}{4}}$, $S=\frac{\etwo^3}{8}$, and $\lstar=\chunknum-\frac{\tstar}{\chunk}$} 
\end{itemize}
We will assume throughout our analysis that the value $\tstar$ that Bob uses in his decoding and the list $\mlist$ of \sj{messages} obtained through Bob's list decoding process can be determined explicitly by Calvin. Moreover, we assume that Calvin knows the message $m$ all along.

The second phase of Calvin starts after $\tstar$ bits have been transmitted. At this point in time we show in Claim~\ref{claim:e-1} that no matter how Calvin behaves there will be at most $\left( \frac{1}{2}-\efour\right) \left(n-\tstar\right)$ bits that can be erased. We stress that in the beginning of Calvin's second phase, the secrets corresponding to the right mega sub-codeword are unknown to Calvin. Indeed, given the causal nature of Alice's encoding, the secrets have not even been chosen by Alice at this point in time. The fact that those secrets are hidden from Calvin imply that the secrets of the right mega sub-codeword are completely independent of the list $\mlist$ determined by Calvin. This fact is crucial to our analysis.

In the second phase, we strengthen Calvin by allowing him to choose which bits to erase in a non-causal manner. Namely, we assume that Calvin chooses his erasure pattern after seeing (or looking ahead at) all the remaining bits of the transmitted codeword. As mentioned above, no matter how these erasures are chosen, there will be at most $\left( \frac{1}{2}-\efour\right) \left(n-\tstar\right)$ erased bits. The fact that the distribution of the secrets is independent from the list $\mlist$ will allow us to show that Bob succeeds in his decoding.

\subsection{Code Analysis for the Erasure Channel}
\label{sec:code-analysis:erase}

We start by summarizing several definitions and claims. The detailed presentations of the definitions and claims are followed by the summary.
\begin{enumerate}[label=\textbf{\arabic*}.]
	\item {\bf The list decoding and energy bounding properties and the existence of a correct decoding point}
	\begin{itemize}
		\item Claim~\ref{claim:e-1}: This is a central claim which establish the existence of the correct decoding point $\tstar$ such that, no matter what the erasure pattern chosen by Calvin is, the list-decoding condition and the energy bounding condition are satisfied.
	\end{itemize}
	\item {\bf List decoding properties}
	\begin{itemize}
		\item Claim~\ref{thm:e-list-decodable}: A left mega sub-code can be list decoded to a list \ml{of message }of size $\blistord$ with high probability.
		\item Corollary~\ref{coro:e-decodable}: Every left mega sub-code can be list decoded to a list \ml{of message }of size $\blistord$ with high probability.
	\end{itemize}
	\item {\bf Utilizing the energy bounding condition}
	\begin{itemize}
		\item Definition~\ref{def:dist:erasure}: Defines the distance between a right mega sub-codeword and a list of codewords.
		\item Definition~\ref{def:e-good-1}: Defines certain {\em goodness} property of a right mega sub-code with respect to \sj{a message, a list of right mega sub-codewords (of messages excluding the transmitted message), and a sequence of secrets.}
		\item Definition~\ref{def:e-good-1}: Defines the $\sigma$-goodness property of a right mega sub-code with respect to \sj{a message and a list of right mega sub-codewords (of messages excluding the transmitted message}).
%		\ml{a message and a list of codewords (without the transmitted codeword).}
%\zx{the left mega sub-codeword of} a codeword, a list without the codeword, and most sequences of secrets.
		\item Claim~\ref{claim:e-good-1}: A right mega sub-code is good with respect to 
\sj{a message, a list of right mega sub-codewords (of messages excluding the transmitted message)}, \ml{and a sequence of secrets.}%		\ml{a message, a list of codewords (without the transmitted codeword), and a sequence of secrets.}
		%\zx{the left mega sub-codeword of }a codeword, a list without the codeword, and a sequence of secrets with high probability.
		\item Claim~\ref{claim:e-good-2}: A right mega sub-code is $\sigma$-good with respect to 
\sj{a message and a list of right mega sub-codewords (of messages excluding the transmitted message).}
%\zx{the left mega sub-codeword of }a codeword, a list without the codeword, and most sequences of secrets with high probability.
		\item Claim~\ref{claim:e-good}: A right mega sub-code is $\sigma$-good with respect to 
		\ml{every transmitted message and every list of right mega sub-codewords (of messages excluding the transmitted message).}
%\zx{the left mega sub-codeword of }every codeword and every list without the codeword.
	\end{itemize}
	\item {\bf Summary and proof of Theorem~\ref{the:main:erase}}
	\begin{itemize}
		\item Claim~\ref{claim:good-code-exist}: With high probability our code $\mcode$ possesses the needed properties.
		\item Claim~\ref{thm:e-error-prob}: With high probability Bob succeeds in decoding.
		\item Theorem~\ref{thm:e-cap}: Rephrasing of Theorem~\ref{the:main:erase} (channel capacity).
	\end{itemize}
\end{enumerate}

We now define several properties for our code $\mathcal{C}$ defined in Section~\ref{sec:model:erase} that once established will allow communication between Alice and Bob. The properties defined are governed by the discussion given in Section~\ref{sec:model:erase}.  Throughout, for a list $\mlist$ \ml{of messages, we denote the list of \sj{right mega sub-}codewords corresponding to $\mlist$ which exclude the \sj{right mega sub-}codewords corresponding to a specific message $m$ as $\exlist$.}

\subsubsection{The list decoding and energy bounding properties and the existence of correct decoding point}

\begin{claim}
	\label{claim:e-1}
	For any erasure pattern chosen by the adversary, let $R=1-2p-\eone$ for $\eone\in\left( 0,1-2p\right) $ and $p\in\rangep$. Let a codeword of length $n$ consists of $\frac{1}{\etwo}$ chunks of sub-codewords and $4\etwo=\eone$. Let $\lambda$ be the number of erasures up to position $\tstar$ where $\lambda\leq np$. Let $\setchend=\choiceschunk$ and $\tstar\in\setchend$. Then there exists $\tstar=\kstar\chunk$ such that
	\begin{align}
	\label{eq:erasure-con-1}
	Rn+\lambda+\ethree \tstar\leq \tstar
	%\tstar\left(1-\ethree\right)-\lambda\geq Rn
	\end{align}
	and that
	\begin{align}
	\label{eq:erasure-con-2}
	np-\lambda+\frac{(n-\tstar)\etwo}{2}\leq\frac{n-\tstar}{2}
	\end{align}
\end{claim}

\begin{proof}
	First we note that $(1-2p-\eone+\efive)n+\lambda\leq n-\chunk$. This holds for
	\begin{align}
	\label{eq:e5-upper}
	-\etwo^2+2\etwo-\eone\leq 2p-\frac{\lambda}{n}
	\end{align}
	where the \textit{right hand side}(RHS) is positive as $2p>\frac{\lambda}{n}$ and $\eone=4\etwo$.
	
	Next let $\tstar$ be any value in the set $\setchend=\choiceschunk$ such that 
	\begin{align}
	\label{nprimelower}
	\tstar\geq(1-2p-\eone+\efive)n+\lambda
	\end{align}
	That is to say, $\tstar$ is in the range $\left[ (1-2p-\eone+\efive)n+\lambda,n-\chunk\right] $. It follows that $\kstar$ is in the range $\left[ \frac{\left( 1-2p-\eone+\efive\right) n+\lambda}{\chunk},\frac{1-\etwo}{\etwo}\right] $. Notice that by our setting of parameters, such an $\tstar$ exists. Namely, as $4\etwo=\eone$ we have
	\begin{align*}
	(n-\chunk)-(1-2p-\eone+\efive)n-\lambda>\chunk
	\end{align*}
	
	%Now let 
	%\begin{align}
	%\label{eq:e3-val}
	%\ethree=1-\frac{\left( 1-2p-\eone+\efive\right) n+\lambda}{\tstar}<1
	%\ethree=1-\frac{n}{\tstar}\left( 1-4\etwo\right) <1
	%\end{align}
	Since $\tstar\in\left[ (1-2p-\eone+\efive)n+\lambda,n-\chunk\right] $, we have
	\begin{align*}
	Rn+\lambda+\ethree\tstar & \leq Rn+\lambda+\ethree(n-\chunk)\\
	& = (1-2p-\eone+\efive)n+\lambda\\
	& \leq \tstar
	\end{align*}
	
	Thus far, we have satisfied Condition \eqref{eq:erasure-con-1} in our claim. Next we will that $\tstar$ also satisfies Condition \eqref{eq:erasure-con-2}. Actually, it is equivalent to showing that
	\begin{align}
	\label{nprimeupper}
	\tstar\leq n-\frac{np-\lambda}{\frac{1}{2}-\efour}=\left( 1-\frac{p}{\frac{1}{2}-\efour}\right)n +\frac{\lambda}{\frac{1}{2}-\efour}
	\end{align}
	or in terms of $\kstar$ we need to show
	\begin{align}
	\label{t-upper}
	\kstar\leq\frac{1}{\etwo}-\frac{np-\lambda}{\left( \frac{1}{2}-\efour\right) \chunk}
	\end{align}
	
	Note that to satisfy Condition~\eqref{eq:erasure-con-1} we need $\tstar$ to be between $(1-2p-\eone+\efive)n+\lambda$ and $n-\chunk$ for any $\etwo>0$. Here, in addition, we also require $\tstar$ to be smaller than $\left( 1-\frac{p}{\frac{1}{2}-\efour}\right)n +\frac{\lambda}{\frac{1}{2}-\efour}$. Hence, it suffices to show that
	\begin{align}
	\label{nprimerange}
	\left( 1-\frac{p}{\frac{1}{2}-\efour}\right)n +\frac{\lambda}{\frac{1}{2}-\efour}-\left( 1-2p-\eone+\efive\right) n-\lambda\geq\chunk
	\end{align}
	
	Since the \textit{left hand side}(LHS) of \eqref{nprimerange} can be viewed as a linear function of $\lambda$, we minimize the LHS of \eqref{nprimerange} over $\lambda$ and get
	\begin{align*}
	\text{LHS of \eqref{nprimerange}}\geq\left( 1-\frac{p}{\frac{1}{2}-\efour}\right)n-\left( 1-2p-\eone+\efive\right) n
	\end{align*}
	Now for the above to be at least $\chunk$, we need
	\begin{align}
	\label{nprimerange_p}
	p\leq\frac{1}{2}\cdot\frac{\left( 1-\etwo\right) \left( \eone-2\etwo+\etwo^2\right)}{\etwo}
	\end{align}
	Notice that if $\frac{\left( 1-\etwo\right) \left( \eone-2\etwo+\etwo^2\right)}{\etwo}>1$ then for any $p\in\left[ 0,\frac{1}{2}\right) $, \eqref{nprimerange_p} always holds and so does \eqref{nprimerange}.
	
	Now for a given $\eone\leq 1-2p<1$ and $4\etwo=\eone$, we have
	\begin{align*}
	\left( 1-\etwo\right) \left( \eone-2\etwo+\etwo^2\right) & \geq \frac{1}{2}\left( \eone-2\etwo+\etwo^2\right) \\
	& = \frac{1}{2}\left( 2\etwo+\etwo^2\right) \\
	& > \etwo
	\end{align*}
\end{proof}

\begin{remark}
	Note that if the rate $R>1-2p$ then there will not exist $\tstar$ such that both conditions in Claim~\ref{claim:e-1} are simultaneously satisfied.
\end{remark}

\subsubsection{List decoding properties}

\begin{lemma}
	\label{le:prob_consistent}
	Let $m\in\msg$, $k\in\left[ \chunknum-1\right]$ and $(s_1,s_2,\cdots,s_k)\in\secr^k$.
	Let $\cpm\in\binaryera{k\chunk}$ be a codeword with $\lambda$ erasures and $\mathbf{x}\in\mathcal{X}^{k\chunk}$ be a codeword encoded by \eqref{code-design}, 
	\begin{align*}
	\mathbf{x}=\megacodw[k]
	\end{align*}
	Then the probability over $\mathcal{C}_{1},\mathcal{C}_{2},\cdots,\mathcal{C}_{k}$ that $\mathbf{x}$ is consistent with $\cpm$ is $\left( \frac{1}{2}\right) ^{k\chunk-k}$.
\end{lemma}

\begin{proof}
	Consider the $i$-th chunk of $\cpm$, denoted by $\mathbf{x}_i^\prime$. If there are no erasures imposed on $\mathbf{x}_i^\prime$, then the probability that $\mathbf{x}_i^\prime$ is consistent with the $i$-th chunk of $\mathbf{x}$, sub-codeword $\mathbf{x}_i$, is $\left( \frac{1}{2}\right) ^{\chunk}$.
	
	Further, if there are no erasures on all $k$ chucks of $\cpm$, then the probability that $\mathbf{x}$ is consistent with $\cpm$ can be determined to be
	\begin{align*}
	\left( \left( \frac{1}{2}\right) ^{\chunk}\right) ^{k}=\left( \frac{1}{2}\right) ^{k\chunk}
	\end{align*}
	since the codewords for all chunks are independent.
	
	Now suppose that there are $\lambda$ erasures in the codeword $\cpm$, then the number of codewords that are consistent with each $\cpm$ is $2^k$. Thus, the probability that $\mathbf{x}$ is consistent with $\cpm$ is 
	\begin{align}
	2^k\left( \frac{1}{2}\right) ^{k\chunk}=\left( \frac{1}{2}\right) ^{k\chunk-k}
	\end{align}
\end{proof}

\begin{claim}
	\label{thm:e-list-decodable}
	Let $\dl>0$, $4\etwo=\eone$, \ml{and $S=\frac{\theta^3}{8}$}.
	Provided that $Rn+\lambda+\etwo\tstar<\tstar$ for $\tstar=\kstar\chunk$ and $\kstar\in\left[ \chunknum-1\right] $, with probability $1-2^{-\dl}$ over code design, the code $\megacode[\kstar]$ is list-decodable for $\lambda$ erasures with list size
	
	\begin{align*}
	\frac{H\left( \frac{\lambda}{\tstar}\right) +O\left( \frac{\log{\tstar}}{\tstar}\right) +1-\frac{\lambda}{\tstar}+\frac{\dl}{\tstar}}{1-\frac{\lambda}{\tstar}-\frac{nR}{\tstar}-\sj{\frac{S}{\theta}}}
	\end{align*}
\end{claim}

\begin{proof}
	The proof follows the same technique used in \cite[Thm. 10.3]{guruswami2001list}.
	
	A code $\mathcal{C}$ is $L$-list decodable for $\lambda$ erasures if for any received word \ml{$\mathbf{x}$} with $\lambda$ bits erased, there are at most $L$ \ml{messages with (at least) one corresponding codeword consistent with $\mathbf{x}$.}
	
	The number of words of length $\tstar$  with $\lambda$ erasures is $\binom{\tstar}{k}2^{\tstar-\lambda}$. By Stirling's approximation, we have $\log{\binom{\tstar}{\lambda}}=\tstar H\left( \frac{\lambda}{\tstar}\right) +O(\log{\tstar} )$. %Let $\ppm=\frac{k}{\tstar}$, t
	Then the number of words of length $\tstar$ with $\lambda$ erasures is at most 
	\begin{align}
	\label{eq:num_word_with_erasures}
	2^{\tstar\left( H\left(\frac{\lambda}{\tstar}\right)+O\left( \frac{\log{\tstar}}{\tstar}\right) +1-\frac{\lambda}{\tstar}\right) }
	\end{align}
	
	For a received word $\mathbf{x}\in\binaryera{\kstar\chunk}$ with $\lambda$ erasures, by Lemma~\ref{le:prob_consistent}, the probability that a codeword of $\kstar=\frac{\tstar}{\chunk}$ chunks is consistent with the received word $\mathbf{x}$ is
	\begin{align}
	\label{eq:prob_consistent}
	\left( \frac{1}{2}\right) ^{\tstar-\lambda}
	\end{align}
	\ml{Each message $m$ corresponds to $2^{\kstar Sn}=2^{\tstar S/\theta}$ codewords of length $\tstar$. Since the encoding of each message is independent of other messages, the probability that there exist more than $L$ messages with each message having some corresponding codeword of length $\tstar$ that is consistent with $\mathbf{x}$ is at most}
	%Therefore, the probability that there exist some $L+1$ codewords of length $\tstar$ all of which are consistent with $\mathbf{x}$ is at most
	
	\begin{align}
	\binom{2^{Rn}}{L+1}\cdot \sj{\left(2^{\tstar S/\echunk}\right)^{L+1}}\cdot \left( 2^{-(\tstar-\lambda)}\right) ^{(L+1)} & < 2^{(Rn\sj{+\tstar S/\theta})(L+1)}2^{-(\tstar-\lambda)(L+1)}\label{eq:prob-list1}
	\end{align}

	Hence, the probability that any codeword received up to position $\tstar$ has a list of size greater than $L$ is at most 
	\begin{align}
	\label{prob-list2}
	2^{\tstar\left( H\left(\frac{\lambda}{\tstar}\right)+O\left( \frac{\log{\tstar}}{\tstar}\right) +1-\frac{\lambda}{\tstar}\right) }\cdot 2^{(Rn+\sj{\tstar S/\theta)(L+1)}}2^{-(\tstar-\lambda)(L+1)}
	\end{align}
	To quantify \eqref{prob-list2} we study
	\begin{align}
	\label{eq:power}
	\tstar\left( H\left(\frac{\lambda}{\tstar}\right)+O\left( \frac{\log{\tstar}}{\tstar}\right) +1-\frac{\lambda}{\tstar}\right) +(Rn \sj{+\tstar S/\theta})(L+1)-(\tstar-\lambda)(L+1)<-\dl
	\end{align}
	where $\dl>0$.
	Dividing \eqref{eq:power} by $\tstar$, we obtain
	
	\begin{align*}
	H\left(\frac{\lambda}{\tstar}\right)+O\left( \frac{\log{\tstar}}{\tstar}\right) +\left(1-\frac{\lambda}{\tstar}\right)+\frac{(nR \sj{+\tstar S/\theta})}{\tstar}(L+1)-\left(1-\frac{\lambda}{\tstar}\right)(L+1)<-\frac{\dl}{\tstar}
	\end{align*}
	From $Rn+\lambda+\etwo\tstar\ml{\leq}\tstar$ \ml{and $S=\frac{\theta^3}{8}$} we have $(1-\frac{\lambda}{\tstar})-\frac{nR}{\tstar}-\sj{\frac{S}{\theta}>0}$. Then we have \eqref{eq:power} if and only if
	\begin{align}
	\label{eq:l_lower}
	L>\frac{H\left(\frac{\lambda}{\tstar}\right)+O\left( \frac{\log{\tstar}}{\tstar}\right) +(1-\frac{\lambda}{\tstar})+\frac{\dl}{\tstar}}{(1-\frac{\lambda}{\tstar})-\frac{nR}{\tstar}-\sj{\frac{S}{\theta}}}-1
	\end{align}
	Therefore, if $L$ satisfies \eqref{eq:l_lower}, then the code $\megacode[\kstar]$ is $L$-list decodable with probability $1-2^{-\dl}$.
\end{proof}

\begin{corollary}
	\label{coro:e-decodable}
	Let $\dl=\log\left(pn^3\right)$, $4\etwo=\eone$, \ml{and $S=\frac{\theta^3}{8}$}. Then with probability at least $1-\frac{1}{n}$ over code design, for any $\lambda$, any $\tstar\in\choiceschunk$, and $\kstar=\frac{\tstar}{\chunk}$ such that $Rn+\lambda+\ethree\tstar\leq\tstar$, the code $\megacode[\kstar]$ is $L$-list decodable for $\lambda$ erasures with
	\begin{align}
	L=\frac{H\left( \frac{\lambda}{\tstar}\right) +O\left( \frac{\log{\tstar}}{\tstar}\right) +1-\frac{\lambda}{\tstar}+\frac{\log\left(pn^3\right)}{\tstar}}{1-\frac{\lambda}{\tstar}-\frac{nR}{\tstar}-\sj{\frac{S}{\theta}}}=\elistord
	\end{align}
\end{corollary}

\begin{proof}
	By Claim~\ref{thm:e-list-decodable}, with probability $1-2^{-\log\left(pn^3\right)}$ the code $\megacode[\kstar]$ is $L$-list decodable with list size $L$ being
	\begin{align*}
	\frac{H\left( \frac{\lambda}{\tstar}\right) +O\left( \frac{\log{\tstar}}{\tstar}\right) +1-\frac{\lambda}{\tstar}+\frac{\log\left(pn^3\right)}{\tstar}}{1-\frac{\lambda}{\tstar}-\frac{nR}{\tstar}-\sj{\frac{S}{\theta}}}
	\end{align*}
	Therefore, the probability that the code has a list with size greater than $L$ is at most $2^{-\log\left(pn^3\right)}=\frac{1}{pn^3}$.
	
	Since $\lambda\leq pn$ and $\tstar<n$, taking the union bound over $\lambda$ and $\tstar$, the probability that the code $\megacode[\kstar]$ is $L$-list decodable for $\lambda$ erasures is at least
	\begin{align*}
	1-pn\cdot n\cdot\frac{1}{pn^3}=1-\frac{1}{n}
	\end{align*}
	
	Since $Rn+\lambda+\ethree\tstar\leq\tstar$  \ml{and $S=\frac{\theta^3}{8}$}, we have $1-\frac{\lambda}{\tstar}-\frac{nR}{\tstar} \sj{ - \frac{S}{\theta}\geq\frac{\etwo}{2}}$. For large enough $n$, we have $\frac{\log\left(pn^3\right)}{\tstar}=O\left( \frac{\log n}{n}\right) <1$. In addition, we have $H\left( \frac{k}{\tstar}\right) \leq1$, $O\left( \frac{\log{\tstar}}{\tstar}\right) <1$ and $1-\frac{\lambda}{\tstar}\leq1$. Thus, we obtain
	\begin{align*}
	L< \frac{4}{1-\frac{\lambda}{\tstar}-\frac{nR}{\tstar}\sj{-\frac{S}{\theta}}}<\frac{\sj{8}}{\etwo}=\elistord
	\end{align*}
\end{proof}

By Claim~\ref{claim:e-1}, no matter what erasure pattern Calvin imposes on the first $\kstar$ chunks, the right mega sub-codeword with respect to position $\tstar=\kstar\chunk$ always has at most $\left( \frac{1}{2}-\efour\right) \left(n-\tstar\right)$ bits that are erased. 

\subsubsection{Utilizing the energy bounding condition}
Unless otherwise specified, for any integer $t\in\setchend$ where $\setchend=\choiceschunk$, integer $k=\frac{t}{\chunk}$ is the number of chunks of a left mega sub-code (or sub-codeword) with respect to position $t$ and $l=\chunknum-\frac{t}{\chunk}$ is the number of chunks of a right mega sub-code (or sub-codeword) with respect to position $t$. Recall that, for a list $\mlist$, we denote the list of right mega sub-codewords corresponding to \ml{messages in $\mlist$ excluding a certain message $m$ as $\exlist=\lecwlist$.}

\begin{definition}[Repeating Definition~\ref{def:dist}]
	\label{def:dist:erasure}
	A right mega sub-codeword $\rmegacodw$ is of \textit{distance} $d$ from a set of right mega sub-codewords if the Hamming distance between the right mega sub-codeword $\rmegacodw$ and any right mega sub-codeword in the given set is at least $d$.
\end{definition}

\begin{definition}
	\label{def:e-good-1}
	A right mega sub-code $\rmegacode$ is \textit{good} with respect to \sj{\zx{a message $m$}, a list $\clist$ of right mega sub-codewords, %\zx{a left mega sub-codeword} $\megacodw$ 
	}
	and a sequence of $l=\chunknum-\frac{t}{\chunk}$ secrets $\left( s_{k+1},s_{k+2},\cdots,s_{\chunknum}\right) $ if the right mega sub-codeword $\rmegacodw$ is of distance more than $\frac{n-t}{2}-\frac{(n-t)\etwo}{2}$ from the list \sj{$\clist$}. %Note that the left and right mega sub-codewords above correspond to the same message $m$.

\end{definition}

\begin{definition}
	\label{def:e-good-2}
	A right mega sub-code $\rmegacode$ is \textit{$\sigma$-good} with respect to a \sj{message $m$ and a list $\clist$ of right mega sub-}codewords %and \zx{a left mega sub-codeword} $\megacodw$ 
	if the right mega sub-code $\rmegacode$ is good with respect to the \sj{message $m$}, %\zx{left mega sub-codeword} $\megacodw$, 
	the list \sj{$\clist$} and a $(1-\sigma)$ portion of sequences of $l=\chunknum-\frac{t}{\chunk}$ secrets in the set $\secr^{l}$.
\end{definition}

\begin{claim}
	\label{claim:e-good-1}
	%Let $\exlist=\lecwlist$ be a list of codewords of list size $L-1$. Let $\mathbf{x}$ be a codeword of message $m$ and $\mathbf{x}\notin\exlist$. Let $\setchend=\choiceschunk$ and $t\in\setchend$.
	Given a sequence of $l=\chunknum-\frac{t}{\chunk}$ secrets $(s_{k+1},s_{k+2},\cdots,s_{\chunknum})\in\secr^{l}$, with probability larger than $1-2^{-\delta(n-t)}$ over code design, a right mega sub-code $\rmegacode$ is good with respect to the \sj{message $m$},
	%\zx{ left mega sub-codeword of the} codeword $\mathbf{x}$, 
	the list \sj{$\clist$}, and the secret sequence $(s_{k+1},s_{k+2},\cdots,s_{\chunknum})$, where $\delta=\frac{\etwo^2}{4}$ and $S\ml{=}\frac{\etwo^3}{8}$.
\end{claim}

\begin{proof}
	%To prove our claim, we need to condition on the realization of the right mega sub-code $\rmegacode$ with respect to the right mega sub-codewords of the codewords in the list $\mlist$. We denote a realization of the right mega sub-codewords as an allocation $A$. Specifically, an allocation is characterized by an assignment of the right mega sub-codewords $\megacwmp$ for every $\mpm\in\mlist$ and $(s_1,s_2,\cdots,s_l)\in\secr^{l}$. In what follows we will consider the set of all allocations denoted by $\mathcal{A}$.
	
	%$\forall A\in\mathcal{A}$, d
	Define the forbidden region with respect to the list \sj{$\clist$} as
	\begin{align*}
	F_{\sj{\clist}}=\bigcup_{i=1}^{L-1} B\left( \mathbf{x}_i,r \right) 
	\end{align*}
	where $B\left( \mathbf{x}_i,r \right)$ is the Hamming ball with center $\mathbf{x}_i$ and radius $r=\frac{n-t}{2}-\frac{(n-t)\etwo}{2}$. 
	\zx{We depict the notion of the forbidden region in Figure~\ref{fig:forbid-2}.}
	\begin{figure}[htbp]
		\centering
		\includegraphics[scale=0.7]{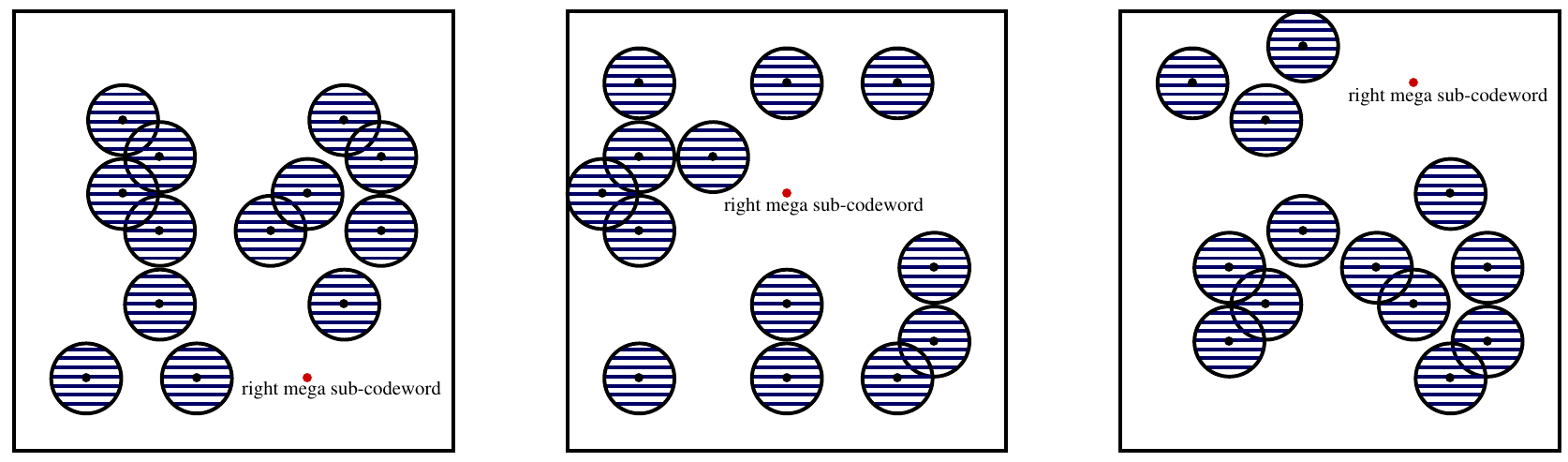}
		\caption[Three realizations of forbidden regions]{Three realizations of forbidden regions: In each realization, shaded disks correspond to the forbidden region and the isolated red point is a right mega sub-codeword outside the forbidden region.}
		\label{fig:forbid-2}
	\end{figure}
	
	Since the size of the list $\clist$ is \sj{$(L-1)\cdot 2^{(n-t)\frac{S}{\etwo}}$}, the number of codewords in $F_{\sj{\clist}}$ can be determined to be
	\begin{align}
	(L-1)\left( 2^{nS}\right) ^{l}\sum_{i=0}^{r}\binom{n-t}{i}&< L2^{(n-t)\frac{S}{\etwo}}2^{(n-t) H\left( \frac{1}{2}-\frac{\etwo}{2}\right) }\label{eq:e-size-forbidden-region1}\\
	&= 2^{(n-t)\left( \frac{\log L}{n-t}+\frac{S}{\etwo}+H\left(\frac{1}{2}-\frac{\etwo}{2}\right) \right) }\nonumber\\
	&< 2^{(n-t)\left( \frac{\log L}{n-t}+\frac{S}{\etwo}+\left(1-\frac{\etwo^2}{2\ln2} \right) \right) }\label{eq:e-size-forbidden-region2}
	\end{align}
	where \eqref{eq:e-size-forbidden-region1} follows by the upper bound on the volume of a Hamming ball, %\cite{rudra2007lecture}, 
	\eqref{eq:e-size-forbidden-region2} follows from the Taylor series of the binary entropy function in a neighborhood of $\frac{1}{2}$, i.e., $H(p)=1-\frac{1}{2\ln 2}\sum\limits_{i=1}^{\infty}\frac{(1-2p)^{2i}}{(2i-1)i}$.
	
	We assume that $L=\elistord$. Since $t\in\setchend$, we have $t\leq n-\chunk$. Then by choosing $n$ to be sufficiently large and $S\ml{=}\frac{\etwo^3}{8}$, we have $\frac{\etwo^2}{2\ln2}-\frac{\log L}{n-t}-\frac{S}{\etwo}>\frac{\etwo^2}{4}$. Thus, we can choose $\delta=\frac{\etwo^2}{4}$ and obtain $\delta<\frac{\etwo^2}{2\ln2}-\frac{\log L}{n-t}-\frac{S}{\etwo}$, and it follows that
	\begin{align}
	\frac{\log L}{n-t}+\frac{S}{\etwo}+\left(1-\frac{\etwo^2}{2\ln2} \right) \leq1-\delta \label{eq:e-size-forbidden-dl}
	\end{align}
	Substituting \eqref{eq:e-size-forbidden-dl} into \eqref{eq:e-size-forbidden-region2}, we have
	\begin{align}
	(L-1)\left( 2^{nS}\right) ^{l}\sum_{i=0}^{r}\binom{n-t}{i}<2^{(n-t)(1-\delta)}\label{eq:e-size-forbidden-region}
	\end{align}
	
	Let $\rmegacodw$ be the right mega sub-codeword of \sj{a} codeword $\mathbf{x}\notin\sj{\clist}$. If the right mega sub-codeword $\rmegacodw$ is not in the region $F_{\sj{\clist}}$, then by Definition~\ref{def:e-good-1}, the code $\rmegacode$ is good with respect to \sj{the message $m$}, %\zx{ left mega sub-codeword of the} codeword $\mathbf{x}$
	the list \sj{$\clist$} and the secrets $(s_{k+1},s_{k+2},\cdots,s_{\chunknum})$.
	
	Therefore, the probability over $\rmegacode$ that mega sub-codeword $\rmegacodw$ does not lie in the forbidden region $F_{\exlist}$ is 
	
	\begin{align*}
	\mathbb{P}\left[ \rmegacodw\notin F_{\exlist}\right] &>\frac{2^{n-t}-2^{(n-t)(1-\delta)}}{2^{n-t}}\\
	&=1-2^{-(n-t)\delta}\sj{.}
	\end{align*}
\end{proof}

\begin{claim}
	\label{claim:e-good-2}
	%Let $\exlist=\lecwlist$ be a list of codewords of list size $L-1$. Let $\mathbf{x}$ be a codeword of message $m$ and $\mathbf{x}\notin\exlist$. 
	With probability larger than $1-2^{-n^2}$ over code design, a right mega sub-code of length $l=\chunknum-\frac{t}{\chunk}$ is $\sigma$-good with respect to \sj{a message $m$
	%\zx{ the left mega sub-codeword of the} codeword $\mathbf{x}$
	 and the list $\exlist$}, where \ml{$\sigma=2^{-\frac{nS}{4}}$} and $S=\frac{\etwo^3}{8}$.
\end{claim}

\begin{proof}
\sj{We begin by defining $\secr_\mathbf{s^*}$ in exactly the same way as in the proof of Claim~\ref{claim:good-2}.}
%	Since each sub-codeword of the right mega sub-codewords of the \sj{message $m$} %codeword $\mathbf{x}$ 
%	is picked independently for different $s$ values, 
\sj{Then, by the same reasoning as in Claim~\ref{claim:good-2}, for $\mathbf{s}=(s_{k+1},s_{k+2},\dots,s_{\frac{1}{\theta}})$ and $\mathbf{s'}=(s'_{k+1},s'_{k+2},\dots,s'_{\frac{1}{\theta}})$ in $\secr_\mathbf{s^*}$,}
	the event that a right mega \ml{sub-code $\rmegacode$} 
	%sub-codeword $\rmegacodw$ 
	is \textit{not} good with respect to the \sj{message $m$},
	%\zx{ the left mega sub-codeword of the} codeword $\mathbf{x}$
	the list \sj{$\clist$} and the secrets $(s_{k+1},s_{k+2},\cdots,s_{\chunknum})$ is independent from the event that a right mega 
	\ml{sub-code $\rmegacode$}
	%sub-codeword $\mathcal{C}_{k+1}\left( m,s^{\prime}_{k+1}\right) \circ\mathcal{C}_{k+2}\left( m,s^{\prime}_{k+2}\right) \circ\cdots\circ \mathcal{C}_{\chunknum}\left( m,s^{\prime}_{\chunknum}\right) $ 
	is \textit{not} good with respect to \sj{message $m$},
	%\zx{ the left mega sub-codeword of the} codeword $\mathbf{x}$, 
	the list \sj{$\clist$} and the secrets $(s^{\prime}_{k+1},s^{\prime}_{k+2},\cdots,s^{\prime}_{\chunknum})$.
	
	From Claim~\ref{claim:e-good-1}, a right mega sub-code $\rmegacode$ is \textit{not} good with respect to \sj{$m$}, the list \sj{$\clist$} and a sequence of secrets $(s_{k+1},s_{k+2},\cdots,s_{\chunknum})$ with probability less than $2^{-(n-t)\delta}$. Thus, the probability that $\rmegacode$ is \textit{not} good with respect to \sj{$m$},
	%\zx{ the left mega sub-codeword of the} codeword $\mathbf{x}$, 
	the list \sj{$\clist$} and a certain $\sigma$ portion of sequences of $l$ secrets in the set \sj{$\secr_\mathbf{s^*}$} is less than 
	\begin{align*}
\sj{	\left( 2^{-(n-t)\delta}\right) ^{\sigma\left( 2^{nS}\right) }=2^{-(n-t)\delta\sigma 2^{nS}}.
}
	\end{align*}
	
	The number of all possible $\sigma$-portions of the set \sj{$\secr_\mathbf{s^*}$} is 
	\begin{align*}
\sj{		\binom{2^{nS}}{\sigma 2^{nS}}<2^{2^{nS}H(\sigma)}. }
	\end{align*}
%	Therefore, \ml{using the notation of Claim~\ref{claim:good-2}}, the probability over code design that $\rmegacode$ is \textit{not} $\sigma$-good with respect to \sj{$m$},
%	%\zx{ the left mega sub-codeword of the} codeword $\mathbf{x}$ 
%	list \sj{$\clist$}, and {\sj{$\secr_\mathbf{s^*}$}} is
%	\begin{align*}
%\sj{		\binom{2^{nS}}{\sigma 2^{nS}}<2^{2^{nS}H(\sigma)}}
%	\end{align*}
	\sj{We say that a right mega sub-code $\rmegacode$ is \textit{$\sigma$-good} with respect to a list $\clist$ of right mega sub-codewords, a message $m$ and a secret set $\secr_\mathbf{s^*}$ if the right mega sub-code $\rmegacode$ is good with respect to the \zx{message $m$}, the list $\clist$ and a $(1-\sigma)$ portion of sequences of secrets in the set $\secr_\mathbf{s^*}$.} Therefore, the probability over code design that $\rmegacode$ is \textit{not} $\sigma$-good with respect to
	\sj{$m$, list $\exlist$, and $\secr_\mathbf{s^*}$} is
\sj{
	\begin{align}
		\mathbb{P}\left[ \rmegacode\text{ is not }\sigma\text{-good w.r.t.  } m, \clist, \secr_\mathbf{s^*}\right] & \leq 2^{-(n-t)\delta\sigma 2^{nS}}2^{2^{nS}H(\sigma)}\nonumber\\
		&\leq 2^{-\chunk \delta\sigma 2^{nS}}2^{2^{nS}H(\sigma)}\label{eq:e-cl-gd-2-larger-1}\\
		&= 2^{2^{nS}\left( -\chunk \delta\sigma+H(\sigma)\right) }\nonumber\\
		&< 2^{2^{nS}\left( -\chunk \delta\sigma-2\sigma\log\sigma\right) }\label{eq:e-cl-gd-2-larger-2}\\
		&= 2^{2^{\frac{3nS}{4}}\left( -\etwo^3+2S\right) \frac{n}{4}}\label{eq:e-cl-gd-2-larger-3}\\
		&= 2^{2^{\frac{3n\etwo^3}{32}}\left( -\frac{3}{4}\right) \frac{n\etwo^3}{4}}\label{eq:e-cl-gd-2-larger-4}\\
		&< 2^{-n^3}\label{eq:e-cl-gd-2-larger-5}
	\end{align}}
	where \eqref{eq:e-cl-gd-2-larger-1} follows by \sj{$n-t \geq \chunk $}, \eqref{eq:e-cl-gd-2-larger-2} follows by Lemma~\ref{lemma-1}, \eqref{eq:e-cl-gd-2-larger-3} follows by substituting \sj{$\sigma=2^{-\frac{nS}{4}}$} and $\delta=\frac{\etwo^2}{4}$, \eqref{eq:e-cl-gd-2-larger-4} follows by substituting $S=\frac{\etwo^3}{8}$, and finally \eqref{eq:e-cl-gd-2-larger-5} follows for sufficiently large $n$.
	
		\sj{Now union bounding over all sets $\secr_\mathbf{s^*}$ in the partition of $\secr^l$, we get for sufficiently large $n$ that 
	\begin{align}
		\mathbb{P}\left[ \exists \mathbf{s^*}:\ \rmegacode\text{ is not }\sigma\text{-good w.r.t.} \ m,\sj{\clist},\secr_\mathbf{s^*}\right] & 
		\leq 2^{-n^3}2^{Sn(l-1)}<2^{-n^2}
	\end{align}}

	\sj{Finally, we notice that being \textit{$\sigma$-good} with respect to a list $\clist$ of right mega sub-codewords, a message $m$ and any secret set $\secr_\mathbf{s^*}$ in the partition of $\secr^l$ implies being $\sigma$-good with respect to \ml{$m$} and list \sj{$\clist$}.}
	Hence, it follows that the probability over code design that $\rmegacode$ is $\sigma$-good with respect to \sj{$m$} and list \sj{$\clist$} is
	\begin{align*}
		\mathbb{P}\left[ \rmegacode\text{ is }\sigma\text{-good w.r.t.  } \sj{m,\clist}\right]>1-2^{-n^2}.
	\end{align*}
\end{proof}

\begin{remark}
	\label{re:e-good-2}
	The goodness of a right mega sub-code is what guarantees that the consistency check in the decoding process succeeds. Specifically, if a code is good with respect to a certain \sj{list $\clist$ (corresponding to $\mlist$)} and a certain message $m$; and in addition the right mega sub-codeword has few erased bits;  then if the \sj{message} is in the \sj{message list $\mlist$} it will be (w.h.p.) the unique element that passes the consistency checking phase of Bob.
%	\zx{The goodness of a right mega sub-code is what guarantees that the consistency check in the decoding process succeeds. Specifically, if a codeword is good with respect to a certain list, and also it has few erasures in its right mega sub-codeword, then if it is in the list it will be the unique element that passes the consistency checking.}
\end{remark}

\begin{claim}
	\label{claim:e-good}
	Let \sj{$\sigma=2^{-\frac{nS}{4}}$} and $S=\frac{\etwo^3}{8}$. With probability larger than $1-2^{-n}$ over code design, for every \sj{message $m$}, every list \sj{$\exlist$}, and every chunk end $t\in\setchend$, the right mega sub-code is \textit{$\sigma$-good} with respect to \sj{$m$} and \sj{$\clist$}.
\end{claim}

\begin{proof}
	%The number of possible lists for a certain left mega chunk end $t$ is
	
%	\begin{align*}
%	\sum_{i=0}^{np}\binom{n}{i}<2^{nH(p)}<2^{n}
%	\end{align*}
%	 
	The number of possible lists that can be obtained at a certain left mega chunk end $t$ \ml{depends on a set of messages of size is $\frac{c}{\epsilon}$ for some constant $c$ and is thus at most of size 
	\begin{align}
		\label{eq:list-upper-bound-erasure}
		{{2^{Rn}} \choose {\frac{c}{\epsilon}}}\leq 2^{Rcn/\epsilon}	
	\end{align}}
	From Claim~\ref{claim:e-good-2} we know that for \ml{$\sigma=2^{-\frac{nS}{4}}$} and $S=\frac{\etwo^3}{8}$, the probability that a right mega sub-code $\rmegacode$ is $\sigma$-good with respect to \ml{all $m$}, every list $\exlist$, and every left mega chunk end $t$ is at least 
	\ml{
	\begin{align}
	1-2^{nR}\cdot 2^{Rcn/\epsilon}\cdot \chunknum\cdot 2^{-n^2} &> 1-2^{-n^2+3cn/\epsilon}\nonumber\\
	&> 1-2^{-n}
	\end{align} }
	for sufficiently large $n$.
\end{proof}

\subsubsection{Summary and proof of Theorem~\ref{the:main:erase}}

\begin{claim}
	\label{claim:good-code-exist}
	Let $p\in\rangep$, $\eone\in\left(0,1-2p\right)$ and $4\etwo=\eone$. With probability at least $1-\frac{1}{n}-2^{-n}$ over code design, there exists a good code $\mcode$ that satisfies the following properties:
	\begin{itemize}
		\item For any erasure pattern of the adversary, there exists a position $\tstar$ such that the left mega sub-code with respect to $\tstar$, $\megacode[\kstar]$, is list decodable with list size $L=\elistord$  and that the transmitted 
		\ml{message $m$} is in $\mlist$. \ml{Let $\clist$ be the list of right mega sub-codewords corresponding to $\mlist \setminus\{m\}$.}
		\item At position $\tstar$, the right mega sub-word received with respect to position $\tstar$ has no more than $\frac{1}{2}-\efour$ of its bit erased; and in addition, the right mega sub-code $\rmegacode[\kstar]$ is $\sigma$-good with respect to the \ml{message $m$} and list $\exlist$ where \ml{$\sigma=2^{-\frac{nS}{4}}$}, and \ml{$S=\frac{\theta^3}{8}$.}
		%, and $\lstar=\chunknum-\frac{\tstar}{\chunk}$} 
	\end{itemize}
\end{claim}

\begin{proof}
	For any erasure pattern chosen by the adversary, by Corollary~\ref{coro:e-decodable}, with probability at least $1-\frac{1}{n}$ over code design, the left mega sub-code $\megacode[\kstar]$ is list decodable \zx{with decoding radius $\lambda$ and }list size $L=\elistord$. Let $\lambda$ be the parameter corresponding to $\tstar$ of Claim~\ref{claim:e-1}. Since we know there are exactly $\lambda$ erasures in the left mega sub-word of the received word, we have \ml{$m\in\mlist$}. In addition, the properties of $\lambda$ satisfy the number of asserted erasures in the right mega sub-word.
	
	By Claim~\ref{claim:e-good}, with probability greater than $1-2^{-n}$ over code design, the right mega sub-code $\rmegacode[\kstar]$ is $\sigma$-good with respect to \ml{$m$} and list $\exlist$.
	
	Hence, the probability for the code $\mcode$ to simultaneously possess the two properties stated in the claim is at least $1-\frac{1}{n}-2^{-n}$.
\end{proof}

\begin{claim}
	\label{thm:e-error-prob}
	Let $p\in\rangep$, $\eone\in\left(0,1-2p\right)$ and $4\etwo=\eone$. Let $R=1-2p-\eone$. For any transmitted message $m\in\left[2^{nR}\right]$
	% and its corresponding encoding $\mathbf{x}\in\binaryset{n}$, 
the decoding procedures described in Section~\ref{sec:model:erase} allows Bob to correctly decode with probability at least $1-2^{-\frac{n\etwo^3}{32}}$ over the random secrets $s\in\secr$ available to Alice.
\end{claim}

\begin{proof}
	A decoding error is declared if the consistency decoder fails to return a single message or if the decoder returns a message that is not equal to the transmitted message. 
	At position $\tstar$, by Claim~\ref{claim:good-code-exist} and Remark~\ref{re:e-good-2}, the consistency check of the decoding process will return the correct message, with probability $1-\sigma$ over the randomness of the encoding.  
	Therefore, the success probability is obtained by the probability that the sequence of \ml{$\lstar=\chunknum-\frac{\tstar}{\chunk}$} secrets used in the right mega sub-codeword is not chosen from the particular $\sigma$ portion of $\secr^{\lstar}$ that may cause a decoding failure. From Claim~\ref{claim:e-good-2}, we have \ml{$\sigma=2^{-\frac{nS}{4}}$} where \ml{$S=\frac{\theta^3}{8}$.}
	
	Hence, the probability of successful decoding is at least 
	\ml{
	\begin{align*}
		1-\sigma=1-2^{-\frac{nS}{4}}=1-2^{-\frac{n\etwo^3}{32}}
	\end{align*}}
\end{proof}

\begin{theorem}[Rephrasing of Theorem~\ref{the:main:erase}]
	\label{thm:e-cap}
	Let $p\in\rangep$. The capacity of the binary causal adversarial erasure channel $C_p$ is $1-2p$.
\end{theorem}

\begin{proof}
	Let $\xi>0$ and $\beta>0$ be arbitrarily small.
	The converse is proven in \cite{bassily2014causal}.
	For achievability (our main result),  setting $R$ to be equal to $C_p-\beta$ and taking sufficiently large $n$, by Claim~\ref{thm:e-error-prob} the probability of decoding error is at most $\xi$. %=2^{-\frac{\sqrt{n}}{32}}=\xi$. 
		Therefore, Alice may reliably transmit $2^{n(1-2p-\beta)}$ distinct messages to Bob by using the code $\mcode$ while the probability of error of the code $\mcode$ is at most $\xi$.
Hence, the capacity of the channel is $C_p=1-2p$.
\end{proof}

\begin{table}[b] %htbp
	\renewcommand{\arraystretch}{1.4}
	\centering
	\caption{Table of Parameters (for erasure case)}
	\begin{tabular}{|c|c|c|}
		\hline
		symbol & description & equality/range \\ \hline
		$\Lambda$ & erasure & \\ \hline
		
		$m$ & message & $m\in\msg$\\ \hline
		$s$ & secret & $s\in\secr$\\ \hline
		$\mathbf{x}$ & codeword & $\mathbf{x}\in\mathcal{X}^{n}$ \\ \hline
		$\cpm$ & codeword with erasures & \\ \hline
		
		$\msg$ & message set & $\msg=\left[2^{nR}\right]$\\ \hline
		$\secr$ & secret set & $\secr=\left[2^{nS}\right]$ \\ \hline
		
		$\mathcal{X}$ & input alphabet & $\binaryset{}$ \\ \hline
		$\mathcal{Y}$ & output alphabet & $\binaryera{}$ \\ \hline
		$\mathcal{T}$ & set of chunk ends & $\choiceschunk$ \\ \hline
		$\mcode$ & code & \eqref{code-design}\\ \hline
		%$\mathcal{C}(\cdot,\cdot)$ & codeword with message and secret & \eqref{codedesign} \\ \hline
		$\Phi$ & uniform distribution of stochastic codes & \\ \hline
		$C_{p}$ & capacity & $C_{p}=1-2p$ \\ \hline
		$R$ & message rate & $R=C_{p}-\eone$ \\ \hline
		$S$	& secret rate &  $S=\frac{\etwo^3}{8}$ \\ \hline
		$\eone$ & gap in the rate from $C_{p}$ & \\ \hline
		$\etwo$ & $\chunknum$ is the number of chunks in a codeword & $\etwo=\frac{\eone}{4}$ \\ \hline
		$p$ & fraction of a codeword that can be erased & $p\in\rangep$ \\ \hline
		
		$n$	& block length &  \\ \hline
		$k$ & number of chunks in the left mega sub-codeword & $k=\frac{t}{\chunk}$\\ \hline
		$l$ & number of chunks in the right mega sub-codeword & $l=\chunknum-\frac{t}{\chunk}$\\ \hline
		$t$ & length of left mega sub-codeword & $t\in\setchend$ \\ \hline
		$\tstar$ & correct decoding point & lower \eqref{nprimelower}, upper \eqref{nprimeupper} \\ \hline
		$\lambda$ & number of erasures up to the list decoding point & \\ \hline
		
		$\mlist$ & \zt{a list of messages} & \\ \hline
		$\exlist$ & \zt{a list of right mega sub-codewords} & \\
				  & \zt{excluding codewords corresponding to $m$} & \\ \hline
		$L$ & list size\zt{ of $\mlist$} & $\elistord$ \\ \hline
		\zt{$\clists$} & \zt{list size of $\exlist$} & \zt{$2^{nSl}\cdot O\left(\frac{1}{\epsilon}\right)$} \\ \hline
		
		%$\tstar$ & special chunk & $\tstar>t$ \\ \hline
		%$\encods$ & all encodings of $\tstar$-th sub-code w.r.t $\mlist$ & \\ \hline
		
		%$H(\cdot)$ & binary entropy function & \\ \hline
		
		%$\dl$ & appears in the proof of Claim~\ref{thm:e-list-decodable} & $\dl>0$ \\ \hline
		%$\delta$ & appears in Claim~\ref{claim:e-good-1} & $\delta=\frac{\etwo^2}{4}$ \\ \hline
		$\sigma$ & fraction of bad secret sequences in $\secr^{l}$ & \ml{$\sigma=2^{-\frac{nS}{4}}$} \\ \hline
		
	\end{tabular}
	\label{tab:e-params}
\end{table}

\end{appendices}

\end{document}